\documentclass[acmsmall, nonacm]{acmart}

\usepackage{graphicx}

\title{Neuro-Relational Programs: Unifying Queries and Neural Computation over Structured Data
}
\author{Arie Soeteman}
\email{a.w.soeteman@uva.nl}
\affiliation{
  \institution{University of Amsterdam}
  \city{Amsterdam}
  \country{Netherlands}
}
\author{Balder ten Cate}
\email{b.d.tencate@uva.nl}
\affiliation{
  \institution{University of Amsterdam}
  \city{Amsterdam}
  \country{Netherlands}
}

\author{Maurice Funk}
\email{maurice.funk@uni-leipzig.de}
\affiliation{
  \institution{Leipzig University}
  \city{Leipzig}
  \country{Germany}
}
\affiliation{
\institution{Center for Scalable
Data Analytics and
Artificial Intelligence}
\city{Leipzig}
\country{Germany}
}

\author{Benny Kimelfeld}
\email{bennyk@cs.technion.ac.il}
\affiliation{
  \institution{Technion - Israel Institute of Technology}
  \city{Haifa}
  \country{Israel}
}
\affiliation{
  \institution{RelationalAI}
  \city{Berkeley}
  \country{United States}
}

\author{Carsten Lutz}
\email{clu@informatik.uni-leipzig.de}
\affiliation{
  \institution{Leipzig University}
  \city{Leipzig}
  \country{Germany}
}
\affiliation{
\institution{Center for Scalable
Data Analytics and
Artificial Intelligence}
\city{Leipzig}
\country{Germany}
}

\author{Moritz Sch{\"o}nherr}
\email{schoenherr@informatik.uni-leipzig.de}
\affiliation{
  \institution{Leipzig University}
  \city{Leipzig}
  \country{Germany}
}
\affiliation{
\institution{Center for Scalable
Data Analytics and
Artificial Intelligence}
\city{Leipzig}
\country{Germany}
}

\usepackage{graphicx} 
\usepackage{amsmath}
\usepackage{mathtools,amsthm}
\usepackage{hyperref}
\usepackage{thm-restate}
\usepackage[capitalize]{cleveref}
\usepackage{dsfont}
\usepackage{tikz}
\usepackage{dsfont}

\usepackage[dvipsnames]{xcolor}

\usepackage{enumitem}
\setlist[itemize]{leftmargin=2em,itemsep=0.5em,topsep=0.6em}
\setlist[enumerate]{leftmargin=2em,itemsep=0.5em,,topsep=0.6em}

\usepackage{color,xspace}

\theoremstyle{plain}
\newtheorem{remark}{Remark}[section]
\newtheorem{theorem}[remark]{Theorem}
\newtheorem{lemma}[theorem]{Lemma}
\newtheorem{corollary}[theorem]{Corollary}
 
\newtheorem{definition}[theorem]{Definition}
\newtheorem{example}[theorem]{Example}

\newcommand{\emb}{\textup{emb}\xspace}

\newcommand{\sur}{\textup{sur}\xspace}

\newcommand{\ord}{\operatorname{ord}}
\newcommand{\tsum}{\mathop{\textsc{sum}}\xspace}

\newcommand{\col}{\textup{col}\xspace}

\newcommand{\relu}{\textup{ReLU}\xspace}

\newcommand{\multiset}[1]{\{\!\!\{#1\}\!\!\}}
\newcommand{\sem}[1]{\llbracket#1\rrbracket}
\DeclareMathOperator{\glueclosure}{glue-closure}

\newcommand{\calA}{\mathcal{A}}

\newcommand{\calC}{\mathcal{C}}
\newcommand{\calF}{\mathcal{F}}

\newcommand{\calL}{\mathcal{L}}

\newcommand{\calN}{\mathcal{N}}

\newcommand{\calQ}{\mathcal{Q}}
\newcommand{\calS}{\mathcal{S}}

\newcommand{\FFN}{\textup{FFN}\xspace}
\newcommand{\FFNs}{\textup{FFNs}\xspace}
\newcommand{\reluFFN}{\ensuremath{\relu\tdash\FFN}\xspace}
\newcommand{\reluFFNs}{\ensuremath{\relu\tdash\FFNs}\xspace}

\newcommand{\FO}{\textup{FO}\xspace}
\newcommand{\FOC}{\textup{FO+C}\xspace}
\newcommand{\FOCrats}{\textup{FOCQ}\xspace}

\newcommand{\tdash}{\textup{-}}

\def\reals{\mathbb{R}}

\def\nats{\mathbb{N}}
\def\rats{\mathbb{Q}}

\newcommand{\TCO}{\ensuremath{\text{TC}^0}\xspace}

\newcommand{\mn}[1]{\ensuremath{\mathsf{#1}}}

\usetikzlibrary{decorations.pathmorphing, calligraphy,calc}
\pgfset{decoration/amplitude = 1.5pt}
\pgfset{decoration/segment length = 3pt}
\pgfset{decoration/post length = 3pt}

\def\angs#1{\langle#1\rangle}
\def\scs{\mathcal{S}}
\def\arity{\mathrm{arity}}

\def\adom{\mathrm{adom}}

\def\consts{\mathsf{Consts}}

\def\fbags{\mathcal{B}_{\mathsf{fin}}}
\def\set#1{\mathord{\{#1\}}}
\def\relu{\mathrm{ReLU}}

\def\from{\mathrel{\Leftarrow}}
\def\fromprod{\mathrel{\Leftarrow_{\odot}}}
\def\fromconcat{\mathrel{\Leftarrow_{\oplus}}}
\def\fromo{\mathrel{\Leftarrow_{o}}}

\def\e#1{\emph{#1}}
\def\tup#1{\mathbf{#1}}

\newcommand{\res}{\textup{\sf eval}\xspace}
\newcommand{\Hom}{\textup{\sf Hom}\xspace}

\newcommand{\bigmultiset}[1]{\big\{\!\!\big\{#1\big\}\!\!\big\}}

\begin{document}
\setcounter{page}{1}

\begin{abstract}
The conventional approach to deep learning over relational databases applies neural models, such as Graph Neural Networks (GNNs), to a graph representation of the database. Recent approaches instead operate on databases directly, associating tuples with embeddings and extending query mechanisms to jointly process embeddings and relational content. Inspired by these developments, we introduce Neuro-Relational Programs (NRPs), a declarative query language for relational databases whose facts carry numeric vector embeddings. NRPs extend Datalog-style rules with operations that combine, aggregate, and transform embeddings, thereby interleaving relational reasoning and learnable neural components within a single formalism.
This yields a general approach to neural computation over relational data: an NRP can be read both as a query with trainable components and as a neural architecture with relational structure built in.
Natural syntactic fragments of NRPs recover existing architectures and query formalisms. Zero-ary NRPs correspond to non-adaptive query algorithms; monadic NRPs generalize GNN-style message-passing and precisely capture Deep Homomorphism Networks, a connection that we extend to frontier-guarded NRPs over databases with row-ids. We characterize the expressive power of unrestricted NRPs with $\reluFFN$ transformations by $\FOCrats$, an extension of first-order logic with counting interpreted over real-weighted structures, yielding a precise connection with uniform {\TCO} over ordered databases. In particular, this resolves a recent question posed in the context of the more general \textsc{Lara} framework. Together, these results establish NRPs as a broad declarative framework for querying and neural computation over relational data.
\end{abstract}

\ccsdesc[500]{Theory of computation~Database query languages (principles)}
\ccsdesc[500]{Theory of computation~Logic and databases}
\ccsdesc[500]{Theory of computation~Machine learning theory}

\keywords{Neural-Relational Learning, Descriptive Complexity}

\setcopyright{cc}
\setcctype{by}
\acmJournal{PACMMOD}
\acmYear{2026} \acmVolume{4} \acmNumber{5 (PODS)} \acmArticle{300}
\acmMonth{11} \acmDOI{10.1145/3850135}

\maketitle

\section{Introduction}
Recent advances in graph-based deep learning have fundamentally affected the practice of applying machine learning to relational databases. The prevailing approach is to represent the database as a graph, where database records become nodes and foreign-key relationships become edges, and then deploy graph-learning architectures such as Graph Neural Networks (GNNs)~\cite{zhou2020graph,DBLP:conf/icml/FeyHHLR0YYL24,relbench,DBLP:conf/icml/WangWGWYWZ25}. To capture relational schema information, these graph architectures have been extended to \emph{heterogeneous} settings, where nodes and edges are associated with different types~\cite{DBLP:conf/aaai/ZhaoWSHSY21,DBLP:conf/www/WangJSWYCY19}. Consequently, designing a neural network over a relational database currently requires developers to work across two separate paradigms: relational query languages are used to define and export graph structures while the predictive model itself is implemented imperatively using graph-learning frameworks such as PyTorch Geometric~\cite{fey2019fastgraphrepresentationlearning}, DGL~\cite{DBLP:journals/corr/abs-1909-01315}, and Scikit-network~\cite{DBLP:journals/jmlr/BonaldLLC20}. 

We argue that this separation can be elegantly avoided. Traditional query languages naturally support the design of graph-learning architectures once the relational model is extended with \e{tuple embeddings}, where each tuple is associated with a numeric vector~\cite{DBLP:conf/cikm/LubarskyTGK23,DBLP:conf/icde/TonshoffFGK23,DBLP:conf/sigmod/CappuzzoPT20}. From this perspective, joins determine not only how tuples are combined, but also how their embeddings are composed, while projections determine how embeddings are aggregated, similarly to provenance propagation in semiring annotations~\cite{DBLP:conf/pods/GreenKT07,DBLP:journals/pacmmod/ImMNP24},
and to the recently proposed key-value algebra of the \textsc{Lara} framework~\cite{DBLP:journals/tcs/BarceloHPS22} where keys and values serve the roles of tuples and embeddings, respectively.
This yields a unified declarative pipeline in which relational and neural computation are interleaved, allowing deep-learning architectures to be expressed directly over relational data.

In this paper, we formalize this approach as a \e{Neuro-Relational Program} (NRP). An NRP operates over an \e{embedded database}, where every fact is annotated with a numeric vector,  its embedding. NRPs are sequences of (non-recursive) Datalog-style rules that derive new embedded relations from existing ones. Conceptually, an NRP defines a computation in which embeddings are propagated and manipulated via relational operations. 
There are three types of NRP rules.
\e{Conjunction rules} implement joins and projections: for every match of a rule body, they combine the embeddings of the participating facts and then aggregate all matches that yield the same head tuple. \e{Disjunction rules} implement unions, merging alternative derivations of the same tuple by aggregating embeddings. \e{Transformation rules} update the embeddings of individual tuples, for example by applying a feedforward neural network. 
\looseness=-1

The central contribution of this paper is a characterization of the expressive power of NRPs
 and of their relationship to existing logical and neural formalisms. We show that several previously studied formalisms arise naturally as syntactic restrictions of NRPs.
 In particular, we establish that:
\begin{itemize}
\item \e{Zero-ary NRPs} exactly capture \e{non-adaptive query algorithms}~\cite{Chen2025,wu2023study,ten2024homomorphism,MFCS2025} (\Cref{sec:zero-ary}).
\item \e{Monadic NRPs} exactly capture \e{deep homomorphism networks}~\cite{maehara2024deep,schonherr2026expressive}; moreover, when tuples in the input database are associated with row identifiers, as is common in database systems, this correspondence extends to \e{frontier-guarded NRPs} (\Cref{sec:monadic-and-guarded}). 
\item 
\e{NRPs with ReLU neural networks as transormation functions}
exactly capture the logic \e{FOCQ}, an
extension of first-order logic with counting and arithmetic interpreted
over real-weighted structures. 
Over classical (unweighted) relational structures that include a linear order, $\FOCrats$ coincides with $\FOC$~\cite{grohe2024descriptive} and captures the complexity class \emph{uniform \TCO} (\Cref{{sec:descriptive}}).
\end{itemize}
These results establish NRPs as a broad framework for combining relational querying and neural computation, which unifies existing approaches to deep learning over databases.

Our work is closely connected to the recently proposed \textsc{Lara} framework~\cite{DBLP:journals/tcs/BarceloHPS22}.
Embedded databases correspond to \textsc{Lara}'s \e{associative tables} where records are key-value pairs, using the tuple as the key, and the embedding as the value. NRP's conjunctive and disjunctive rules further correspond to \textsc{Lara}'s joins and unions, and NRP's transformations are captured by \textsc{Lara}'s \e{extension functions}. 
NRPs can thus be viewed as a fragment of \textsc{Lara}, and our results both extend and refine those obtained for the \textsc{Lara} language.
A key difference is that \textsc{Lara} contains functions that derive key tuples from values, for example by copying values to keys or by filtering key tuples with 
equality checks on values.
In contrast, in NRPs, database tuples affect embeddings but not the other way around. NRP transformations only alter the embeddings, matching the behavior of deep-learning models on structured data. 
Since we study syntactic restrictions and transformation function classes that have not been considered for the \textsc{Lara} language, the logical characterizations in this work do not follow from known correspondences between \textsc{Lara} and \FO with aggregation, nor vice versa. 
Finally, restricting NRPs to neural network transformations 
allows us to associate NRPs with the complexity class \TCO(\Cref{sec:descriptive}). This positively answers an open question posed in the \textsc{Lara} article~\cite[Section~7]{DBLP:journals/tcs/BarceloHPS22}, namely whether the addition of arithmetic operations and a linear order yields a query language that exactly 
characterizes a parallel complexity class.

\section{Related Work}
We now highlight relevant literature and how it relates to this work.
\paragraph{Annotated databases}
As mentioned above, the ideas of augmenting tuples with numerical information, and maintaining the augmented tuples through queries, are established in databases through the notion of \e{provenance}~\cite{DBLP:conf/pods/GreenKT07,DBLP:journals/pacmmod/ImMNP24} (also known as \e{lineage} and \e{annotation}). In the \e{provenance semiring} framework, every tuple is annotated by an element of a commutative semiring; the multiplication operation of the semiring is used for joining tuples, and the addition operation of the semiring is used for aggregating tuples on projection (and union)~\cite{DBLP:conf/pods/GreenKT07}. Special attention has been given to the case where the annotation is numerical, as in this work, but specifically focusing on the case where this annotation is a  probability~\cite{DBLP:conf/vldb/DalviS04}. Recent years have seen research advances in the context of Datalog, on both provenance~\cite{DBLP:journals/pacmmod/ImMNP24,DBLP:journals/pacmmod/ZhaoDKRT24,DBLP:journals/vldb/DeutchGM18} and probability annotation~\cite{DBLP:conf/pods/GroheKKL20,DBLP:journals/tods/BaranyCKOV17,DBLP:conf/pods/AlvianoLMP23}, with ties to probabilistic logic programming~\cite{DBLP:conf/ijcai/RaedtKT07}.
Unlike semiring provenance, the transformations and aggregations considered here are not restricted to fixed algebraic operations and are not required to satisfy semiring axioms. Moreover, our annotations are numerical vectors rather than scalar values.

\paragraph{Embedding database records}
Considerable work has been devoted to tuple embeddings in a Euclidean space~\cite{DBLP:conf/icde/TonshoffFGK23,DBLP:conf/sigmod/CappuzzoPT20,DBLP:conf/sigmod/BordawekarS17,DBLP:conf/cikm/LubarskyTGK23,DBLP:conf/sigmod/MudgalLRDPKDAR18,DBLP:journals/ml/CvetkovIlievAV23,DBLP:conf/icml/KimGV24}, where the goal is very different from the focus of our study: compute vector annotations such that semantic proximity is translated into geometrical closeness. This task is typically translated into an optimization problem of finding embeddings that allow for precise reconstruction (decoding) of semantic relationships (kernels). 
Our use of embedded tuples also differs from the intention of \e{vector database} approaches~\cite{DBLP:journals/vldb/PanWL24}; we treat embeddings as a complementary representation of the tuples---not as raw data for efficient retrieval. Unlike RAG-based systems~\cite{DBLP:conf/nips/LewisPPPKGKLYR020}, our framework uses vector representations to compose neural networks rather than deriving similarity and context search. As explained in the introduction, our work is further closely connected to the \textsc{Lara} language that expresses relational and linear algebra operations on databases with tuple embeddings~\cite{DBLP:journals/tcs/BarceloHPS22}.

\paragraph{Probabilistic logic}
Our neuro-relational model is inspired by probabilistic-logic systems such as Scallop~\cite{DBLP:journals/pacmpl/LiHN23} and DeepProbLog~\cite{DBLP:conf/nips/ManhaeveDKDR18}, where every tuple is associated with a \e{probability} and inference rules use a neural-network formalism to state how the probability propagates to inferred tuples. In contrast, our annotation should be viewed as an embedding in a multi-dimensional Euclidean space, detached from any probabilistic interpretation.

\paragraph{Query-language views of graph learning.}
A closely related line of work studies GNN-style computation through
database- and logic-inspired languages over graphs~\cite{DBLP:conf/iclr/GeertsR22,DBLP:conf/pods/Geerts23,10.1145/3801914}. These works provide languages for reasoning about neural
computation over graph-structured inputs. For example, the \emph{Graph Embedding Language} (GEL)~\cite{DBLP:conf/pods/Geerts23} captures precisely the separating power of GNNs using a syntax of local aggregation operators and functions acting on embeddings.
NRPs differ in taking relational
databases, rather than graphs, as the primary abstraction: facts are
annotated with embeddings, and relational operations such as joins,
projections, and unions act simultaneously on relational facts and their
embeddings. A further distinction is that these works primarily focus on
non-uniform expressive-power analyses, whereas our results
are uniform in nature.

\paragraph{Neuro-Symbolic Artificial Intelligence (NSAI)} 
Our approach can also be viewed as an instantiation of NSAI, which encompasses any computational paradigm that combines neural networks (``connectionist'' AI) with logical reasoning (``symbolic'' AI, involving derivation by rules). We refer to, e.g.,~\cite{zhang2024neurosymbolicaiexplainabilitychallenges,10.1007/s00521-024-09960-z,DBLP:series/faia/342,DBLP:journals/nn/YuYLWP23} for recent surveys on NSAI. This can be a formalism for deep learning enhanced by (soft) logical rules, towards improved expressiveness and reduced effort in data labeling~\cite{DBLP:conf/cvpr/0004YG18,DBLP:conf/aaai/SilvaG21}, all the way to logical rules that incorporate neural-network inference in addition to logical inference~\cite{DBLP:conf/nips/ShahZSVYC20,DBLP:conf/iclr/ZhangCYRLQS20,domingos2025tensorlogic}. Arguably, models of this type have the potential of being easy to program and considerably more explainable than pure deep-learning architectures. 
\looseness=-1

\section{Neuro-Relational Programs}\label{sec:nrp}
In this section, we formally define the \e{Neuro-Relational Program} (NRP) framework, which unifies traditional database queries and deep neural networks. The underlying data model is an \e{annotated relational database}, where each fact is annotated with a numerical vector. Conceptually, the annotation can be viewed as an \e{embedding} of the fact in a Euclidean space. 
We first adapt the concept of a database \e{schema}, and then that of its database \e{instances}.
\looseness=-1

An \emph{extended arity} is an expression $k\angs{d}$ for natural numbers $k,d\geq 0$, where $k$ is a relation's \emph{content arity}, and $d$ its \emph{embedding dimension}. 
A \emph{schema} $\scs$ is a finite set of expressions $R[k]\angs{d}$, where each $R$ is a unique \emph{relation symbol} with extended arity $\arity(R) = k\angs{d}$.
Let $\consts$ be an infinite set of constants representing database elements. An \emph{embedded fact} (\emph{e-fact} for short) over $\scs$ is an expression of the form $R(\tup c)\angs{\tup e}$ where $R$ is a relation symbol in $\scs$ with  $\arity(R)=k\angs{d}$, $\tup c\in\consts^k$ is the \e{content tuple}, and $\tup e\in\reals^d$ is the embedding.

An \emph{embedded database} (\e{e-database} for short) over $\scs$ is a finite set $D$ of e-facts over $\scs$ such that no two e-facts share the same relation symbol and content; that is, if  $R(\tup c)\angs{\tup e}$ and  $R(\tup c)\angs{\tup e'}$ belong to $D$, then
$\tup e=\tup e'$.
We write $R(\tup a) \in D$ to denote that $R(\tup a)\angs{\tup e} \in D$ for some $\tup e \in \reals^d$. The \e{active domain} of an e-database $D$, denoted $\adom(D)$, is the set of all constants that occur in facts in the e-database $D$.
By an \e{e-relation} we refer to the subset of $D$ that consists of all e-facts $R(\tup c)\angs{\tup e}$ with the same $R$ for some relation symbol $R$.  An e-relation is \e{classical} if its embedding dimension is zero (i.e., the e-facts consist of only content, as in ordinary databases). An e-database is classical if all its e-facts have embedding dimension zero.

\paragraph{NRP syntax}
A \emph{rule} $\Psi$ over schema $\scs$ maps e-relations over the relation symbols of $\scs$ to a new e-relation over a relation symbol $R$ with extended arity $k\angs{d}$ outside of $\scs$. We say that $\Psi$ \e{produces} $R[k]\angs{d}$ from $\scs$. An NRP $\Pi$ is a sequence of rules $(\Psi_1, \dots, \Psi_n)$, where each $\Psi_i$ produces 
$R_i[k_i]\angs{d_i}$ from  schema $\scs \cup \{R_j[k_j]\angs{d_j} \,\mid\, 1 \leq j < i\}$.
We thus only consider non-recursive programs, where each rule produces an e-relation over a fresh relation symbol.
Let $\fbags(\reals^{d})$ denote the set of all finite multisets of $d$-dimensional real-valued vectors. 
Let $\arity(R)={k}\angs{d}$ and $\arity(R_i)={k_i}\angs{d_i}$. 
An NRP uses three types of rules:

\begin{enumerate}
    \item
          A \emph{conjunction rule} has the form
          $$R(\tup x)\angs{\alpha}\fromo R_1(\tup x_1)\land\dots\land R_\ell(\tup x_\ell)$$
          where every variable in $\tup x$ appears in at least one $\tup x_i$, 
$\alpha:\fbags(\reals^d)\rightarrow\reals^d$ is an \e{aggregation function} and 
$o:\reals^{d_1} \times \dots \times \reals^{d_\ell} \to \reals^d$ is a \e{combination function}. 

    \item
          A \emph{disjunction rule} has the form
          $$R(\tup x)\angs{\alpha}\from R_1(\tup x)\lor\dots\lor R_\ell(\tup x)$$
          where $d_1=\cdots=d_\ell=d$ and $\alpha:\fbags(\reals^{d})\rightarrow \reals^{d}$ is an aggregation function.

    \item
          A \emph{transformation rule} has the form
        \[
        R(\tup x)\angs{\mu} \from R_0(\tup x)
        \]
      where $\mu: \reals^{d_0} \to \reals^d$ is a \e{transformation function}.
\end{enumerate}
In all cases, the expressions left and right of ``$\from$'' are the rule's \e{head} and \e{body}, respectively.
Throughout this paper we fix element-wise sum (denoted~$\tsum$) as the aggregation function, where summing over an empty multiset yields the zero vector $\tup 0^{d}$. We fix element-wise product (denoted~$\odot$) as the combination function, 
adjusted for variable-length inputs. To be precise, in conjunction rules of the form (1) above, we require that $d=\max\{d_1, \ldots, d_\ell\}$, and we use the combination function $\odot:\reals^{d_1}\times\cdots \times \reals^{d_\ell}\to\reals^d$  where $\odot(\tup{e}_1,\ldots,\tup{e}_\ell)$ is the tuple $\tup{e}\in\reals^d$ whose $j$-th element (for $j\leq d$) is the product of the
$j$-th elements of all input tuples $\tup{e}_i$ of 
length at least $j$. 
Following conventional practice in the expressiveness analysis of Graph Neural Networks~\cite{grohe2024descriptive},
we typically use $\relu$-$\FFNs$ (feedforward neural networks with $\relu$ activations) as transformation functions, although many of our results apply to more general function classes. 
We assume that classes of transformations are closed under composition.
Following standard Datalog terminology, we call the initial relations of $\scs$ \e{EDBs} (the \e{extensional database}), and the relations produced by $\Pi$ \e{IDBs} (the \e{intensional database}).

\paragraph{NRP semantics} 
Let $D$ be an e-database. A \emph{homomorphism} from a conjunction $R_1(\tup x_1)\land\dots\land R_\ell(\tup x_\ell)$ to $D$ is a function $h$ that maps every variable in $(\tup x_1, \dots, \tup x_\ell)$ to a constant, so that for all $i = 1, \dots, \ell$ there is an e-fact $R_i(h(\tup x_i))\angs{\tup e}$ in $D$ for some embedding $\tup e$, which we denote by $\tup e_{i,h}$. Here, $h(\tup x_i)$ denotes the tuple obtained by replacing every variable $x \in \tup x_i$ with $h(x)$. A rule $\Psi$ producing $R[k]\angs{d}$ derives from $D$ a set $\Psi(D)$ of e-facts as follows:
\begin{enumerate}
    \item
          If $\Psi$ is a conjunction rule $$R(\tup x)\angs{\alpha}\fromo R_1(\tup x_1)\land\dots\land R_\ell(\tup x_\ell),$$
          let $H$ be the set of homomorphisms from $R_1(\tup x_1)\land\dots\land R_\ell(\tup x_\ell)$ to $D$. Then $R(\tup a) \in \Psi(D)$ if and only if $h(\tup x)=\tup a$ for some $h \in H$. The embedding of $R(\tup a)$ is $\alpha(\multiset{o(\tup e_{1,h}, \dots, \tup e_{\ell, h}) \mid  \text{$h \in H$ and } h(\tup x)=\tup a})$.

    \item
          If $\Psi$ is a disjunction rule 
          $$R(\tup x)\angs{\alpha}\from R_1(\tup x)\lor\dots\lor R_\ell(\tup x),$$
          then $R(\tup a)\in\Psi(D)$ if and only if $R_i(\tup a) \in D$ for some $i = 1 \dots \ell$. The embedding of $R(\tup a)$ is $\alpha(\multiset{\tup e \mid 1\leq i\leq \ell \text{ and } R_i(\tup a)\angs{\tup e}\in D })$.

    \item
          If $\Psi$ is a transformation rule 
          $$R(\tup x)\angs{\mu}\from R_0(\tup x),$$
          then $R(\tup a)\angs{\tup e} \in \Psi(D)$ if and only if $R_0(\tup a)\angs{\tup e_0} \in D$ where $\tup e = \mu(\tup e_0)$. 
\end{enumerate}
The result of applying an NRP $\Pi=(\Psi_1,\dots, \Psi_n)$ to an e-database $D$ is defined as follows. Let $D_0=D$ and, for $i = 1 \dots n$, let $D_i = D_{i-1} \cup \Psi_i(D_{i-1})$. Then $\Pi(D) = D_n$. Since $\Pi$ is non-recursive and every rule in $\Pi$ produces a different IDB, $D_n$ is guaranteed to be an e-database where each fact has exactly one embedding.
We provide concrete examples of NRPs later in the paper, which use shorthand notations introduced next. Specifically, \Cref{sec:example} illustrates how to model an end-to-end relational learning pipeline as an NRP. 

\paragraph{Shorthand notations}
We introduce several syntactic conveniences. 
In conjunction rules, the aggregation function is irrelevant when all body variables occur in the head, and the combination function is irrelevant when at most one body atom has a positive embedding dimension. In disjunction rules, the aggregation function is irrelevant when the body has at most one atom. In these cases we omit the function from the rule description.
We let ``$\adom$'' denote a $1\angs{0}$-ary relation containing all active domain elements. It is produced by aggregation and disjunction over all relations in the signature, followed by a transformation that removes the embedding. Finally, for $\tup{e} \in \reals^k$, we write $\text{True}_{\tup{e}}()$ for a $0\angs{k}$-ary fact
with embedding $\tup{e}$, produced by the pair of rules 
$\text{True}() \from$ (with an empty rule body) and 
$\text{True}_{\tup{e}}()\angs{\mu_{\tup{e}}}\from \text{True}()$, where $\mu_{\tup{e}}$ is the constant function returning $\tup{e}$. 

\paragraph{Syntactic restrictions}
A relation $R$ is  \emph{zero-ary} if $\arity(R) = 0\angs{d}$ and \emph{monadic} if $\arity(R) = 1\angs{d}$ for some $d\geq 0$. An NRP $\Pi$ is zero-ary if all its IDBs are zero-ary and all bodies of rules in $\Pi$ use zero-ary relations or relations with the extended arity $k\angs{0}$. Similarly, $\Pi$ is monadic if all its IDBs are zero-ary or monadic and all bodies use zero-ary or monadic relations or relations with extended arity $k\angs{0}$.
$\Pi$ is \e{frontier-guarded} if, for every rule $\Psi$, the variables in the head of $\Psi$ co-occur in at least one body atom. $\Pi$ is \e{disjunction-free} if it has no disjunction rules.

\paragraph{NRP queries}

We study the expressiveness of NRPs as a query language over embedded 
databases. When all EDBs and IDBs have embedding dimension $0$, NRPs 
are non-recursive Datalog programs and thus
compute unions of conjunctive queries. Their added expressive power stems from their ability to process embeddings.
We accordingly define \emph{embedded queries} on embedded databases.
A $k\angs{d}$-ary embedded query $\calQ$ over $\scs$ is a function from e-databases over $\scs$ to e-relations of extended arity $k\angs{d}$.
When $d=0$, $\calQ$ is a \emph{flat query}, mapping e-databases over $\scs$ to classical $k$-ary relations $\calQ(D) \subseteq \adom(D)^k$. If both $d=0$ and $k=0$, $\calQ$ is a \emph{Boolean flat query}.

Let $\Pi = (\Psi_1,\dots,\Psi_n)$ be an NRP where $\Psi_n$ produces $R_n$ with $\arity(R_n)=k\angs{d}$.
Then $\Pi$ computes the embedded query $\calQ_\Pi$ such that $\calQ_\Pi(D)$ is the e-relation $R_n$ in $\Pi(D)$. To compute flat queries we pair an NRP $\Pi$ with an \e{acceptance policy} $P: \reals^d \to \{0,1\}$. We call $(\Pi,P)$ a \e{gated NRP}, which computes the $k$-ary flat query
\begin{align*}
    \calQ_{(\Pi,P)}(D) = 
\set{\tup a\mid
\exists\tup e:
R_n(\tup a)\angs{\tup e}\in\Pi(D)\land P(\tup e)=1}.
\end{align*}
When $\calQ_{(\Pi,P)}$ is a Boolean flat query we write $D \models \calQ(\Pi,P)$ when $\calQ(\Pi,P)(D) \neq \emptyset$ and $D \not\models \calQ(\Pi,P)$ otherwise.
 The specific class of admitted acceptance policies is largely irrelevant to our results, except as discussed in Section~\ref{sec:descriptive}. For concreteness, unless specified otherwise one may assume a fixed policy where $P(\tup e)=1$ if the minimum element of $\tup e$ is strictly positive and $P(\tup e)=0$ otherwise.

\subsection{Comparing Combination Functions}
\label{sec:CompareNRPs}

We briefly justify fixing the element-wise product $(\odot)$ as our combination function. While element-wise sum $(\tsum)$ and concatenation $(\oplus)$ are natural alternatives, both can be simulated using $\odot$.
Let NRPs $\Pi_1$ and $\Pi_2$ be \emph{equivalent} if they 
compute the same embedded queries, that is, for every e-database $D$, the final rules of $\Pi_1$ and $\Pi_2$ produce the same e-relations in $\Pi_1(D)$ and $\Pi_2(D)$.
\begin{restatable}{proposition}{propNRPProdSimulatesConcat}
\label{prop:NRM_prod_simulates_concat}
    For every NRP\/ $\Pi_1$ that uses $\tsum$ or $\oplus$ as its combination function, 
    there exists an equivalent 
    $\Pi_2$ that uses $\odot$ as its combination function. Moreover:
    \begin{enumerate}
        \item If \,$\Pi_1$ is zero-ary / monadic / frontier-guarded / disjunction-free, then so is $\Pi_2$;
        \item For any function class $\mathcal{F}$ 
    containing $\reluFFNs$, if\,
    $\Pi_1$ uses only transformations from $\mathcal{F}$, then the same holds for $\Pi_2$. 
    \end{enumerate}
    \end{restatable}
\begin{proof}
    Consider a conjunction rule in $\Pi_1$ with $\tsum$ combination:
    \[ R(\tup x)\angs{\tsum} \;\from_{\tsum}\;
        R_1(\tup x_1)\land\dots\land R_\ell(\tup x_\ell)
    \]
Substitute this for a rule with $\oplus$ combination and a transformation rule:
     \[ R'(\tup x)\angs{\tsum} \;\Leftarrow_{\oplus}\;
        R_1(\tup x_1)\land\dots\land R_\ell(\tup x_\ell) \qquad
        R(\tup x)\angs{\mu} \from R'(\tup x)
    \]
    where the embedding of a fact $R'(\tup a)$ is the concatenation of the embeddings of $R_1(\tup a),\dots,R_\ell(\tup a)$, and  $\mu$ is a $\reluFFN$ that applies element-wise sum over this concatenation, reducing the embedding dimension back to that of $R$. 
    Now consider a rule in $\Pi_1$ with $\oplus$ combination:
     \[ R(\tup x)\angs{\tsum} \;\Leftarrow_{\oplus}\;
        R_1(\tup x_1)\land\dots\land R_\ell(\tup x_\ell)
    \]      
   Let $\arity(R_i)=k_i\angs{d_i}$ and assume w.l.o.g. that for some $m$, $d_1, \dots, d_m >0$ while $d_{m+1}, \dots, d_\ell =0$. We in turn substitute this rule for transformation rules and a rule with $\odot$ combination:
    \[\begin{array}{lllll}
        &\text{ for $1 \leq i \leq m$: } &R'_i(\tup x_i)\angs{\mu_i} & \Leftarrow &R_i(\tup x_i)\\
        &&R(\tup x)\angs{\tsum}
        &
        \Leftarrow_{\odot} &
        R'_1(\tup x_1)\land\dots\land R'_m(\tup x_m) \land R_{m+1}(\tup x_{m+1}) \land\dots\land R_\ell(\tup x_\ell) 
    \end{array}\]
    where $\mu_i(\tup e) = \tup 1^{\sum^{i-1}_{j=1} d_j} \oplus \tup e \oplus \tup 1^{\sum^{\ell}_{j=i+1} d_j}$. If $\Pi_1$ is zero-ary, then by definition all body atoms in $\Pi_1$ are zero-ary or have extended arity $k\angs{0}$, so that the heads of added transformation rules are zero-ary. Monadicity is preserved for the same reason.
\end{proof}
    
In fact, Proposition~\ref{prop:NRM_prod_simulates_concat} holds more generally for any function class $\mathcal{F}$ that includes affine transformations, i.e., transformations defined by a single linear feedforward layer.
    Product combination further suffices to simulate element-wise product transformations. We define $\calF_\times$ as the class of transformation functions $f(x_1,\ldots, x_d)=\Pi_{i\in S}(x_i)$ for subsets $S\subseteq\{1,\ldots,d\}$, closed under composition.
\begin{restatable}{proposition}{PropProductTransformationtoCombination}
\label{prop:NRPproductCombSimulatesProductTrans}
Let $\mathcal{F}$ be a function class containing $\reluFFNs$. For every NRP $\Pi_1$ using transformations from $\mathcal{F}\cup\calF_\times$, there exists an equivalent
$\Pi_2$ using $\odot$ combination and transformations from $\mathcal{F}$. Moreover, if \,$\Pi_1$ is zero-ary / monadic / frontier-guarded / disjunction-free, the same holds for~$\Pi_2$.
\end{restatable}
\begin{proof}
Consider a transformation rule in $\Pi_1$ where $\mu_S(x_1,\dots,x_d) = \prod_{i \in S} x_i$ for some non-empty $S \subseteq \{1,\dots,d\}$:
\begin{align*}
R(\tup x)\angs{\mu_S} &\from R_0(\tup x)
\end{align*}
Let $\arity(R)=k\angs{1}, \arity(R_0)=k\angs{d_0}$. We substitute this for $|S|$ transformation rules and a conjunction rule:
\begin{align*}
    \text{for $i \in S$:} \quad R_i(\tup x)\angs{\mu_i} &\from R_0(\tup x) \qquad \text{where $\mu_i$ selects the $i$-th element}\\
    R(\tup x) &\fromprod \bigwedge_{i \in S} R_i(\tup x)
\end{align*}
Note that if $\Pi_1$ is zero-ary, monadic, frontier-guarded or disjunction-free, the same holds for $\Pi_2$.
\end{proof}

The combination function $\odot$ can conversely be simulated by $\oplus$ combination and $\calF_\times$ transformations, but monadicity and frontier-guardedness are not preserved by this translation. 

\subsection{Example: Relational Graph Construction and Message-Passing}
\label{sec:example}

In machine learning over relational data, a common strategy is to
construct a graph representation of the data, with record ids as nodes and relationships derived from foreign keys as edges. Graph learning models such as Graph Neural Networks are then used to learn from the database.
Such a pipeline can be expressed
end-to-end as an NRP: graph construction, initial feature construction and
message-passing are all rules of the same program. We illustrate this with a stylized example.

Suppose the input contains relations
\[
  \mathrm{Customer}(cid,s),\qquad
  \mathrm{Product}(pid,k),\qquad
  \mathrm{Order}(oid,cid,pid,t,r),
\]
where \(s\) is the customer segment, \(k\) the product category, \(t\) the country,  and \(r\)
the order channel. The segments, categories, countries and channels in addition have associated embeddings stored in embedding tables
\(\mathrm{SegEmb}(s)\), \(\mathrm{CatEmb}(k)\), 
\(\mathrm{CountryEmb}(t)\),
and
\(\mathrm{ChanEmb}(r)\). 
We first construct a directed graph whose nodes are customers,
products, and orders:
\[
\begin{array}{lll}
  \text{CNode}(cid)\angs{\tsum}
  &\from&
   \text{Customer}(cid,s)\wedge \text{SegEmb}(s)\\

  \text{PNode}(pid)\angs{\tsum}
  &\from&
  \text{Product}(pid,k)\wedge \text{CatEmb}(k)\\

  \text{ONode}(oid)\angs{\tsum}
  &\from_{\oplus}&
   \text{Order}(oid,cid,pid,t,r)\wedge \text{CountryEmb}(t)\wedge \text{ChanEmb}(r)\\[1mm]

  \text{Node}(x)\angs{\tsum} &\from& 
      \text{CNode}(x) \lor \text{PNode}(x) \lor \text{ONode}(x) \\[3mm]

  \mathrm{PlacedBy}(oid,cid) &\from& \mathrm{Order}(oid,cid,pid,t,r)\\
  \mathrm{Contains}(oid,pid) &\from& \mathrm{Order}(oid,cid,pid,t,r)\\[1mm]

  \mathrm{Edge}(u,v) &\from&
  \mathrm{PlacedBy}(u,v)\lor \mathrm{Contains}(u,v).
\end{array}
\]
Here, embeddings for database tuples are 
stored as annotations in the monadic e-relation Node.
Note also that we omitted the aggregation function in the last three rules for readability, because all relations in question are unembedded.

For a graph $G$ with nodes $V(G)$, edges $E(G) \subseteq V(G)\times V(G)$ and an embedding function $\lambda: V(G) \to \reals^d$, a layer of local-sum message-passing in a Graph Neural Network (GNN) produces a new embedding function as:
\begin{align*}
    \lambda'(u) = \rho(\lambda(u),\, \tsum\multiset{\lambda(v) \,\mid\,(u,v) \in E(G)})
\end{align*}
where $\rho$ is usually a parametrized function such as a neural network. We perform the same computation on our constructed graph with a monadic NRP:
\[
\begin{array}{lll}
  \text{Aggr}(u)\angs{\tsum}
  &\from&
  \text{Edge}(u,v)\wedge \text{Node}(v)\\
  \text{AggrZeroes}(u)\angs{\mu_{\tup{0}}}
  &\from&
  \text{Node(u)} \\
  \text{AggrOrZeroes}(u)\angs{\tsum}&\from& \text{Aggr}(u)\lor\text{AggrZeroes}(u) \\[1mm]

  \mathrm{Concat}(u)
  &\from_{\oplus}&
  \mathrm{Node}(u)\wedge \textrm{AggrOrZeroes}(u) \\[1mm]
  
  \mathrm{Node}'(u)\langle \rho\rangle
  &\from&
  \mathrm{Concat}(u).
\end{array}
\]
where $\mu_{\tup{0}}$ is the constant function returning an all-zeroes vector. Recall that concatenation $(\oplus)$ can be rewritten to product $(\odot)$ using Proposition~\ref{prop:NRM_prod_simulates_concat}. 
The relation \(\textrm{AggrOrZeroes}\) aggregates messages from the neighbors of $u$, and its embedding is the all-zeroes vector if \(u\) does not have any neighbors.
The trainable transformation \(\rho\)
computes the updated embedding. Repeating this block yields a multi-layer GNN
over the graph derived from the relational input. NRPs can similarly express other forms of message-passing, such as those using global readouts or edge labels.

\section{Zero-Ary Programs and Non-Adaptive Query Algorithms}
\label{sec:zero-ary}

In this section, we consider the special case of zero-ary NRPs where all IDBs have content arity 0. It turns out that these neatly correspond to an existing formalism: \emph{non-adaptive query algorithms}.

For a classical database schema (where all relations have embedding dimension 0),
a \emph{non-adaptive left query algorithm over $\nats$} (or simply, \emph{non-adaptive query algorithm}) is a pair
$A=(\{F_1, \ldots, F_r\},X)$, consisting of a finite set of databases and a set $X\subseteq\nats^k$. Such a query algorithm $A$ \emph{accepts} a database $D$ if $(|\Hom(F_1,D)|,\ldots,|\Hom(F_r, D)|) \in X$, 
where $\Hom(F,D)$ is the set of homomorphisms from $F$ to $D$, i.e., the functions $h: \adom(F) \to \adom(D)$ such that $R(a_1,\dots,a_k) \in F$ implies $R(h(a_1),\dots,h(a_k)) \in D$.
A class $\mathcal{C}$ of databases \emph{admits a non-adaptive query algorithm}
if there exists such an algorithm that accepts precisely the databases in $\mathcal{C}$.
The term `algorithm' is fitting  
because one can think of a query algorithm as a procedure for determining class membership: compute the homomorphism-count vector $(|\Hom(F_1,D)|,\ldots,|\Hom(F_r,D)|)$ and verify whether it belongs to the set $X$. The algorithm is `non-adaptive' since the $F_i$ are fixed in advance, rather than chosen based on previously obtained homomorphism counts. Note, however, that this is a rather abstract notion of an algorithm since no computability requirements are placed on $X$. 
The study of query algorithms was initiated by Chen et al.~\cite{Chen2025} for the special case of graphs, and subsequently generalized to arbitrary relational databases (as defined here) in~\cite{ten2024homomorphism}. 
Various results were obtained in \cite{Chen2025,wu2023study,ten2024homomorphism,MFCS2025}, such as the fact that every class defined by an alternation-free first-order sentence (i.e., Boolean combination of universal FO sentences) admits a query algorithm, whereas there is a class defined by an
FO sentence with the quantifier prefix $\exists\forall$, namely $\exists x\forall y\neg E(x,y)$, that does not admit a non-adaptive query algorithm~\cite{Chen2025}. 
Non-adaptive query algorithms are also related to \emph{graph motif parameters}~\cite{curticapean2017homomorphisms, DBLP:conf/icml/JinBCL24}, which are linear combinations of homomorphism counts. Concretely, a class of graphs is defined by a linear threshold on a graph motif parameter if and only if it admits a non-adaptive query algorithm in which $X$ is defined by a linear inequality.

Recall that a gated zero-ary NRP $(\Pi,P)$ over schema $\scs$ computes a Boolean flat query $\calQ_{\Pi,P}$. Equivalently, $(\Pi,P)$ computes a class of $\scs$-databases   $\mathcal{C}_{(\Pi,P)} = \{ D \mid D \models \calQ_{\Pi,P}\}$.
We prove that gated zero-ary NRPs compute precisely the classes that admit non-adaptive query algorithms. 
The following statements apply to NRPs with arbitrary transformation functions.
\begin{restatable}{theorem}{thmQueryAlgs}
\label{thm:query-algs}
    For all classes $\mathcal{C}$ of classical databases, 
\begin{center}
$\mathcal{C}$ admits a non-adaptive query algorithm
\quad
iff 
\quad $\mathcal{C}$ is computed
    by a gated zero-ary NRP.
\end{center}
    \end{restatable}
Theorem~\ref{thm:query-algs} follows essentially from the definitions, since, on classical databases, the embeddings computed by NRPs are determined by homomorphism counts. Its main value lies in transferring known results for non-adaptive query algorithms regarding the definability of first-order queries. We obtain the following corollary, which we revisit in the next section:
\begin{corollary}~
\label{cor:zeroary}
\begin{enumerate}
    \item
    Every Boolean flat query defined by an alternation-free FO sentence is computed by a gated zero-ary  NRP;
\item    
    The Boolean flat query $\exists x\forall y\neg R(x,y)$ is not computed by a gated zero-ary NRP. 
\end{enumerate}
\end{corollary}
We illustrate Theorem~\ref{thm:query-algs} with an example, in which an NRP computes the non-adaptive query algorithm that accepts all undirected graphs with an even number of triangles. Since parity is not definable in first-order logic (\FO)~\cite{immerman1983languages}, this shows together with Corollary \ref{cor:zeroary} that zero-ary NRPs and \FO have incomparable expressive power. In the following section (\cref{thm:Guarded vs Monadic separation}), we show that the same holds for monadic NRPs.
\begin{example}\label{ex:query-algs}
        Consider simple graphs (i.e., undirected graphs without loops), represented as databases over the schema $\scs=\{V[1]\angs{0},E[2]\angs{0}\}$. 
        The non-adaptive query algorithm $(\{C_3\},X)$
        where $C_3$ is the 3-cycle and $X=\{12n\mid n\geq 0\}$ accepts precisely the graphs containing an even number of 
        triangles. To see this, note that $C_3$ admits 6 homomorphisms to any given triangle. The same class of graphs is computed by 
        the following NRP, where $\text{True}_0(), \text{True}_1()$ have embeddings $(0)$ and $(1)$, respectively, and where  $\mu(x)=1$ if $x$ is divisible by 12 and $\mu(x)=0$ otherwise:
        \[\begin{array}{lll}
          \text{Triangle}()\angs{\tsum} &\from& E(x,y)\land E(y,z)\land E(x,z) \land \text{True}_1() \\
          \text{TriangleOrZero}()\angs{\tsum} &\from& \text{Triangle}() \lor \text{True}_0() \\
          \text{Ans}()\angs{\mu} & \from& \text{TriangleOrZero}()
        \end{array}\]
        Note that a Triangle e-fact is only derived if the graph contains a triangle. To handle graphs without triangles we apply disjunction with $\text{True}_0()$.  
    \end{example}

    \begin{remark}
In~\cite{ten2024homomorphism}, the authors also 
study non-adaptive left query algorithms over the Boolean semiring $\mathbb{B}$. Such algorithms query homomorphism existence rather than homomorphism counts. A straightforward variation of the proof of Theorem~\ref{thm:query-algs} establishes a similar connection between these query algorithms and zero-ary NRPs with max aggregation.
\end{remark}

\section{Monadic and Frontier-Guarded Programs}\label{sec:monadic-and-guarded}

In this section, we discuss two syntactic fragments of the class of NRPs, namely \e{monadic} and \e{frontier-guarded} NRPs, focusing on their connections to deep homomorphism networks~\cite{maehara2024deep,schonherr2026expressive}.

\subsection{Monadic Programs and Deep Homomorphism Networks}
\label{sec:DHN}

The example in \cref{sec:example} illustrated that monadic NRPs can represent Graph Neural Networks (GNNs), where node embeddings are updated by summing over local neighborhoods.
We show in this section that monadic NRPs are in fact more expressive than GNNs since they exactly capture the expressive power of Deep Homomorphism Networks (DHNs). 
DHNs are machine learning models that were originally defined on graphs with node embeddings, as an extension of GNNs~\cite{maehara2024deep}. They have recently been generalized to element-embedded databases \cite{schonherr2026expressive} where all elements (as opposed to all facts) have a numeric embedding. An element-embedded database 
$(D,\lambda)$ is a classical database with an embedding function $\lambda: \adom(D) \to \reals^d$, with $d \geq 0$. We use $D^a$ and $(D^a,\lambda)$ to denote pointed (element-embedded) databases with a distinguished element $a$, or $D^\bullet$ without explicit reference to this element. For the remainder of this section, we fix an input schema $\scs$ which for simplicity contains no zero-ary relations. We now define DHNs on databases as in \cite{schonherr2026expressive}:
\begin{definition}[Homomorphism query]
    A homomorphism $h \in \Hom(F^\bullet,D^\bullet)$ is a function $h: \adom(F) \to \adom(D)$ such that $R(a_1,\dots,a_k) \in F$ implies $R(h(a_1),\dots,h(a_k)) \in D$ and $h$ maps the distinguished element of \,$F$ to that of \,$D$.
    A homomorphism query
    with input dimension $d\geq 0$ and output dimension $d'\geq0$ is a pair $(F^\bullet, \mathbf{\mu})$, where $F^\bullet$ is a pointed classical database,
    and $\mathbf{\mu} = \{\mu^y : y \in \adom(F)\}$ is a set of transformation functions $\mu^y: \reals^d \to \reals^{d'}$. The result of evaluating $(F^\bullet,\mu)$ on a pointed element-embedded database $(D^\bullet,\lambda)$ with embeddings of size $d$ is the vector:
    \begin{align*}
    \res((F^\bullet,\mu),(D^\bullet,\lambda)) := \tsum \bigmultiset{ \prod_{y \in \adom(F)} \mu_y(\lambda(h(y))) \mid h \in \Hom(F^\bullet, D^\bullet)}
    \end{align*}
    Where $\tsum$ and $\prod$ denote element-wise sum and product.
\end{definition}
\begin{definition}[DHN layer]
    A Deep Homomorphism Network layer
    with input dimension $d \geq 0$ and output dimension $d' \geq 0$  is a pair $\calL = (((F^\bullet_1, \mathbf{\mu}_1), \dots, (F^\bullet_m, \mathbf{\mu}_m)), \rho)$ where each $(F^\bullet_i,\mathbf{\mu}_i)$ is a
    homomorphism query with input dimension $d$ and output dimension $d_i$, and $\rho: \reals^{d_1} \times \dots \times \reals^{d_m} \to \reals^{d'}$ is a transformation function. 
    Given an element-embedded database $(D,\lambda)$ with embeddings of size $d$,
    $\calL$ produces an embedding function mapping each $a \in \adom(D)$ to:
    \begin{align*}
        \calL(D,\lambda)(a) & = \rho(\displaystyle\bigoplus_{1 \leq i \leq m} \left(\res((F^\bullet_i, \mu_i),(D^a,\lambda)) \right))
    \end{align*}
\end{definition}
\begin{definition}[DHN]
    A Deep Homomorphism Network with input dimension $d\geq0$ and output dimension $d'\geq0$
    is a sequence  $\calN=(\mathcal{L}_1, \ldots ,\mathcal{L}_n)$ of DHN layers where the input dimension of $\calL_1$ is $d$, the output dimension of $\calL_n$ is $d'$ and for $i = 1, \dots, n-1$ the output dimension of $\calL_i$ equals the input dimension of $\calL_{i+1}$. Applying $\calN$ to an element-embedded database $(D,\lambda)$ with embeddings of size $d$  
    produces the embedding function $(\calL_{n} \circ \cdots\circ \calL_1)(D,\lambda)$, where we write $(\calL \circ \calL')(D,\lambda)$ denoting $\calL(D,\calL'(D,\lambda))$.
\end{definition}
\begin{definition}[DHN classifier]
A  \emph{DHN classifier} is a pair  $(\calN,P)$ with 
$\calN$ a DHN of some output dimension $d$ and  $P: \reals^d \to \{0,1\}$ an acceptance policy.
\end{definition}
DHNs generalize the standard message passing of GNNs from edges to structural patterns. The practical value of this generalization is demonstrated in \cite{irmai2026triangle}, where triangle-based message-passing is used to outperform state-of-the-art heuristics for the multicut problem while maintaining feasible runtimes.
The expressive power of DHNs was studied in \cite{schonherr2026expressive}, where, in particular, it was shown that DHNs are strictly more expressive than GNNs, even when the latter are augmented with homomorphism counts as additional node features, as in~\cite{barcelo2021graph,DBLP:conf/icml/JinBCL24}.
We formalize how DHNs compute queries over e-databases via \emph{element-embedded representations}, which retain only the embeddings of monadic relations and discard all others. Let $P_1, \dots, P_n$ be the monadic relations in $\scs$ with $\arity(P_i) = 1\angs{d_i}$. The element-embedded representation of an e-database $D$ over $\scs$ is the pair $\varepsilon(D) = (D', \lambda')$, 
where $D'$ is the classical database over 
$\scs\cup\{\adom\}$ 
obtained from $D$ by dropping all embeddings and adding the $1\angs0$-ary relation $\adom$ that holds every element of $\adom(D)$. Since monadic DHNs only have body atoms with content arity $1$ or embedding dimension $0$, to match DHNs on $\varepsilon(D)$ we apply monadic NRPs to $D^\downarrow$, the e-database obtained from $D$ by dropping embeddings for non-monadic relations.
In the remainder of this section, ``$Q$ is computed by a monadic NRP $\Pi$'' thus means $Q(D) = Q_\Pi(D^{\downarrow})$ for all $D$.
The embedding function $\lambda': \adom(D) \to \reals^m$ with $m=\sum^n_{i=1}(d_i+1)$ stores, for each $a \in \adom(D)$, the embeddings of monadic relations $P_i$, with an additional bit that records if $P_i(a)$ holds:
$$\lambda'(a) = \displaystyle\bigoplus_{1 \leq i \leq n} \begin{cases}
    (1)\oplus\tup e &\text{ if } P_i(a)\angs{\tup e} \in D\\
    (0)\oplus\tup 0^{d_i} &\text{ otherwise.}
\end{cases}$$
A DHN $\calN$ then computes the following unary embedded query
\begin{align*}
  \calQ_\calN(D) = \{(a)\angs{\tup e} \,\mid\, \calN(\varepsilon(D))(a) = \tup e\}.  
\end{align*}
Similarly, a DHN classifier $(\calN,P)$ computes the unary flat query $\calQ_{(\calN,P)}$ containing all elements of which the output embedding is accepted by $P$. We call a $k\angs{d}$-ary embedded query $\calQ$ \emph{total} if for every database $D$ and tuple $\tup a \in \adom(D)^k$ there is an $\tup e \in \reals^d$ such that $\tup a \angs{\tup e} \in \calQ(D)$. 
Since DHNs work on element-embedded databases they always compute an embedding for each element of the active domain. Hence, unlike for NRPs, every embedded query they compute is total.
\begin{restatable}{theorem}{ThmNRPTotalDHN}
\label{Thm:NRPMatchesDHN}
For every total unary embedded query $\calQ$, 
\begin{center}
     $\calQ$ is computed by a DHN $\calN$
    \quad iff \quad $\calQ$ is computed by a monadic NRP $\Pi$.
\end{center}
Moreover, for each class of transformation functions $\mathcal{F}$ including $\reluFFNs$, $\calN$ has transformations in $\calF \cup \calF_\times$ if and only if \,$\Pi$ has transformations in $\calF$.
\end{restatable}
Recall from \Cref{sec:nrp} that $\calF_\times$ contains element-wise multiplications over subspaces, and that we fix $\tsum$ aggregation and $\odot$ as combination, which can simulate $\oplus$ combination (Proposition~\ref{prop:NRM_prod_simulates_concat}). To prove Theorem~\ref{Thm:NRPMatchesDHN}, we simulate DHNs with monadic NRPs and vice versa. 
We give an example illustrating how a DHN is represented by a monadic NRP: 
\begin{example}
\label{ex:triangle}
    This example concerns $\reals$-labeled simple graphs (i.e., undirected graphs without loops and $1$-dimensional embeddings associated to vertices), represented by e-databases over the schema $\scs=\{V[1]\angs{1},E[2]\angs{0}\}$. Let $\calQ_{\triangle}$ be the $1\angs{1}$-ary embedded query given by
    \begin{align*}\calQ_{\triangle}(D)=\{u\angs{r}\mid~ & V(u)\angs{r_1}\in D, \\& r=\tsum \multiset{r_1\cdot r_2\cdot r_3\mid \{V(v_2)\angs{r_2}, V(v_3)\angs{r_3}, E(u,v_2),E(v_2,v_3),E(v_3,u)
    \}\subseteq D}.\end{align*}
    Thus, $\calQ_\triangle$ takes the sum of products of embeddings for all triangles containing $u$, yielding embedding $(0)$ if no such triangle exists. 
    While $\calQ_\triangle$ is not computed by a GNN \cite{morris2019weisfeiler}\cite{xu2019powerful}, it is computed by a single DHN layer that queries homomorphisms from the triangle, as well as by the following monadic NRP:
    \[\begin{array}{lll}
          \text{Zero}(x) &\from &\adom(x) \land \text{True}_0()\\
          \text{Triangle}(x)\angs{\tsum} &\from_\odot& E(x,y)\land  E(y,z)\land  E(z,x)\land P(x)\land P(y)\land P(z)\\
          \calQ_\triangle(x) &\from &\text{Triangle}(x) \lor \text{Zero}(x)
    \end{array}\]
\end{example}
DHNs can also represent non-total queries computed by monadic NRPs, using an additional Boolean in embeddings to indicate whether the NRP’s output relation contains a fact:
\begin{restatable}{theorem}{ThmNonTotalQueriesDHN}
\label{Thm:NRPtoDHN}
Let $\calQ$ be a $1\angs{d}$-ary embedded query computed by a monadic NRP $\Pi$.
Then there exists a DHN $\calN$ such that for each e-database $D$ and $a \in \adom(D)$:
\begin{align*}
    \calN(\varepsilon(D))(a) = \begin{cases}
            (1)\oplus \tup e & \text{ if } a\angs{\tup e} \in \calQ(D) \\
            (0)\oplus\tup 0^{d} &\text{ if for all } \tup e \in \reals^d,  a\angs{\tup e} \not\in\calQ(D)
        \end{cases}
\end{align*}
If\, $\Pi$ has transformations in $\mathcal{F}$, including $\reluFFNs$, $\calN$ has transformations in $\mathcal{F} \cup \calF_\times$.
\end{restatable}
As a corollary of Theorems~\ref{Thm:NRPMatchesDHN} and~\ref{Thm:NRPtoDHN}, the flat queries computed by DHN classifiers exactly match those computed by gated monadic NRPs:
\begin{restatable}{corollary}{CorDHNvsNRPRelational}
\label{cor:DHN_matches_NRP_relational}
    For every unary flat query $\calQ$,
\begin{center}
     $\calQ$ is computed by a DHN classifier
    \quad iff \quad $\calQ$ is computed by a gated monadic NRP.
\end{center}
Moreover, for each class of transformation functions $\calF$ including $\reluFFNs$, the DHN has transformations in $\calF \cup \calF_\times$ if and only if\, the NRP has transformations in $\calF$. 
\end{restatable}
Indeed, suppose a flat query is computed by DHN classifier $(\calN,P)$. Then $\calN$ computes a total query which, by Theorem~\ref{Thm:NRPMatchesDHN}, is also computed by a monadic NRP $\Pi$; thus $(\Pi,P)$ computes the same flat query. Conversely, a flat query computed by a gated monadic NRP $(\Pi,P)$ is the result of selecting 
embeddings from an embedded query. By Theorem~\ref{Thm:NRPtoDHN} there exists a DHN $\calN$ that produces the same embeddings appended with the value $1$. Thus, $(\calN,P')$ computes the same flat query, where $P'(\tup e)=1$ if and only if the first value of $\tup e$ is $1$ and $P(\tup e_{>1})=1$, where $\tup e_{>1}$ is obtained by removing the first value from $\tup e$.

It was shown in \cite{schonherr2026expressive} that DHNs express
all queries in the \emph{unary quantifier-alternation fragment (UQAFO)}, which consists of first-order formulas $\varphi(\tup x)$ 
generated by the following grammar:
\[
\begin{array}{r@{\;}c@{\;}l}
\varphi(\tup{x}) &::=& \varphi_\exists(\tup{x})\mid \varphi_\forall(\tup{x}) \mid \varphi(\tup{x}) \circ \varphi(\tup{x}) 
\\[1mm]
\varphi_\exists(\tup{x}) &::=& R(\tup{x})\mid \neg R(\tup{x})\mid \varphi_\exists(\tup{x})\circ \varphi_\exists(\tup{x})
\mid x_i=x_j\mid x_i\neq x_j\mid \exists\tup{y}\,\varphi_\exists(\tup{x}, \tup{y})\mid \varphi_\forall(x_i) \\[1mm]
    \varphi_\forall(\tup{x}) &::=& R(\tup{x})\mid \neg R(\tup{x})\mid \varphi_\forall(\tup{x})\circ \varphi_\forall(\tup{x})\mid 
    x_i=x_j\mid x_i\neq x_j\mid \forall\tup{y}\,\varphi_\forall(\tup{x}, \tup{y})\mid \varphi_\exists(x_i)
  \end{array}
  \]
  where $\circ \in \{\lor,\wedge\}$ and $\varphi(\tup{x})$ only has free variables in $\tup{x}$. Note that quantifier alternation is only allowed when the quantified subformula has at most one free variable. For example,
  $\exists x\forall y \neg E(x,y)$ is allowed because $\forall y \neg E(x,y)$ has a single free variable, but $\exists x\exists y\forall z (E(x,z)\lor E(y,z))$ is not. 
  It was shown in \cite[Theorem 12]{schonherr2026expressive} that, for each UQAFO-formula $\varphi(x)$ in one free variable, there exists a DHN that derives element-embedding $0$ when $\varphi(x)$ is false and $1$ when $\varphi(x)$ is true. It was further shown that DHNs express first-order properties that lie outside UQAFO.
  Theorem~\ref{Thm:NRPMatchesDHN} allows us to transfer these results about DHNs to NRPs. This yields:\footnote{The claim for Boolean flat queries follows from the one for unary flat queries: we treat an UQAFO-sentence $\varphi$ as 
a formula with one (unused) free variable, and extend the resulting program with 
a final $\reluFFN$-transformation rule, followed by an aggregation rule of the form $\text{GlobalAns}()\angs{\tsum}\from \text{LocalAns}(x)$, and a disjunction rule for the empty domain case.}
  \begin{corollary}
  \label{cor:NRP_strict_subsumes_UQAFO}
    The Boolean and unary flat queries computed by gated monadic NRPs strictly subsume those definable by UQAFO-formulas, with a separating example in FO and on classical databases. 
  \end{corollary}
  Similarly, results in~\cite{maehara2024deep} and \cite[Theorem 13]{schonherr2026expressive}, combined with Theorem~\ref{Thm:NRPMatchesDHN}, imply:
  \begin{corollary}
        The unary flat queries computed by gated monadic NRPs strictly subsume those computed by GNN-classifiers augmented by finitely many homomorphism-count features (as in~\cite{barcelo2021graph}), with a separating example in FO and on unembedded simple graphs.
\end{corollary}
Here, simple graphs are undirected graphs without loops. Finally, recall from Section~\ref{sec:zero-ary} (\cref{cor:zeroary}) that the Boolean flat query $\exists x\forall y \neg E(x,y)$ is not computed by a zero-ary NRP. Since this first-order sentence belongs to UQAFO, we get:
\begin{corollary}
    The Boolean flat queries computed by gated monadic NRPs strictly subsume those computed by gated zero-ary NRPs, with a separating example in FO and on unembedded simple graphs.
\end{corollary}

\subsection{Moving Beyond Monadic: Frontier-Guarded Programs}
\label{sec:guarded}
We now extend the connection between NRPs and DHNs beyond monadic programs. We show that, while frontier-guarded NRPs are strictly more expressive than monadic NRPs, they can still
be translated to DHNs, provided that the input database is suitably normalized to have row-ids. Recall that an NRP is \emph{frontier-guarded} if each conjunction rule has a guard atom in the body that contains all variables in the head.
Frontier-guarded NRPs are more
general than  monadic NRPs
in that they 
 can define $k$-ary queries with $k>1$.
We provide examples to illustrate the usage of non-monadic aggregation, also for unary queries:
\begin{example}
\label{ex:guarded}
Consider a schema $\scs=\{\text{Supervises}[2]\angs{0}, \text{PostDoc}[1]\angs{0}, \text{PhD}[1]\angs{0}\}$,
and consider the unary flat query:
\[
\text{People $x$ supervising a PostDoc $y$ that co-supervises all (and at least one) PhD students of $x$.}
\]
This is computed by the following
frontier-guarded NRP with answer relation $\text{Ans}[1]\angs{1}$ and the standard acceptance policy that checks if the embedding is positive:
\[
\begin{array}{lll}
\text{SPD}(x,y) &\from& \text{Supervises}(x,y)\land\text{PostDoc}(y) \\
\text{JointPhDCount}(x,y)\angs{\tsum} &\from& 
\text{SPD}(x,y)\land\text{Supervises}(x,z)\land\text{Supervises}(y,z)\\
&&\land\, \text{PhD}(z) \land \text{True}_1() \\
\text{MainPhDCount}(x,y)\angs{\tsum} &\from& 
\text{SPD}(x,y)\land\text{Supervises}(x,z)\land\text{PhD}(z) \land \text{True}_1() \\
\text{CoSupStats}(x,y) &\from_\oplus& \text{JointPhDCount}(x,y) \land \text{MainPhDCount}(x,y)\\
\text{Good}(x,y)\angs{\mu} &\from& \text{CoSupStats}(x,y) \textnormal{\mbox{~~~}~~~ where $\mu(e,e')=\relu(1-(e'-e))$}\\
\text{Ans}(x)\angs{\tsum} &\from& \text{Good}(x,y)
\end{array}
\]
Indeed, $\text{Ans}[1]\angs{1}$ includes any person $x$ supervising a postdoc that co-supervises at least one of $x$'s PhD students. The embedding of $\text{Ans}(x)$ is the number of $x$'s postdocs co-supervising all of $x'$s PhD students.
The second rule performs an aggregation grouped by pairs, which monadic rules cannot do.
\end{example}

\begin{example}
    GNNs as described in Section~\ref{sec:example} 
    only compute node embeddings.
    There are also GNN variants that compute both node and edge embeddings~\cite{DBLP:conf/icml/GilmerSRVD17,DBLP:journals/corr/abs-1806-01261,dwivedi2021generalization}. 
    Such architectures cannot be represented by monadic NRPs but they can be represented by frontier-guarded NRPs. 
    For instance, an edge-embedding update of the form 
        $\lambda'_{i + 1}((u, v)) = \rho(\lambda_i(u), \lambda_i(v), \lambda'_i(u, v))$ translates to a pair of frontier-guarded NRP rules:
    \[
    \text{Concat}_i(x, y) \from_\oplus \lambda_{i}(x) \land \lambda_{i}(y) \land \lambda'_{i}(x, y)\qquad
 \lambda'_{i+1}(x, y)\angs{\rho} \from \text{Concat}_i(x, y).
    \]
\end{example}
We next prove that, also for unary queries, frontier-guarded NRPs are strictly more expressive than monadic NRPs.
\begin{restatable}{theorem}{thmGuardedvsMonadic}
\label{thm:Guarded vs Monadic separation}
The flat queries computed by gated frontier-guarded NRPs strictly subsume those computed by gated monadic NRPs, with a unary flat query in \FO and over classical databases as separating example. 
\end{restatable}

Every monadic NRP is frontier-guarded. To prove strictness, we
consider the unary flat query $\mathcal{Q}_M$ over schema $\{E[2]\langle
0\rangle, P[1]\langle 0\rangle\}$ defined by
the FO formula
\[
  \varphi(x)
  =
  \exists y\bigl(
    E(x,y) \wedge P(y) \wedge
    \neg\exists z
      (E(x,z) \wedge E(y,z)) 
  \bigr).
\]
It is easy to define $\mathcal{Q}_M$ by a gated frontier-guarded NRP. Suppose
that $\mathcal{Q}_M$ was definable by a gated monadic NRP. By Corollary~\ref{cor:DHN_matches_NRP_relational}, it
would then be computed by a DHN classifier. To obtain a contradiction,
we adapt the pattern-WL test for DHN classifiers of Luo and
Virtema~\cite{Luo2026} from graphs to classical databases.
This test extends classical Weisfeiler-Leman (WL) coloring by incorporating embeddings of a fixed set of pattern databases. Every DHN classifier is invariant under finitely many rounds of $\mathcal{P}$-WL for some finite pattern set $\mathcal{P}$. Using a careful analysis,
we show that the patterns may be assumed to be connected and non-decomposable, that is, to have no
articulation points. For any pattern set
$\mathcal{P}$, we then construct pointed databases
$D_1^v$ and $D_2^u$ with the same one-round $\mathcal{P}$-WL color, but with $v \in
\mathcal{Q}_M(D_1)$ and $u \notin \mathcal{Q}_M(D_2)$. The construction gives $v$ a
$P$-neighbor with no common neighbors, whereas every $P$-neighbor of
$u$ has a  neighbor in common with $u$.

Suitable binomial multiplicities nevertheless equalize all relevant pattern counts.
Then, an iterative copying construction lifts this
to any number of rounds while preserving
membership in $\mathcal{Q}_M$, contradicting invariance of the classifier.

In the remainder of this section, we show that frontier-guarded NRPs, although more expressive in general, can nevertheless be translated to monadic NRPs over databases equipped with linearly ordered row-ids.
An e-database \emph{is row-id normalized} if (i) the first attribute of each relation is a unique unary key (row-id) not appearing elsewhere in the database, and (ii) 
all non-monadic relations in its schema have embedding dimension 0.\footnote{Condition (ii) is important for the reduction from frontier-guarded NRPs to monadic NRPs, since monadic NRPs only allow rule bodies to contain EDBs with embeddings if their content arity is at most $1$.}
Monadic NRPs are more powerful on row-id normalized e-databases since they can effectively store row embeddings in monadic IDBs via row-ids; in particular, they become as expressive as frontier-guarded NRPs when equipped with a linear order.
To state this equivalence precisely we recall a well-known fact about frontier-guarded Datalog programs which also applies to frontier-guarded NRPs, namely that every content tuple occurring in the output is a projection of a 
tuple occurring in one of the EDB relations. Let
a $k$-ary \emph{EDB-projection} 
be a pair $(R,(i_1,\ldots,i_k))$ where
$R\in\scs$ has content arity $m$ and $i_1, \ldots, i_k\in\{1, \ldots, m\}$.
\begin{restatable}{theorem}{thmGuardedNRP}
\label{thm:guardedNRPs}
Let \(\Pi\) be a frontier-guarded NRP over schema $\scs$ and let $S$ be a $k\angs{d}$-ary IDB in \(\Pi\). 
There exists a monadic NRP $\Pi'$ over schema $\scs \cup \{R_<[2]\angs{0}\}$ that produces a $1\angs{d}$-ary IDB $S_\rho$ for every $k$-ary EDB projection \(\rho=(R,(i_1,\ldots,i_k))\) so that on every row-id normalized e-database \(D\):
\[
\{ (a)\angs{\tup e}\mid S_\rho(a)\angs{\tup e}\in \Pi'(D^<)\} ~~=~~ \{ (a_1)\angs{\tup e}\mid R(a_1,\ldots,a_m)\in D
\text{ and } S(a_{i_1},\ldots,a_{i_k})\langle \tup{e}\rangle\in \Pi(D)\}
\]
where $m$ is the content arity of $R$ 
and \(D^<\) is an expansion of $D$ with a linear order \(R_<\) on the row-id values.
Moreover, if $\Pi$ only uses transformations in $\calF$, including $\reluFFNs$, the same holds for $\Pi'$.
\looseness=-1
\end{restatable}
Intuitively, Theorem~\ref{thm:guardedNRPs} shows that frontier-guarded NRPs can be compiled into monadic NRPs over
row-id normalized databases. The linear order allows us to map each derived IDB fact to the row-id of the EDB tuple of which it is the projection, in a canonical way, thus preventing double counting during aggregation.
It follows that embedded queries computed by frontier-guarded NRPs are also computed by monadic NRPs after a canonical translation, and hence by DHNs using the results of Section ~\ref{sec:DHN}.
Given a schema $\scs$, let:
\begin{align*}
    \widehat{\scs} = \{R_{\text{data}}[k+1]\angs{0} \,\mid\, R[k]\angs{d} \in \scs\} \cup \{R_{\text{emb}}[1]\angs{d} \,\mid\, R[k]\angs{d} \in \scs \text{ and } d>0\}
\end{align*} 
For an e-database $D$ over $\scs$, its \emph{ordered row-id expansion} $\widehat{D}$ over $\widehat{\scs}$ is the row-id normalized e-database that contains for each e-fact $R(a_1,\ldots, a_{k})\angs{\tup{e}} \in D$ 
the e-facts 
$R_{\text{data}}(r, a_1,\ldots, a_{k})\angs{}$ 
where $r$ is a fresh distinct value that serves
as the row-id, and
$R_{\text{emb}}(r)\angs{\tup{e}}$
if $d>0$, with, in addition, a $[2]\angs{0}$ ary order $R_<$ on the row-id values. We obtain the following corollary of Theorems~\ref{Thm:NRPtoDHN} and~\ref{thm:guardedNRPs}:
\begin{corollary}
Let \(\calQ\) be a \(k\angs{d}\)-ary embedded query  computed by a frontier-guarded NRP over an arbitrary schema \(\scs\), and let $\rho=(R,(i_1,\ldots,i_k))$, be a $k$-ary EDB projection. There exists a DHN $\calN$ computing the total $1\angs{d+1}$-ary embedded query such that, for each e-database over $\scs$ with ordered row-id expansion $\widehat{D}$ and each $r \in \adom(\widehat{D})$, we have $r\angs{\tup e} \in \calQ_\calN(\widehat{D})$ where:
\begin{align*}
    \tup e = \begin{cases}
        (1)\oplus \tup f &\text{ if } R_{\mathrm{data}}(r,a_1,\ldots,a_m)\in \widehat{D}
\text{ and }
(a_{i_1},\ldots,a_{i_k})\langle \tup{f}\rangle\in \calQ(D)\\
        (0)\oplus \tup 0^{d} &\text{ otherwise}
    \end{cases}
\end{align*}
\end{corollary}
\begin{remark}
The above reduction from frontier-guarded NRPs to monadic NRPs is based on expanding each relation with a row-id column, and assuming a linear order on the row-ids. Another option is to use a graph encoding that assigns a vertex to every projection of each e-fact in the input database. This avoids the need for a linear order, at the expense of increasing the size of the input to the DHN.
\end{remark}

\section{Descriptive Complexity of \reluFFN-based NRPs}
\label{sec:descriptive}

We next study the expressive power
of arbitrary (not necessarily monadic or frontier-guarded)
NRPs. Specifically, we consider \emph{ReLU-FFN-based NRPs}, whose transformations
are \reluFFNs: feedforward neural networks with rational parameters and the activation $\relu(x) = \max(0,x)$. 
\reluFFNs are standard building blocks in machine learning and compute continuous, almost everywhere differentiable functions, as required for gradient-based optimization.
These properties naturally extend to NRPs: for a fixed input e-database, the computation performed by a \reluFFN-based NRP is continuous and almost everywhere differentiable with respect to its parameters.\footnote{It suffices to observe that, for a fixed input, an NRP can be compiled into a large but finite \FFN (with parameter sharing over tuples derived by the same transformation rule) where sum-aggregation is replaced by finitary addition.}

We exactly characterize the 
expressive power of $\reluFFN$-based NRPs as \FOCrats, an extension of first-order logic with real-valued counting terms. Using this characterization we show that the flat queries over ordered rational weighted e-databases computed by gated $\reluFFN$-based NRPs 
are contained in uniform \TCO, and that,
over
ordered classical databases, 
they capture uniform \TCO.

\subsection{\texorpdfstring{\FOCrats}{FO+C(1/2)}}

The extension of first-order logic with counting terms was first suggested by Immerman \cite{immerman1987expressibility}, and later made explicit by Gr{\"a}del and Otto \cite{gradel1992inductive}. There exist various formalizations where counting terms range over natural numbers such as FOCN$(\mathbb{P}_\leq$) \cite{kuske2017first} and \FOC \cite{grohe2024descriptive}, which all capture uniform \TCO over ordered structures. To represent the queries computed by NRPs we define \FOCrats, a first-order logic with real-valued counting terms, which we interpret on real-weighted first-order structures. First-order logic with counting over real values has recently been proposed as a query language for verification and interpretation of neural networks
\cite{grohe2025querylanguages,grohe2026recursive}. Crucial differences between this formalism and our logic are that \FOCrats does not include division and does not allow for conditioning counting terms on formulas.

\paragraph{Weighted Structures}
Given a signature $\sigma$ with relation symbols $R_1, \ldots, R_n$,
a \emph{weighted $\sigma$-structure} $\calA=(A,R_1^\calA,\ldots, R_n^\calA)$ consists of a finite domain $A$ and a function $R_i^\calA: A^k \to \reals$ for every relation symbol $R_i$ of arity $k$.
A \emph{Boolean $\sigma$-structure} is a weighted structure satisfying $R_i^\calA(\tup a)\in\{0,1\}$ for each relation symbol $R_i$ of arity $k$ and tuple $\tup a \in A^k$.
\paragraph{Syntax}
The terms and formulas of \FOCrats are given by the following grammar, where $q \in \rats$.
\[\begin{array}{llll}
        \text{Terms:}    &
        \theta\!            &::=& q \mid \mathds{1}_{x=y} \mid R(x_1,\ldots,x_k) \mid \theta+\theta'\mid \theta\cdot \theta' \mid 
        \displaystyle \sum(x_1,\ldots,x_k).\theta \mid \max(\theta,\theta') \\
        \text{Formulas:} &
        \phi\!              &::=& \theta>\theta'
        \mid \neg\phi\mid\phi\land\psi
    \end{array}\]

\paragraph{Semantics}
We define a real value $\sem{\chi}^{\calA}(\tup a)$ for every term or formula $\chi$, weighted structure $\calA$ and $\tup a \in A^k$, where $k$ is the number of free variables of $\chi$. If $\chi$ is a formula, $\sem{\chi}^\calA(\tup a) \in \{0,1\}$, and if $\chi$ is a term, $\sem{\chi}^\calA(\tup a) \in \reals$. The value $\sem{\chi}^\calA(\tup a)$ is defined by induction, cf.~Table~\ref{tab:FOCrats}.
\begin{table}[t]
    \caption{Semantics of $\FOCrats$}
\[\begin{array}{lll}
         \sem{q}^{\calA}                                        & = & q \text{~ 
         for } q \in \rats                                        \\
        \sem{\mathds{1}_{x=y}}^{\calA}(a_1,a_2)                          & = & 1 \text{ if } a_1 = a_2 \text{ and } 0 \text{ otherwise}          \\
        \sem{R(x_1, \ldots, x_n)}^{\calA}(\tup a)                        & = & R^{\calA}(\tup{a}) \\
        \sem{\theta+\theta'}^{\calA}(\tup a)                             & = & \sem{\theta}^{\calA}(\tup a) + \sem{\theta'}^{\calA}(\tup a)      \\
        \sem{\theta\cdot\theta'}^{\calA}(\tup a)                         & = & \sem{\theta}^{\calA}(\tup a) \cdot \sem{\theta'}^{\calA}(\tup a)  \\
        \sem{\displaystyle\sum(x_1, \ldots, x_k).\theta}^{\calA}(\tup a) & = &
        \text{sum of $\sem{\theta}^{\calA}(\tup a')$ where $\tup a'$ extends $\tup a$ with values for $x_1, \dots, x_k$}
        \\
        \sem{\max(\theta,\theta')}^{\calA}(\tup a)                       & = & \max(\sem{\theta}^{\calA}(\tup a),\sem{\theta'}^{\calA}(\tup a))  \\
        \sem{\theta>\theta'}^{\calA}(\tup a)                         & = &
        \text{1 if $\sem{\theta}^{\calA}(\tup a)> \sem{\theta'}^{\calA}(\tup a)$, and 0 otherwise}                                            \\
        \sem{\neg\phi}^{\calA}(\tup a)                                   & = & 1-\sem{\phi}^{\calA}(\tup a)                                      \\
        \sem{\phi\land\psi}^{\calA}(\tup a)                              & = & \min(\sem{\phi}^{\calA}(\tup a), \sem{\psi}^{\calA}(\tup a))\\
    \end{array}\]
    \label{tab:FOCrats}
    \end{table}
$\FOCrats$ extends first-order logic over Boolean structures. To see this, note that every first-order formula $\phi$ can be inductively translated to a $\FOCrats$-term $\theta_\phi$ denoting its truth value: $\theta_{\phi \wedge \psi} := \theta_\phi \cdot \theta_\psi$, $\theta_{\neg \phi} := 1-\theta_\phi$ and $\exists x \phi(x) := -\max(-\sum(x).\theta_\phi,-1)$.  It follows that the \FOCrats-formula
$\theta_\phi > 0$ is equivalent to $\phi$.
In fact, over ordered Boolean structures $\FOCrats$ is equivalent to $\FOC$ as defined in \cite{grohe2024descriptive}, as we show in Theorem~\ref{thm:FOC_is_FOCrats_is_TCO}.

\subsection{Neuro-Relational Programs and \texorpdfstring{\FOCrats}{FO+C(1/2)}}
\label{sec:NRPandFOCrats}
Given an input schema $\scs$, fix $\sigma$ to the signature that, for every $R[k]\angs{d} \in \scs$, contains $k$-ary relations $R$ and $R_1, \dots, R_d$. An e-database $D$ over $\scs$ is represented by a weighted $\sigma$-structure $\calA_D$ with domain $A=\adom(D)$ such that for all $\tup a \in A^k$, $\sem{R}^{\calA_D}(\tup a)$ is $1$ if $R(\tup a) \in D$ and $0$ otherwise, and for $i = 1, \dots, d$, $\sem{R_i}^{\calA_D}(\tup a)$ is the $i$-th embedding value of $R(\tup a)$ in $D$, or $0$ if $R(\tup a) \not\in D$. An \FOCrats formula $\phi$ with $k$ free variables then defines the $k$-ary flat query:
$$
\calQ_\phi(D) = \{\tup a \,\mid\, \tup a \in \adom(D)^k \text{ and } \sem{\phi}^{\calA_D}(\tup a)=1\}
$$

A \emph{simply-gated NRP} is a gated NRP with a \emph{simple acceptance policy}, that is, an acceptance policy $P:\reals^d\to \{0,1\}$ defined by a formula $\phi(x_1, \ldots, x_d)$ that is a Boolean combination of inequalities of the form $x_i > 0$. Because \FOCrats is interpreted on real-valued input structures, and allows for constructing all
rationals as variable-free terms, it can express \reluFFNs with 
rational parameters over e-databases. This yields the following equivalence:
\begin{theorem}
\label{thm:NRP_and_FOC(1/2)_same_relations}
For every flat query $\calQ$:
\begin{center}
    $\calQ$ is defined in $\FOCrats$ \quad iff \quad $\calQ$ is computed by a simply-gated \reluFFN-based NRP.
\end{center}
\end{theorem}
The proof of Theorem~\ref{thm:NRP_and_FOC(1/2)_same_relations} establishes a one-to-one correspondence between \FOCrats terms and the embedding values computed by \reluFFN-based NRPs. The result then follows since $\FOCrats$ formulas can express simple acceptance policies.
\begin{remark}
\label{remark:Embedded_NRP_TCO}
    For non-gated NRPs, the proof of Theorem~\ref{thm:NRP_and_FOC(1/2)_same_relations} implies an alternative characterization. A $k\angs{d}$-ary embedded query is computed by a \reluFFN-based NRP if and only if it can be represented by a tuple $(q,\theta_1,\ldots, \theta_d)$ where $q$ is a classical $k$-ary union of conjunctive queries determining which facts are derived and $\theta_1,\dots,\theta_d$ are $k$-ary \FOCrats terms representing the derived embeddings.
\end{remark}

\subsection{Relation to \texorpdfstring{\FOC}{FO+C} and \texorpdfstring{\TCO}{TC0}}
\label{sec:FOC_TC0}
The complexity class \emph{dlogtime-uniform \TCO} (henceforth simply \emph{uniform \TCO}) consists of all languages $L \subseteq \{0,1\}^*$ decided by a family $\calC = (\mathfrak{C}_n)_{n \in \mathbb{N}_{>0}}$ of polynomial size, constant depth Boolean circuits with threshold gates, subject to the additional uniformity condition: there is a deterministic Turing machine that can, in time $O(\log n)$, compute local information about $\mathfrak{C}_n$ such as the type of its \(i\)-th gate, or the address of the \(j\)-th input to that gate.
Intuitively, uniform \TCO captures the computational power of constant-time parallel algorithms with threshold gates, which equip constant-depth circuits with the ability to perform basic arithmetic operations.

A weighted structure is ordered if it includes a linear order on the domain. Ordered 
rational weighted structures admit canonical bitstring representations~\cite{DBLP:books/daglib/0095988}, where a rational weight $\frac{p}{q}$ is represented by a bitstring of length $O(\log(|p|)+\log(q))$.
For an ordered 
rational weighted structure $\calA$ and a tuple of domain elements $\tup a$, let $s(\calA, \tup a) \in \{0,1\}^*$ be the bitstring representation of $\calA$, followed by the binary positions of the elements of $\tup a$ in the order. We define flat queries on weighted structures analogously to database queries: as functions from weighted structures to classical relations. A $k$-ary flat query $\calQ$ over ordered 
rational weighted structures is then represented by the language:\looseness=-1
$$
L(\calQ) = \{s(\calA, \tup a) \,\mid\, \tup a \in \calQ(\calA)\}
$$
and we say $\calQ$ is in \TCO if $L(\calQ) \in \TCO$. Over ordered Boolean structures, the extension of first-order logic with counting quantifiers and number variables over a discrete domain (\FOC) captures uniform \TCO \cite{barrington1990uniformity}. We show that this equivalence extends to \FOCrats. Over ordered rational weighted structures, \FOCrats remains contained in \TCO, although it is properly contained.
\begin{restatable}{theorem}{thmFOCratsisFOCunweightedordered}
\label{thm:FOC_is_FOCrats_is_TCO}
    For every flat query $\calQ$ over ordered Boolean structures the following are equivalent:
    \begin{enumerate}
        \item $\calQ$ is defined in \FOCrats;
        \item $\calQ$ is defined in \FOC (as in \cite{grohe2024descriptive});
        \item $\calQ$ is in uniform \TCO.
    \end{enumerate}
\end{restatable}
\begin{restatable}{theorem}{thmFOCratsContainedInTCO}
\label{thm:FOCrats_contained_in_TCO}
Over ordered 
rational weighted structures,
the flat queries definable in $\FOCrats$ are strictly contained in uniform \TCO.
\end{restatable}

To prove Theorem~\ref{thm:FOC_is_FOCrats_is_TCO}, we use that the inequalities between the rational terms that \FOCrats builds over Boolean structures are represented in \FOC. Conversely, \FOC number variables over $\mathbb{N}$ can be simulated by variables over the ordered domain. For the containment in Theorem~\ref{thm:FOCrats_contained_in_TCO}, we encode 
weights $\frac{p}{q}$ with $p \in \mathbb{Z}, q \in \mathbb{N}$ by the concatenation of bitstring of $p$ and $q$.
 The inclusion then follows from known results on arithmetic in uniform \TCO (see \cite{barrington1990uniformity, chandra1984constant} and the result for iterated integer multiplication, and thus iterated rational addition, in \cite{hesse2002uniform}). Stritness is witnessed by the query containing tuples for which the weight of relation $R$ is an even integer.\looseness=-1

\subsection{Putting everything together}
\label{sec:Putting_everything_together}
The combination of Theorems~\ref{thm:NRP_and_FOC(1/2)_same_relations}, \ref{thm:FOC_is_FOCrats_is_TCO} and~\ref{thm:FOCrats_contained_in_TCO} bounds 
the descriptive complexity of NRPs
as follows:
\begin{corollary}
\label{cor:NRP_is_TCO_classical}
     Over ordered classical databases, for every flat query $\calQ$:
     \begin{center}
         $\calQ$ is in uniform \TCO \quad iff \quad $\calQ$ is computed by a simply-gated \reluFFN-based NRP.
     \end{center}
\end{corollary}
\begin{corollary}
\label{cor:NRP_in_TCO_weighted}
Over ordered rational weighted e-databases:
the flat queries computed by simply-gated \reluFFN-based NRPs are strictly contained in uniform \TCO. 
\end{corollary}

In Remark \ref{remark:Embedded_NRP_TCO}, we noted that not only the flat queries, but also the embedded queries computed by \reluFFN-based NRPs are characterized by \FOCrats. Accordingly, we expect that the correspondence between \reluFFN-based NRPs and uniform \TCO extends to embedded queries, but this requires careful consideration of the representation of embedded queries as binary languages.
 An interesting non-trivial consequence of Corollary~\ref{cor:NRP_is_TCO_classical} is that, over ordered classical databases, gated \reluFFN-based NRPs are closed under composition. 
Consequently, inserting 
simple acceptance policies, or equivalently Heaviside activations, anywhere in the program
does not increase expressive power.
\looseness=-1

\section{Conclusion}
Neuro-Relational Programs provide a declarative framework that combines relational querying with neural computation over  databases with fact embeddings. The syntactic fragments studied in this paper formally connect NRPs to non-adaptive query algorithms, deep homomorphism networks, the counting logic \FOCrats and the complexity class \TCO. 
These findings refine our understanding of the more general \textsc{Lara} framework by tightening its known expressivity and complexity bounds. Whereas \textsc{Lara} was known to express standard convolutions and operate in Logarithmic space, we demonstrated that the strictly weaker class of monadic NRPs already suffices to capture homomorphism queries (and thereby graph convolutions), and that \reluFFN-based NRPs fall within the narrower class of uniform \TCO.
These connections position NRPs as a unifying formalism for understanding the expressive power of neural models over structured data, while retaining the compositional and declarative character of database query languages. 

One direction that deserves further study is the development of recursive NRPs. Since recursive rules do not always reach a fixed point, this requires the use of alternative semirings, such as the tropical semiring used in shortest-path finding. It further remains open to develop exact logical characterizations of the embedded queries computed by zero-ary, monadic and guarded programs.
While this paper focuses on forward-pass expressiveness, the NRP framework naturally supports end-to-end gradient-based training, where the parameters in transformation functions $(\mu)$ are the learnable weights. 
We complement the theoretical foundation in this paper with an implemented system
~\cite{lubarsky2026incorporating} 
that translates NRPs into an embedding-aware extension of relational algebra (``Neuro-Relational Algebra''), and subsequently compiles them into physical plans over PyTorch, cuDF, and SQL. 
Our experimental results indicate that the resulting programs achieve performance comparable to state-of-the-art implementations while being substantially simpler to express. We view this as a promising direction for natively integrating deep learning models into relational query engines.

\begin{acks}
We thank Martin Grohe for valuable discussions contributing to this work.
Carsten Lutz was supported by DFG project LU 1417/4-1.
This work is partly supported by BMFTR (Federal Ministry of Research, Technology and Space)
in DAAD project 57616814 (\href{https://secai.org/}{SECAI, School of Embedded Composite AI})
as part of the program Konrad Zuse Schools of Excellence in Artificial Intelligence.
\end{acks}
\newpage
\bibliographystyle{ACM-Reference-Format}
\bibliography{references_fixed}


\begin{thebibliography}{73}


\ifx \showCODEN    \undefined \def \showCODEN     #1{\unskip}     \fi
\ifx \showISBNx    \undefined \def \showISBNx     #1{\unskip}     \fi
\ifx \showISBNxiii \undefined \def \showISBNxiii  #1{\unskip}     \fi
\ifx \showISSN     \undefined \def \showISSN      #1{\unskip}     \fi
\ifx \showLCCN     \undefined \def \showLCCN      #1{\unskip}     \fi
\ifx \shownote     \undefined \def \shownote      #1{#1}          \fi
\ifx \showarticletitle \undefined \def \showarticletitle #1{#1}   \fi
\ifx \showURL      \undefined \def \showURL       {\relax}        \fi
\providecommand\bibfield[2]{#2}
\providecommand\bibinfo[2]{#2}
\providecommand\natexlab[1]{#1}
\providecommand\showeprint[2][]{arXiv:#2}

\bibitem[Alviano et~al\mbox{.}(2023)]%
        {DBLP:conf/pods/AlvianoLMP23}
\bibfield{author}{\bibinfo{person}{Mario Alviano}, \bibinfo{person}{Matthias Lanzinger}, \bibinfo{person}{Michael Morak}, {and} \bibinfo{person}{Andreas Pieris}.} \bibinfo{year}{2023}\natexlab{}.
\newblock \showarticletitle{Generative Datalog with Stable Negation}. In \bibinfo{booktitle}{\emph{Proceedings of the 42nd ACM SIGMOD-SIGACT-SIGAI Symposium on Principles of Database Systems (PODS 2023)}}. \bibinfo{publisher}{ACM}, \bibinfo{address}{New York, NY, USA}, \bibinfo{pages}{21--32}.
\newblock


\bibitem[B{\'a}r{\'a}ny et~al\mbox{.}(2017)]%
        {DBLP:journals/tods/BaranyCKOV17}
\bibfield{author}{\bibinfo{person}{Vince B{\'a}r{\'a}ny}, \bibinfo{person}{Balder ten Cate}, \bibinfo{person}{Benny Kimelfeld}, \bibinfo{person}{Dan Olteanu}, {and} \bibinfo{person}{Zografoula Vagena}.} \bibinfo{year}{2017}\natexlab{}.
\newblock \showarticletitle{Declarative Probabilistic Programming with Datalog}.
\newblock \bibinfo{journal}{\emph{ACM Transactions on Database Systems}} \bibinfo{volume}{42}, \bibinfo{number}{4} (\bibinfo{year}{2017}), \bibinfo{pages}{22:1--22:35}.
\newblock
\href{https://doi.org/10.1145/3132700}{doi:\nolinkurl{10.1145/3132700}}


\bibitem[Barcel\'{o} et~al\mbox{.}(2026)]%
        {10.1145/3801914}
\bibfield{author}{\bibinfo{person}{Pablo Barcel\'{o}}, \bibinfo{person}{Floris Geerts}, \bibinfo{person}{Matthias Lanzinger}, \bibinfo{person}{Klara Pakhomenko}, {and} \bibinfo{person}{Jan Van~den Bussche}.} \bibinfo{year}{2026}\natexlab{}.
\newblock \showarticletitle{A Logical View of {GNN}-Style Computation and the Role of Activation Functions}.
\newblock \bibinfo{journal}{\emph{Proc. ACM Manag. Data}} \bibinfo{volume}{4}, \bibinfo{number}{2}, Article \bibinfo{articleno}{118} (\bibinfo{date}{May} \bibinfo{year}{2026}), \bibinfo{numpages}{19}~pages.
\newblock
\href{https://doi.org/10.1145/3801914}{doi:\nolinkurl{10.1145/3801914}}


\bibitem[Barcel{\'o} et~al\mbox{.}(2021)]%
        {barcelo2021graph}
\bibfield{author}{\bibinfo{person}{Pablo Barcel{\'o}}, \bibinfo{person}{Floris Geerts}, \bibinfo{person}{Juan Reutter}, {and} \bibinfo{person}{Maksimilian Ryschkov}.} \bibinfo{year}{2021}\natexlab{}.
\newblock \showarticletitle{Graph neural networks with local graph parameters}. In \bibinfo{booktitle}{\emph{Advances in Neural Information Processing Systems 34 (NeurIPS 2021)}}. \bibinfo{publisher}{Curran Associates, Inc.}, \bibinfo{address}{Red Hook, NY, USA}, \bibinfo{pages}{25280--25293}.
\newblock


\bibitem[Barcel{\'o} et~al\mbox{.}(2022)]%
        {DBLP:journals/tcs/BarceloHPS22}
\bibfield{author}{\bibinfo{person}{Pablo Barcel{\'o}}, \bibinfo{person}{Nelson Higuera}, \bibinfo{person}{Jorge P{\'e}rez}, {and} \bibinfo{person}{Bernardo Subercaseaux}.} \bibinfo{year}{2022}\natexlab{}.
\newblock \showarticletitle{On the expressiveness of {Lara}: {A} proposal for unifying linear and relational algebra}.
\newblock \bibinfo{journal}{\emph{Theoretical Computer Science}}  \bibinfo{volume}{935} (\bibinfo{year}{2022}), \bibinfo{pages}{105--127}.
\newblock
\href{https://doi.org/10.1016/j.tcs.2022.09.003}{doi:\nolinkurl{10.1016/j.tcs.2022.09.003}}


\bibitem[Barrington et~al\mbox{.}(1990)]%
        {barrington1990uniformity}
\bibfield{author}{\bibinfo{person}{David A.~Mix Barrington}, \bibinfo{person}{Neil Immerman}, {and} \bibinfo{person}{Howard Straubing}.} \bibinfo{year}{1990}\natexlab{}.
\newblock \showarticletitle{On uniformity within {NC}$^1$}.
\newblock \bibinfo{journal}{\emph{J. Comput. System Sci.}} \bibinfo{volume}{41}, \bibinfo{number}{3} (\bibinfo{year}{1990}), \bibinfo{pages}{274--306}.
\newblock
\href{https://doi.org/10.1016/0022-0000(90)90022-D}{doi:\nolinkurl{10.1016/0022-0000(90)90022-D}}


\bibitem[Battaglia et~al\mbox{.}(2018)]%
        {DBLP:journals/corr/abs-1806-01261}
\bibfield{author}{\bibinfo{person}{Peter~W. Battaglia}, \bibinfo{person}{Jessica~B. Hamrick}, \bibinfo{person}{Victor Bapst}, \bibinfo{person}{Alvaro Sanchez{-}Gonzalez}, \bibinfo{person}{Vin{\'{\i}}cius~Flores Zambaldi}, \bibinfo{person}{Mateusz Malinowski}, \bibinfo{person}{Andrea Tacchetti}, \bibinfo{person}{David Raposo}, \bibinfo{person}{Adam Santoro}, \bibinfo{person}{Ryan Faulkner}, \bibinfo{person}{{\c{C}}aglar G{\"{u}}l{\c{c}}ehre}, \bibinfo{person}{H.~Francis Song}, \bibinfo{person}{Andrew~J. Ballard}, \bibinfo{person}{Justin Gilmer}, \bibinfo{person}{George~E. Dahl}, \bibinfo{person}{Ashish Vaswani}, \bibinfo{person}{Kelsey~R. Allen}, \bibinfo{person}{Charles Nash}, \bibinfo{person}{Victoria Langston}, \bibinfo{person}{Chris Dyer}, \bibinfo{person}{Nicolas Heess}, \bibinfo{person}{Daan Wierstra}, \bibinfo{person}{Pushmeet Kohli}, \bibinfo{person}{Matthew~M. Botvinick}, \bibinfo{person}{Oriol Vinyals}, \bibinfo{person}{Yujia Li}, {and} \bibinfo{person}{Razvan Pascanu}.}
  \bibinfo{year}{2018}\natexlab{}.
\newblock \bibinfo{title}{Relational inductive biases, deep learning, and graph networks}.
\newblock
\showeprint[arxiv]{1806.01261}~[cs.LG]


\bibitem[Bhuyan et~al\mbox{.}(2024)]%
        {10.1007/s00521-024-09960-z}
\bibfield{author}{\bibinfo{person}{Bikram~Pratim Bhuyan}, \bibinfo{person}{Amar Ramdane{-}Cherif}, \bibinfo{person}{Ravi Tomar}, {and} \bibinfo{person}{T.~P. Singh}.} \bibinfo{year}{2024}\natexlab{}.
\newblock \showarticletitle{Neuro-symbolic artificial intelligence: a survey}.
\newblock \bibinfo{journal}{\emph{Neural Computing and Applications}} \bibinfo{volume}{36}, \bibinfo{number}{21} (\bibinfo{year}{2024}), \bibinfo{pages}{12809--12844}.
\newblock
\href{https://doi.org/10.1007/s00521-024-09960-z}{doi:\nolinkurl{10.1007/s00521-024-09960-z}}


\bibitem[Bonald et~al\mbox{.}(2020)]%
        {DBLP:journals/jmlr/BonaldLLC20}
\bibfield{author}{\bibinfo{person}{Thomas Bonald}, \bibinfo{person}{Nathan de Lara}, \bibinfo{person}{Quentin Lutz}, {and} \bibinfo{person}{Bertrand Charpentier}.} \bibinfo{year}{2020}\natexlab{}.
\newblock \showarticletitle{Scikit-network: Graph Analysis in Python}.
\newblock \bibinfo{journal}{\emph{Journal of Machine Learning Research}} \bibinfo{volume}{21}, \bibinfo{number}{185} (\bibinfo{year}{2020}), \bibinfo{pages}{1--6}.
\newblock


\bibitem[Bordawekar and Shmueli(2017)]%
        {DBLP:conf/sigmod/BordawekarS17}
\bibfield{author}{\bibinfo{person}{Rajesh Bordawekar} {and} \bibinfo{person}{Oded Shmueli}.} \bibinfo{year}{2017}\natexlab{}.
\newblock \showarticletitle{Using Word Embedding to Enable Semantic Queries in Relational Databases}. In \bibinfo{booktitle}{\emph{Proceedings of the 1st Workshop on Data Management for End-to-End Machine Learning (DEEM 2017)}}. \bibinfo{publisher}{ACM}, \bibinfo{address}{New York, NY, USA}, \bibinfo{pages}{5:1--5:4}.
\newblock
\href{https://doi.org/10.1145/3076246.3076251}{doi:\nolinkurl{10.1145/3076246.3076251}}


\bibitem[Cappuzzo et~al\mbox{.}(2020)]%
        {DBLP:conf/sigmod/CappuzzoPT20}
\bibfield{author}{\bibinfo{person}{Riccardo Cappuzzo}, \bibinfo{person}{Paolo Papotti}, {and} \bibinfo{person}{Saravanan Thirumuruganathan}.} \bibinfo{year}{2020}\natexlab{}.
\newblock \showarticletitle{Creating Embeddings of Heterogeneous Relational Datasets for Data Integration Tasks}. In \bibinfo{booktitle}{\emph{Proceedings of the 2020 ACM SIGMOD International Conference on Management of Data (SIGMOD 2020)}}. \bibinfo{publisher}{ACM}, \bibinfo{address}{New York, NY, USA}, \bibinfo{pages}{1335--1349}.
\newblock
\href{https://doi.org/10.1145/3318464.3389742}{doi:\nolinkurl{10.1145/3318464.3389742}}


\bibitem[Chandra et~al\mbox{.}(1984)]%
        {chandra1984constant}
\bibfield{author}{\bibinfo{person}{Ashok~K. Chandra}, \bibinfo{person}{Larry Stockmeyer}, {and} \bibinfo{person}{Uzi Vishkin}.} \bibinfo{year}{1984}\natexlab{}.
\newblock \showarticletitle{Constant depth reducibility}.
\newblock \bibinfo{journal}{\emph{SIAM J. Comput.}} \bibinfo{volume}{13}, \bibinfo{number}{2} (\bibinfo{year}{1984}), \bibinfo{pages}{423--439}.
\newblock
\href{https://doi.org/10.1137/0213028}{doi:\nolinkurl{10.1137/0213028}}


\bibitem[Chen et~al\mbox{.}(2025)]%
        {Chen2025}
\bibfield{author}{\bibinfo{person}{Yijia Chen}, \bibinfo{person}{J{\"o}rg Flum}, \bibinfo{person}{Mingjun Liu}, {and} \bibinfo{person}{Zhiyang Xun}.} \bibinfo{year}{2025}\natexlab{}.
\newblock \showarticletitle{On algorithms based on finitely many homomorphism counts}.
\newblock \bibinfo{journal}{\emph{Information and Computation}}  \bibinfo{volume}{306} (\bibinfo{year}{2025}), \bibinfo{pages}{105326}.
\newblock
\href{https://doi.org/10.1016/j.ic.2025.105326}{doi:\nolinkurl{10.1016/j.ic.2025.105326}}


\bibitem[Curticapean et~al\mbox{.}(2017)]%
        {curticapean2017homomorphisms}
\bibfield{author}{\bibinfo{person}{Radu Curticapean}, \bibinfo{person}{Holger Dell}, {and} \bibinfo{person}{D{\'a}niel Marx}.} \bibinfo{year}{2017}\natexlab{}.
\newblock \showarticletitle{Homomorphisms are a good basis for counting small subgraphs}. In \bibinfo{booktitle}{\emph{Proceedings of the 49th Annual ACM SIGACT Symposium on Theory of Computing (STOC 2017)}}. \bibinfo{publisher}{ACM}, \bibinfo{address}{New York, NY, USA}, \bibinfo{pages}{210--223}.
\newblock
\href{https://doi.org/10.1145/3055399.3055502}{doi:\nolinkurl{10.1145/3055399.3055502}}


\bibitem[Cvetkov{-}Iliev et~al\mbox{.}(2023)]%
        {DBLP:journals/ml/CvetkovIlievAV23}
\bibfield{author}{\bibinfo{person}{Alexis Cvetkov{-}Iliev}, \bibinfo{person}{Alexandre Allauzen}, {and} \bibinfo{person}{Ga{\"e}l Varoquaux}.} \bibinfo{year}{2023}\natexlab{}.
\newblock \showarticletitle{Relational data embeddings for feature enrichment with background information}.
\newblock \bibinfo{journal}{\emph{Machine Learning}} \bibinfo{volume}{112}, \bibinfo{number}{2} (\bibinfo{year}{2023}), \bibinfo{pages}{687--720}.
\newblock
\href{https://doi.org/10.1007/s10994-022-06277-7}{doi:\nolinkurl{10.1007/s10994-022-06277-7}}


\bibitem[Dalvi and Suciu(2004)]%
        {DBLP:conf/vldb/DalviS04}
\bibfield{author}{\bibinfo{person}{Nilesh~N. Dalvi} {and} \bibinfo{person}{Dan Suciu}.} \bibinfo{year}{2004}\natexlab{}.
\newblock \showarticletitle{Efficient Query Evaluation on Probabilistic Databases}. In \bibinfo{booktitle}{\emph{Proceedings of the 30th International Conference on Very Large Data Bases (VLDB 2004)}}. \bibinfo{publisher}{Morgan Kaufmann}, \bibinfo{address}{San Francisco, CA, USA}, \bibinfo{pages}{864--875}.
\newblock


\bibitem[Deutch et~al\mbox{.}(2018)]%
        {DBLP:journals/vldb/DeutchGM18}
\bibfield{author}{\bibinfo{person}{Daniel Deutch}, \bibinfo{person}{Amir Gilad}, {and} \bibinfo{person}{Yuval Moskovitch}.} \bibinfo{year}{2018}\natexlab{}.
\newblock \showarticletitle{Efficient provenance tracking for Datalog using top-k queries}.
\newblock \bibinfo{journal}{\emph{The VLDB Journal}} \bibinfo{volume}{27}, \bibinfo{number}{2} (\bibinfo{year}{2018}), \bibinfo{pages}{245--269}.
\newblock
\href{https://doi.org/10.1007/s00778-018-0496-7}{doi:\nolinkurl{10.1007/s00778-018-0496-7}}


\bibitem[Diestel(2025)]%
        {Diestel2025}
\bibfield{author}{\bibinfo{person}{Reinhard Diestel}.} \bibinfo{year}{2025}\natexlab{}.
\newblock \bibinfo{booktitle}{\emph{Graph Theory} (\bibinfo{edition}{6th} ed.)}. \bibinfo{series}{Graduate Texts in Mathematics}, Vol.~\bibinfo{volume}{173}.
\newblock \bibinfo{publisher}{Springer}, \bibinfo{address}{Berlin, Heidelberg}.
\newblock
\showISBNx{978-3-662-70107-2}
\href{https://doi.org/10.1007/978-3-662-70107-2}{doi:\nolinkurl{10.1007/978-3-662-70107-2}}


\bibitem[Domingos(2025)]%
        {domingos2025tensorlogic}
\bibfield{author}{\bibinfo{person}{Pedro Domingos}.} \bibinfo{year}{2025}\natexlab{}.
\newblock \bibinfo{title}{Tensor Logic: The Language of {AI}}.
\newblock
\showeprint[arxiv]{2510.12269}~[cs.AI]


\bibitem[Dwivedi and Bresson(2021)]%
        {dwivedi2021generalization}
\bibfield{author}{\bibinfo{person}{Vijay~Prakash Dwivedi} {and} \bibinfo{person}{Xavier Bresson}.} \bibinfo{year}{2021}\natexlab{}.
\newblock \showarticletitle{A Generalization of Transformer Networks to Graphs}. In \bibinfo{booktitle}{\emph{AAAI 2021 Workshop on Deep Learning on Graphs: Methods and Applications (DLG-AAAI 2021)}}.
\newblock


\bibitem[Fey et~al\mbox{.}(2024)]%
        {DBLP:conf/icml/FeyHHLR0YYL24}
\bibfield{author}{\bibinfo{person}{Matthias Fey}, \bibinfo{person}{Weihua Hu}, \bibinfo{person}{Kexin Huang}, \bibinfo{person}{Jan~Eric Lenssen}, \bibinfo{person}{Rishabh Ranjan}, \bibinfo{person}{Joshua Robinson}, \bibinfo{person}{Rex Ying}, \bibinfo{person}{Jiaxuan You}, {and} \bibinfo{person}{Jure Leskovec}.} \bibinfo{year}{2024}\natexlab{}.
\newblock \showarticletitle{Position: Relational Deep Learning -- Graph Representation Learning on Relational Databases}. In \bibinfo{booktitle}{\emph{Proceedings of the 41st International Conference on Machine Learning (ICML 2024)}} \emph{(\bibinfo{series}{Proceedings of Machine Learning Research}, Vol.~\bibinfo{volume}{235})}. \bibinfo{publisher}{PMLR}, \bibinfo{pages}{13592--13607}.
\newblock


\bibitem[Fey and Lenssen(2019)]%
        {fey2019fastgraphrepresentationlearning}
\bibfield{author}{\bibinfo{person}{Matthias Fey} {and} \bibinfo{person}{Jan~Eric Lenssen}.} \bibinfo{year}{2019}\natexlab{}.
\newblock \showarticletitle{Fast Graph Representation Learning with {PyTorch Geometric}}. In \bibinfo{booktitle}{\emph{ICLR 2019 Workshop on Representation Learning on Graphs and Manifolds}}.
\newblock


\bibitem[Geerts(2023)]%
        {DBLP:conf/pods/Geerts23}
\bibfield{author}{\bibinfo{person}{Floris Geerts}.} \bibinfo{year}{2023}\natexlab{}.
\newblock \showarticletitle{A Query Language Perspective on Graph Learning}. In \bibinfo{booktitle}{\emph{Proceedings of the 42nd ACM SIGMOD-SIGACT-SIGAI Symposium on Principles of Database Systems (PODS 2023)}}. \bibinfo{publisher}{ACM}, \bibinfo{address}{New York, NY, USA}, \bibinfo{pages}{373--379}.
\newblock


\bibitem[Geerts and Reutter(2022)]%
        {DBLP:conf/iclr/GeertsR22}
\bibfield{author}{\bibinfo{person}{Floris Geerts} {and} \bibinfo{person}{Juan~L. Reutter}.} \bibinfo{year}{2022}\natexlab{}.
\newblock \showarticletitle{Expressiveness and Approximation Properties of Graph Neural Networks}. In \bibinfo{booktitle}{\emph{10th International Conference on Learning Representations (ICLR 2022)}}. \bibinfo{publisher}{OpenReview.net}.
\newblock


\bibitem[Gilmer et~al\mbox{.}(2017)]%
        {DBLP:conf/icml/GilmerSRVD17}
\bibfield{author}{\bibinfo{person}{Justin Gilmer}, \bibinfo{person}{Samuel~S. Schoenholz}, \bibinfo{person}{Patrick~F. Riley}, \bibinfo{person}{Oriol Vinyals}, {and} \bibinfo{person}{George~E. Dahl}.} \bibinfo{year}{2017}\natexlab{}.
\newblock \showarticletitle{Neural Message Passing for Quantum Chemistry}. In \bibinfo{booktitle}{\emph{Proceedings of the 34th International Conference on Machine Learning (ICML 2017)}} \emph{(\bibinfo{series}{Proceedings of Machine Learning Research}, Vol.~\bibinfo{volume}{70})}. \bibinfo{publisher}{PMLR}, \bibinfo{pages}{1263--1272}.
\newblock
\urldef\tempurl%
\url{http://proceedings.mlr.press/v70/gilmer17a.html}
\showURL{%
\tempurl}


\bibitem[Gr{\"a}del and Otto(1992)]%
        {gradel1992inductive}
\bibfield{author}{\bibinfo{person}{Erich Gr{\"a}del} {and} \bibinfo{person}{Martin Otto}.} \bibinfo{year}{1992}\natexlab{}.
\newblock \showarticletitle{Inductive definability with counting on finite structures}. In \bibinfo{booktitle}{\emph{Computer Science Logic, 6th Workshop (CSL '92)}} \emph{(\bibinfo{series}{Lecture Notes in Computer Science}, Vol.~\bibinfo{volume}{702})}. \bibinfo{publisher}{Springer}, \bibinfo{address}{Berlin, Heidelberg}, \bibinfo{pages}{231--247}.
\newblock
\href{https://doi.org/10.1007/3-540-56992-8_15}{doi:\nolinkurl{10.1007/3-540-56992-8_15}}


\bibitem[Graham et~al\mbox{.}(1994)]%
        {Graham1990}
\bibfield{author}{\bibinfo{person}{Ronald~L. Graham}, \bibinfo{person}{Donald~E. Knuth}, {and} \bibinfo{person}{Oren Patashnik}.} \bibinfo{year}{1994}\natexlab{}.
\newblock \bibinfo{booktitle}{\emph{Concrete Mathematics: A Foundation for Computer Science} (\bibinfo{edition}{2nd} ed.)}.
\newblock \bibinfo{publisher}{Addison-Wesley}, \bibinfo{address}{Reading, MA, USA}.
\newblock
\showISBNx{0-201-55802-5}


\bibitem[Green et~al\mbox{.}(2007)]%
        {DBLP:conf/pods/GreenKT07}
\bibfield{author}{\bibinfo{person}{Todd~J. Green}, \bibinfo{person}{Gregory Karvounarakis}, {and} \bibinfo{person}{Val Tannen}.} \bibinfo{year}{2007}\natexlab{}.
\newblock \showarticletitle{Provenance semirings}. In \bibinfo{booktitle}{\emph{Proceedings of the 26th ACM SIGMOD-SIGACT-SIGART Symposium on Principles of Database Systems (PODS 2007)}}. \bibinfo{publisher}{ACM}, \bibinfo{address}{New York, NY, USA}, \bibinfo{pages}{31--40}.
\newblock
\href{https://doi.org/10.1145/1265530.1265535}{doi:\nolinkurl{10.1145/1265530.1265535}}


\bibitem[Grohe(2024)]%
        {grohe2024descriptive}
\bibfield{author}{\bibinfo{person}{Martin Grohe}.} \bibinfo{year}{2024}\natexlab{}.
\newblock \showarticletitle{The descriptive complexity of graph neural networks}.
\newblock \bibinfo{journal}{\emph{TheoretiCS}} \bibinfo{volume}{3}, \bibinfo{number}{25} (\bibinfo{year}{2024}), \bibinfo{pages}{1--93}.
\newblock
\href{https://doi.org/10.46298/theoretics.24.25}{doi:\nolinkurl{10.46298/theoretics.24.25}}


\bibitem[Grohe et~al\mbox{.}(2020)]%
        {DBLP:conf/pods/GroheKKL20}
\bibfield{author}{\bibinfo{person}{Martin Grohe}, \bibinfo{person}{Benjamin~Lucien Kaminski}, \bibinfo{person}{Joost{-}Pieter Katoen}, {and} \bibinfo{person}{Peter Lindner}.} \bibinfo{year}{2020}\natexlab{}.
\newblock \showarticletitle{Generative Datalog with Continuous Distributions}. In \bibinfo{booktitle}{\emph{Proceedings of the 39th ACM SIGMOD-SIGACT-SIGAI Symposium on Principles of Database Systems (PODS 2020)}}. \bibinfo{publisher}{ACM}, \bibinfo{address}{New York, NY, USA}, \bibinfo{pages}{347--360}.
\newblock
\href{https://doi.org/10.1145/3375395.3387659}{doi:\nolinkurl{10.1145/3375395.3387659}}


\bibitem[Grohe et~al\mbox{.}(2025)]%
        {grohe2025querylanguages}
\bibfield{author}{\bibinfo{person}{Martin Grohe}, \bibinfo{person}{Christoph Standke}, \bibinfo{person}{Juno Steegmans}, {and} \bibinfo{person}{Jan Van~den Bussche}.} \bibinfo{year}{2025}\natexlab{}.
\newblock \showarticletitle{Query Languages for Neural Networks}. In \bibinfo{booktitle}{\emph{28th International Conference on Database Theory (ICDT 2025)}} \emph{(\bibinfo{series}{Leibniz International Proceedings in Informatics (LIPIcs)}, Vol.~\bibinfo{volume}{328})}, \bibfield{editor}{\bibinfo{person}{Sudeepa Roy} {and} \bibinfo{person}{Ahmet Kara}} (Eds.). \bibinfo{publisher}{Schloss Dagstuhl -- Leibniz-Zentrum f{\"u}r Informatik}, \bibinfo{address}{Dagstuhl, Germany}, \bibinfo{pages}{9:1--9:18}.
\newblock
\href{https://doi.org/10.4230/LIPIcs.ICDT.2025.9}{doi:\nolinkurl{10.4230/LIPIcs.ICDT.2025.9}}


\bibitem[Grohe et~al\mbox{.}(2026)]%
        {grohe2026recursive}
\bibfield{author}{\bibinfo{person}{Martin Grohe}, \bibinfo{person}{Christoph Standke}, \bibinfo{person}{Juno Steegmans}, {and} \bibinfo{person}{Jan Van~den Bussche}.} \bibinfo{year}{2026}\natexlab{}.
\newblock \showarticletitle{Recursive Querying of Neural Networks via Weighted Structures}.
\newblock \bibinfo{journal}{\emph{Proc. ACM Manag. Data}} \bibinfo{volume}{4}, \bibinfo{number}{2}, Article \bibinfo{articleno}{115} (\bibinfo{date}{May} \bibinfo{year}{2026}), \bibinfo{numpages}{20}~pages.
\newblock
\href{https://doi.org/10.1145/3801911}{doi:\nolinkurl{10.1145/3801911}}


\bibitem[Hesse et~al\mbox{.}(2002)]%
        {hesse2002uniform}
\bibfield{author}{\bibinfo{person}{William Hesse}, \bibinfo{person}{Eric Allender}, {and} \bibinfo{person}{David~A. {Mix Barrington}}.} \bibinfo{year}{2002}\natexlab{}.
\newblock \showarticletitle{Uniform constant-depth threshold circuits for division and iterated multiplication}.
\newblock \bibinfo{journal}{\emph{J. Comput. System Sci.}} \bibinfo{volume}{65}, \bibinfo{number}{4} (\bibinfo{year}{2002}), \bibinfo{pages}{695--716}.
\newblock
\href{https://doi.org/10.1016/S0022-0000(02)00025-9}{doi:\nolinkurl{10.1016/S0022-0000(02)00025-9}}


\bibitem[Hitzler and Sarker(2021)]%
        {DBLP:series/faia/342}
\bibfield{editor}{\bibinfo{person}{Pascal Hitzler} {and} \bibinfo{person}{Md.~Kamruzzaman Sarker}} (Eds.). \bibinfo{year}{2021}\natexlab{}.
\newblock \bibinfo{booktitle}{\emph{Neuro-Symbolic Artificial Intelligence: The State of the Art}}. \bibinfo{series}{Frontiers in Artificial Intelligence and Applications}, Vol.~\bibinfo{volume}{342}.
\newblock \bibinfo{publisher}{IOS Press}, \bibinfo{address}{Amsterdam, The Netherlands}.
\newblock


\bibitem[Im et~al\mbox{.}(2024)]%
        {DBLP:journals/pacmmod/ImMNP24}
\bibfield{author}{\bibinfo{person}{Sungjin Im}, \bibinfo{person}{Benjamin Moseley}, \bibinfo{person}{Hung~Q. Ngo}, {and} \bibinfo{person}{Kirk Pruhs}.} \bibinfo{year}{2024}\natexlab{}.
\newblock \showarticletitle{Polynomial Time Convergence of the Iterative Evaluation of Datalogo Programs}.
\newblock \bibinfo{journal}{\emph{Proc. ACM Manag. Data}} \bibinfo{volume}{2}, \bibinfo{number}{5} (\bibinfo{year}{2024}), \bibinfo{pages}{221:1--221:19}.
\newblock


\bibitem[Immerman(1983)]%
        {immerman1983languages}
\bibfield{author}{\bibinfo{person}{Neil Immerman}.} \bibinfo{year}{1983}\natexlab{}.
\newblock \showarticletitle{Languages which capture complexity classes}. In \bibinfo{booktitle}{\emph{Proceedings of the 15th Annual ACM Symposium on Theory of Computing (STOC '83)}}. \bibinfo{publisher}{ACM}, \bibinfo{address}{New York, NY, USA}, \bibinfo{pages}{347--354}.
\newblock
\href{https://doi.org/10.1145/800061.808765}{doi:\nolinkurl{10.1145/800061.808765}}


\bibitem[Immerman(1987)]%
        {immerman1987expressibility}
\bibfield{author}{\bibinfo{person}{Neil Immerman}.} \bibinfo{year}{1987}\natexlab{}.
\newblock \showarticletitle{Expressibility as a complexity measure: results and directions}. In \bibinfo{booktitle}{\emph{Proceedings of the 2nd Annual Conference on Structure in Complexity Theory}}. \bibinfo{publisher}{IEEE Computer Society}, \bibinfo{address}{Los Alamitos, CA, USA}, \bibinfo{pages}{194--202}.
\newblock


\bibitem[Immerman(1999)]%
        {DBLP:books/daglib/0095988}
\bibfield{author}{\bibinfo{person}{Neil Immerman}.} \bibinfo{year}{1999}\natexlab{}.
\newblock \bibinfo{booktitle}{\emph{Descriptive complexity}}.
\newblock \bibinfo{publisher}{Springer}, \bibinfo{address}{New York, NY, USA}.
\newblock
\showISBNx{978-1-4612-6809-3}
\href{https://doi.org/10.1007/978-1-4612-0539-5}{doi:\nolinkurl{10.1007/978-1-4612-0539-5}}


\bibitem[Irmai et~al\mbox{.}(2026)]%
        {irmai2026triangle}
\bibfield{author}{\bibinfo{person}{Jannik Irmai}, \bibinfo{person}{Lucas~Fabian Naumann}, {and} \bibinfo{person}{Bjoern Andres}.} \bibinfo{year}{2026}\natexlab{}.
\newblock \bibinfo{title}{Graph Neural Networks with Triangle-Based Messages for the Multicut Problem}.
\newblock
\showeprint[arxiv]{2605.13673}~[cs.LG]


\bibitem[Jin et~al\mbox{.}(2024)]%
        {DBLP:conf/icml/JinBCL24}
\bibfield{author}{\bibinfo{person}{Emily Jin}, \bibinfo{person}{Michael~M. Bronstein}, \bibinfo{person}{{\.I}smail~{\.I}lkan Ceylan}, {and} \bibinfo{person}{Matthias Lanzinger}.} \bibinfo{year}{2024}\natexlab{}.
\newblock \showarticletitle{Homomorphism Counts for Graph Neural Networks: All About That Basis}. In \bibinfo{booktitle}{\emph{Proceedings of the 41st International Conference on Machine Learning (ICML 2024)}} \emph{(\bibinfo{series}{Proceedings of Machine Learning Research}, Vol.~\bibinfo{volume}{235})}. \bibinfo{publisher}{PMLR}, \bibinfo{pages}{22075--22098}.
\newblock
\urldef\tempurl%
\url{https://proceedings.mlr.press/v235/jin24a.html}
\showURL{%
\tempurl}


\bibitem[Kim et~al\mbox{.}(2024)]%
        {DBLP:conf/icml/KimGV24}
\bibfield{author}{\bibinfo{person}{Myung~Jun Kim}, \bibinfo{person}{L{\'e}o Grinsztajn}, {and} \bibinfo{person}{Ga{\"e}l Varoquaux}.} \bibinfo{year}{2024}\natexlab{}.
\newblock \showarticletitle{{CARTE}: Pretraining and Transfer for Tabular Learning}. In \bibinfo{booktitle}{\emph{Proceedings of the 41st International Conference on Machine Learning (ICML 2024)}} \emph{(\bibinfo{series}{Proceedings of Machine Learning Research}, Vol.~\bibinfo{volume}{235})}. \bibinfo{publisher}{PMLR}, \bibinfo{pages}{23843--23866}.
\newblock
\urldef\tempurl%
\url{https://proceedings.mlr.press/v235/kim24d.html}
\showURL{%
\tempurl}


\bibitem[Kuske and Schweikardt(2017)]%
        {kuske2017first}
\bibfield{author}{\bibinfo{person}{Dietrich Kuske} {and} \bibinfo{person}{Nicole Schweikardt}.} \bibinfo{year}{2017}\natexlab{}.
\newblock \showarticletitle{First-order logic with counting}. In \bibinfo{booktitle}{\emph{32nd Annual ACM/IEEE Symposium on Logic in Computer Science (LICS 2017)}}. \bibinfo{publisher}{IEEE Computer Society}, \bibinfo{address}{Los Alamitos, CA, USA}, \bibinfo{pages}{1--12}.
\newblock
\href{https://doi.org/10.1109/LICS.2017.8005133}{doi:\nolinkurl{10.1109/LICS.2017.8005133}}


\bibitem[Lewis et~al\mbox{.}(2020)]%
        {DBLP:conf/nips/LewisPPPKGKLYR020}
\bibfield{author}{\bibinfo{person}{Patrick Lewis}, \bibinfo{person}{Ethan Perez}, \bibinfo{person}{Aleksandra Piktus}, \bibinfo{person}{Fabio Petroni}, \bibinfo{person}{Vladimir Karpukhin}, \bibinfo{person}{Naman Goyal}, \bibinfo{person}{Heinrich K{\"u}ttler}, \bibinfo{person}{Mike Lewis}, \bibinfo{person}{Wen{-}tau Yih}, \bibinfo{person}{Tim Rockt{\"a}schel}, \bibinfo{person}{Sebastian Riedel}, {and} \bibinfo{person}{Douwe Kiela}.} \bibinfo{year}{2020}\natexlab{}.
\newblock \showarticletitle{Retrieval-Augmented Generation for Knowledge-Intensive {NLP} Tasks}. In \bibinfo{booktitle}{\emph{Advances in Neural Information Processing Systems 33 (NeurIPS 2020)}}. \bibinfo{publisher}{Curran Associates, Inc.}, \bibinfo{address}{Red Hook, NY, USA}, \bibinfo{pages}{9459--9474}.
\newblock


\bibitem[Li et~al\mbox{.}(2023)]%
        {DBLP:journals/pacmpl/LiHN23}
\bibfield{author}{\bibinfo{person}{Ziyang Li}, \bibinfo{person}{Jiani Huang}, {and} \bibinfo{person}{Mayur Naik}.} \bibinfo{year}{2023}\natexlab{}.
\newblock \showarticletitle{Scallop: {A} Language for Neurosymbolic Programming}.
\newblock \bibinfo{journal}{\emph{Proc. ACM Program. Lang.}} \bibinfo{volume}{7}, \bibinfo{number}{PLDI} (\bibinfo{year}{2023}), \bibinfo{pages}{1463--1487}.
\newblock
\href{https://doi.org/10.1145/3591280}{doi:\nolinkurl{10.1145/3591280}}


\bibitem[Lov{\'a}sz(2012)]%
        {Lovasz2012}
\bibfield{author}{\bibinfo{person}{L{\'a}szl{\'o} Lov{\'a}sz}.} \bibinfo{year}{2012}\natexlab{}.
\newblock \bibinfo{booktitle}{\emph{Large Networks and Graph Limits}}. \bibinfo{series}{Colloquium Publications}, Vol.~\bibinfo{volume}{60}.
\newblock \bibinfo{publisher}{American Mathematical Society}, \bibinfo{address}{Providence, RI, USA}.
\newblock


\bibitem[Lubarsky et~al\mbox{.}(2026)]%
        {lubarsky2026incorporating}
\bibfield{author}{\bibinfo{person}{Yuval~Lev Lubarsky}, \bibinfo{person}{Dean Light}, \bibinfo{person}{Boaz Berger}, \bibinfo{person}{Shunit Agmon}, {and} \bibinfo{person}{Benny Kimelfeld}.} \bibinfo{year}{2026}\natexlab{}.
\newblock \bibinfo{title}{Incorporating Deep Learning Design in Database Queries}.
\newblock
\showeprint[arxiv]{2605.24207}~[cs.DB]


\bibitem[Lubarsky et~al\mbox{.}(2023)]%
        {DBLP:conf/cikm/LubarskyTGK23}
\bibfield{author}{\bibinfo{person}{Yuval~Lev Lubarsky}, \bibinfo{person}{Jan T{\"o}nshoff}, \bibinfo{person}{Martin Grohe}, {and} \bibinfo{person}{Benny Kimelfeld}.} \bibinfo{year}{2023}\natexlab{}.
\newblock \showarticletitle{Selecting Walk Schemes for Database Embedding}. In \bibinfo{booktitle}{\emph{Proceedings of the 32nd ACM International Conference on Information and Knowledge Management (CIKM 2023)}}. \bibinfo{publisher}{ACM}, \bibinfo{address}{New York, NY, USA}, \bibinfo{pages}{1677--1686}.
\newblock


\bibitem[Luo and Virtema(2026)]%
        {Luo2026}
\bibfield{author}{\bibinfo{person}{Huan Luo} {and} \bibinfo{person}{Jonni Virtema}.} \bibinfo{year}{2026}\natexlab{}.
\newblock \showarticletitle{Unifying Approach to Uniform Expressivity of Graph Neural Networks}. In \bibinfo{booktitle}{\emph{Proceedings of the 23rd International Conference on Principles of Knowledge Representation and Reasoning (KR 2026)}}. \bibinfo{publisher}{IJCAI Organization}, \bibinfo{pages}{1030--1040}.
\newblock
\href{https://doi.org/10.24963/kr.2026/96}{doi:\nolinkurl{10.24963/kr.2026/96}}


\bibitem[Maehara and NT(2024)]%
        {maehara2024deep}
\bibfield{author}{\bibinfo{person}{Takanori Maehara} {and} \bibinfo{person}{Hoang NT}.} \bibinfo{year}{2024}\natexlab{}.
\newblock \showarticletitle{Deep homomorphism networks}. In \bibinfo{booktitle}{\emph{Advances in Neural Information Processing Systems 37 (NeurIPS 2024)}}. \bibinfo{publisher}{Curran Associates, Inc.}, \bibinfo{address}{Red Hook, NY, USA}, \bibinfo{pages}{56076--56107}.
\newblock


\bibitem[Manhaeve et~al\mbox{.}(2018)]%
        {DBLP:conf/nips/ManhaeveDKDR18}
\bibfield{author}{\bibinfo{person}{Robin Manhaeve}, \bibinfo{person}{Sebastijan Dumancic}, \bibinfo{person}{Angelika Kimmig}, \bibinfo{person}{Thomas Demeester}, {and} \bibinfo{person}{Luc~De Raedt}.} \bibinfo{year}{2018}\natexlab{}.
\newblock \showarticletitle{{DeepProbLog}: Neural Probabilistic Logic Programming}. In \bibinfo{booktitle}{\emph{Advances in Neural Information Processing Systems 31 (NeurIPS 2018)}}. \bibinfo{publisher}{Curran Associates, Inc.}, \bibinfo{address}{Red Hook, NY, USA}, \bibinfo{pages}{3753--3763}.
\newblock


\bibitem[Morris et~al\mbox{.}(2019)]%
        {morris2019weisfeiler}
\bibfield{author}{\bibinfo{person}{Christopher Morris}, \bibinfo{person}{Martin Ritzert}, \bibinfo{person}{Matthias Fey}, \bibinfo{person}{William~L. Hamilton}, \bibinfo{person}{Jan~Eric Lenssen}, \bibinfo{person}{Gaurav Rattan}, {and} \bibinfo{person}{Martin Grohe}.} \bibinfo{year}{2019}\natexlab{}.
\newblock \showarticletitle{{Weisfeiler} and {Leman} go neural: Higher-order graph neural networks}. In \bibinfo{booktitle}{\emph{Proceedings of the 33rd AAAI Conference on Artificial Intelligence (AAAI 2019)}}. \bibinfo{publisher}{AAAI Press}, \bibinfo{address}{Washington, DC, USA}, \bibinfo{pages}{4602--4609}.
\newblock
\href{https://doi.org/10.1609/aaai.v33i01.33014602}{doi:\nolinkurl{10.1609/aaai.v33i01.33014602}}


\bibitem[Mudgal et~al\mbox{.}(2018)]%
        {DBLP:conf/sigmod/MudgalLRDPKDAR18}
\bibfield{author}{\bibinfo{person}{Sidharth Mudgal}, \bibinfo{person}{Han Li}, \bibinfo{person}{Theodoros Rekatsinas}, \bibinfo{person}{AnHai Doan}, \bibinfo{person}{Youngchoon Park}, \bibinfo{person}{Ganesh Krishnan}, \bibinfo{person}{Rohit Deep}, \bibinfo{person}{Esteban Arcaute}, {and} \bibinfo{person}{Vijay Raghavendra}.} \bibinfo{year}{2018}\natexlab{}.
\newblock \showarticletitle{Deep Learning for Entity Matching: {A} Design Space Exploration}. In \bibinfo{booktitle}{\emph{Proceedings of the 2018 International Conference on Management of Data (SIGMOD 2018)}}. \bibinfo{publisher}{ACM}, \bibinfo{address}{New York, NY, USA}, \bibinfo{pages}{19--34}.
\newblock
\href{https://doi.org/10.1145/3183713.3196926}{doi:\nolinkurl{10.1145/3183713.3196926}}


\bibitem[Pan et~al\mbox{.}(2024)]%
        {DBLP:journals/vldb/PanWL24}
\bibfield{author}{\bibinfo{person}{James~Jie Pan}, \bibinfo{person}{Jianguo Wang}, {and} \bibinfo{person}{Guoliang Li}.} \bibinfo{year}{2024}\natexlab{}.
\newblock \showarticletitle{Survey of vector database management systems}.
\newblock \bibinfo{journal}{\emph{The VLDB Journal}} \bibinfo{volume}{33}, \bibinfo{number}{5} (\bibinfo{year}{2024}), \bibinfo{pages}{1591--1615}.
\newblock
\href{https://doi.org/10.1007/s00778-024-00864-x}{doi:\nolinkurl{10.1007/s00778-024-00864-x}}


\bibitem[Raedt et~al\mbox{.}(2007)]%
        {DBLP:conf/ijcai/RaedtKT07}
\bibfield{author}{\bibinfo{person}{Luc~De Raedt}, \bibinfo{person}{Angelika Kimmig}, {and} \bibinfo{person}{Hannu Toivonen}.} \bibinfo{year}{2007}\natexlab{}.
\newblock \showarticletitle{{ProbLog}: {A} Probabilistic Prolog and Its Application in Link Discovery}. In \bibinfo{booktitle}{\emph{Proceedings of the 20th International Joint Conference on Artificial Intelligence (IJCAI 2007)}}. \bibinfo{publisher}{IJCAI}, \bibinfo{pages}{2462--2467}.
\newblock


\bibitem[Robinson et~al\mbox{.}(2024)]%
        {relbench}
\bibfield{author}{\bibinfo{person}{Joshua Robinson}, \bibinfo{person}{Rishabh Ranjan}, \bibinfo{person}{Weihua Hu}, \bibinfo{person}{Kexin Huang}, \bibinfo{person}{Jiaqi Han}, \bibinfo{person}{Alejandro Dobles}, \bibinfo{person}{Matthias Fey}, \bibinfo{person}{Jan~Eric Lenssen}, \bibinfo{person}{Yiwen Yuan}, \bibinfo{person}{Zecheng Zhang}, \bibinfo{person}{Xinwei He}, {and} \bibinfo{person}{Jure Leskovec}.} \bibinfo{year}{2024}\natexlab{}.
\newblock \showarticletitle{{RelBench}: {A} Benchmark for Deep Learning on Relational Databases}. In \bibinfo{booktitle}{\emph{Advances in Neural Information Processing Systems 37 (NeurIPS 2024), Datasets and Benchmarks Track}}. \bibinfo{publisher}{Curran Associates, Inc.}, \bibinfo{address}{Red Hook, NY, USA}, \bibinfo{pages}{21330--21341}.
\newblock
\urldef\tempurl%
\url{http://papers.nips.cc/paper_files/paper/2024/hash/25cd345233c65fac1fec0ce61d0f7836-Abstract-Datasets_and_Benchmarks_Track.html}
\showURL{%
\tempurl}


\bibitem[Sch{\"o}nherr et~al\mbox{.}(2026)]%
        {schonherr2026expressive}
\bibfield{author}{\bibinfo{person}{Moritz Sch{\"o}nherr}, \bibinfo{person}{Balder ten Cate}, \bibinfo{person}{Maurice Funk}, \bibinfo{person}{Benny Kimelfeld}, \bibinfo{person}{Carsten Lutz}, {and} \bibinfo{person}{Arie Soeteman}.} \bibinfo{year}{2026}\natexlab{}.
\newblock \bibinfo{title}{Expressive Power of Deep Homomorphism Networks over Relational Databases}.
\newblock
\showeprint[arxiv]{2605.22852}~[cs.DB]


\bibitem[Shah et~al\mbox{.}(2020)]%
        {DBLP:conf/nips/ShahZSVYC20}
\bibfield{author}{\bibinfo{person}{Ameesh Shah}, \bibinfo{person}{Eric Zhan}, \bibinfo{person}{Jennifer~J. Sun}, \bibinfo{person}{Abhinav Verma}, \bibinfo{person}{Yisong Yue}, {and} \bibinfo{person}{Swarat Chaudhuri}.} \bibinfo{year}{2020}\natexlab{}.
\newblock \showarticletitle{Learning Differentiable Programs with Admissible Neural Heuristics}. In \bibinfo{booktitle}{\emph{Advances in Neural Information Processing Systems 33 (NeurIPS 2020)}}. \bibinfo{publisher}{Curran Associates, Inc.}, \bibinfo{address}{Red Hook, NY, USA}.
\newblock


\bibitem[Silva and Gombolay(2021)]%
        {DBLP:conf/aaai/SilvaG21}
\bibfield{author}{\bibinfo{person}{Andrew Silva} {and} \bibinfo{person}{Matthew~C. Gombolay}.} \bibinfo{year}{2021}\natexlab{}.
\newblock \showarticletitle{Encoding Human Domain Knowledge to Warm Start Reinforcement Learning}. In \bibinfo{booktitle}{\emph{Proceedings of the 35th AAAI Conference on Artificial Intelligence (AAAI 2021)}}. \bibinfo{publisher}{AAAI Press}, \bibinfo{address}{Washington, DC, USA}, \bibinfo{pages}{5042--5050}.
\newblock


\bibitem[ten Cate et~al\mbox{.}(2024)]%
        {ten2024homomorphism}
\bibfield{author}{\bibinfo{person}{Balder ten Cate}, \bibinfo{person}{Victor Dalmau}, \bibinfo{person}{Phokion~G. Kolaitis}, {and} \bibinfo{person}{Wei{-}Lin Wu}.} \bibinfo{year}{2024}\natexlab{}.
\newblock \showarticletitle{When Do Homomorphism Counts Help in Query Algorithms?}. In \bibinfo{booktitle}{\emph{27th International Conference on Database Theory (ICDT 2024)}} \emph{(\bibinfo{series}{Leibniz International Proceedings in Informatics (LIPIcs)}, Vol.~\bibinfo{volume}{290})}, \bibfield{editor}{\bibinfo{person}{Graham Cormode} {and} \bibinfo{person}{Michael Shekelyan}} (Eds.). \bibinfo{publisher}{Schloss Dagstuhl -- Leibniz-Zentrum f{\"u}r Informatik}, \bibinfo{address}{Dagstuhl, Germany}, \bibinfo{pages}{8:1--8:20}.
\newblock
\href{https://doi.org/10.4230/LIPIcs.ICDT.2024.8}{doi:\nolinkurl{10.4230/LIPIcs.ICDT.2024.8}}


\bibitem[ten Cate et~al\mbox{.}(2025)]%
        {MFCS2025}
\bibfield{author}{\bibinfo{person}{Balder ten Cate}, \bibinfo{person}{Phokion~G. Kolaitis}, {and} \bibinfo{person}{Arnar~\'{A}. Kristj\'{a}nsson}.} \bibinfo{year}{2025}\natexlab{}.
\newblock \showarticletitle{Adaptive Query Algorithms for Relational Structures Based on Homomorphism Counts}. In \bibinfo{booktitle}{\emph{50th International Symposium on Mathematical Foundations of Computer Science (MFCS 2025)}} \emph{(\bibinfo{series}{Leibniz International Proceedings in Informatics (LIPIcs)}, Vol.~\bibinfo{volume}{345})}, \bibfield{editor}{\bibinfo{person}{Pawe{\l} Gawrychowski}, \bibinfo{person}{Filip Mazowiecki}, {and} \bibinfo{person}{Micha{\l} Skrzypczak}} (Eds.). \bibinfo{publisher}{Schloss Dagstuhl -- Leibniz-Zentrum f{\"u}r Informatik}, \bibinfo{address}{Dagstuhl, Germany}, \bibinfo{pages}{34:1--34:18}.
\newblock
\href{https://doi.org/10.4230/LIPIcs.MFCS.2025.34}{doi:\nolinkurl{10.4230/LIPIcs.MFCS.2025.34}}


\bibitem[T{\"o}nshoff et~al\mbox{.}(2023)]%
        {DBLP:conf/icde/TonshoffFGK23}
\bibfield{author}{\bibinfo{person}{Jan T{\"o}nshoff}, \bibinfo{person}{Neta Friedman}, \bibinfo{person}{Martin Grohe}, {and} \bibinfo{person}{Benny Kimelfeld}.} \bibinfo{year}{2023}\natexlab{}.
\newblock \showarticletitle{Stable Tuple Embeddings for Dynamic Databases}. In \bibinfo{booktitle}{\emph{39th IEEE International Conference on Data Engineering (ICDE 2023)}}. \bibinfo{publisher}{IEEE}, \bibinfo{address}{Piscataway, NJ, USA}, \bibinfo{pages}{1286--1299}.
\newblock
\href{https://doi.org/10.1109/ICDE55515.2023.00103}{doi:\nolinkurl{10.1109/ICDE55515.2023.00103}}


\bibitem[Wang et~al\mbox{.}(2019b)]%
        {DBLP:journals/corr/abs-1909-01315}
\bibfield{author}{\bibinfo{person}{Minjie Wang}, \bibinfo{person}{Da Zheng}, \bibinfo{person}{Zihao Ye}, \bibinfo{person}{Quan Gan}, \bibinfo{person}{Mufei Li}, \bibinfo{person}{Xiang Song}, \bibinfo{person}{Jinjing Zhou}, \bibinfo{person}{Chao Ma}, \bibinfo{person}{Lingfan Yu}, \bibinfo{person}{Yu Gai}, \bibinfo{person}{Tianjun Xiao}, \bibinfo{person}{Tong He}, \bibinfo{person}{George Karypis}, \bibinfo{person}{Jinyang Li}, {and} \bibinfo{person}{Zheng Zhang}.} \bibinfo{year}{2019}\natexlab{b}.
\newblock \bibinfo{title}{Deep Graph Library: Towards Efficient and Scalable Deep Learning on Graphs}.
\newblock
\showeprint[arxiv]{1909.01315}~[cs.LG]


\bibitem[Wang et~al\mbox{.}(2019a)]%
        {DBLP:conf/www/WangJSWYCY19}
\bibfield{author}{\bibinfo{person}{Xiao Wang}, \bibinfo{person}{Houye Ji}, \bibinfo{person}{Chuan Shi}, \bibinfo{person}{Bai Wang}, \bibinfo{person}{Yanfang Ye}, \bibinfo{person}{Peng Cui}, {and} \bibinfo{person}{Philip~S. Yu}.} \bibinfo{year}{2019}\natexlab{a}.
\newblock \showarticletitle{Heterogeneous Graph Attention Network}. In \bibinfo{booktitle}{\emph{Proceedings of the World Wide Web Conference (WWW 2019)}}. \bibinfo{publisher}{ACM}, \bibinfo{address}{New York, NY, USA}, \bibinfo{pages}{2022--2032}.
\newblock
\href{https://doi.org/10.1145/3308558.3313562}{doi:\nolinkurl{10.1145/3308558.3313562}}


\bibitem[Wang et~al\mbox{.}(2018)]%
        {DBLP:conf/cvpr/0004YG18}
\bibfield{author}{\bibinfo{person}{Xiaolong Wang}, \bibinfo{person}{Yufei Ye}, {and} \bibinfo{person}{Abhinav Gupta}.} \bibinfo{year}{2018}\natexlab{}.
\newblock \showarticletitle{Zero-Shot Recognition via Semantic Embeddings and Knowledge Graphs}. In \bibinfo{booktitle}{\emph{2018 IEEE/CVF Conference on Computer Vision and Pattern Recognition (CVPR 2018)}}. \bibinfo{publisher}{IEEE Computer Society}, \bibinfo{address}{Los Alamitos, CA, USA}, \bibinfo{pages}{6857--6866}.
\newblock
\href{https://doi.org/10.1109/CVPR.2018.00717}{doi:\nolinkurl{10.1109/CVPR.2018.00717}}


\bibitem[Wang et~al\mbox{.}(2025)]%
        {DBLP:conf/icml/WangWGWYWZ25}
\bibfield{author}{\bibinfo{person}{Yanbo Wang}, \bibinfo{person}{Xiyuan Wang}, \bibinfo{person}{Quan Gan}, \bibinfo{person}{Minjie Wang}, \bibinfo{person}{Qibin Yang}, \bibinfo{person}{David Wipf}, {and} \bibinfo{person}{Muhan Zhang}.} \bibinfo{year}{2025}\natexlab{}.
\newblock \showarticletitle{Griffin: Towards a Graph-Centric Relational Database Foundation Model}. In \bibinfo{booktitle}{\emph{Proceedings of the 42nd International Conference on Machine Learning (ICML 2025)}} \emph{(\bibinfo{series}{Proceedings of Machine Learning Research}, Vol.~\bibinfo{volume}{267})}. \bibinfo{publisher}{PMLR}, \bibinfo{pages}{64604--64627}.
\newblock
\urldef\tempurl%
\url{https://proceedings.mlr.press/v267/wang25da.html}
\showURL{%
\tempurl}


\bibitem[Wu(2023)]%
        {wu2023study}
\bibfield{author}{\bibinfo{person}{Wei{-}Lin Wu}.} \bibinfo{year}{2023}\natexlab{}.
\newblock \emph{\bibinfo{title}{A Study of the Expressive Power of Homomorphism Counts}}.
\newblock \bibinfo{thesistype}{Ph.\,D. Dissertation}. \bibinfo{school}{University of California, Santa Cruz}, \bibinfo{address}{Santa Cruz, CA, USA}.
\newblock
\urldef\tempurl%
\url{https://www.escholarship.org/uc/item/4647715d}
\showURL{%
\tempurl}


\bibitem[Xu et~al\mbox{.}(2019)]%
        {xu2019powerful}
\bibfield{author}{\bibinfo{person}{Keyulu Xu}, \bibinfo{person}{Weihua Hu}, \bibinfo{person}{Jure Leskovec}, {and} \bibinfo{person}{Stefanie Jegelka}.} \bibinfo{year}{2019}\natexlab{}.
\newblock \showarticletitle{How Powerful are Graph Neural Networks?}. In \bibinfo{booktitle}{\emph{7th International Conference on Learning Representations (ICLR 2019)}}. \bibinfo{publisher}{OpenReview.net}.
\newblock


\bibitem[Yu et~al\mbox{.}(2023)]%
        {DBLP:journals/nn/YuYLWP23}
\bibfield{author}{\bibinfo{person}{Dongran Yu}, \bibinfo{person}{Bo Yang}, \bibinfo{person}{Dayou Liu}, \bibinfo{person}{Hui Wang}, {and} \bibinfo{person}{Shirui Pan}.} \bibinfo{year}{2023}\natexlab{}.
\newblock \showarticletitle{A survey on neural-symbolic learning systems}.
\newblock \bibinfo{journal}{\emph{Neural Networks}}  \bibinfo{volume}{166} (\bibinfo{year}{2023}), \bibinfo{pages}{105--126}.
\newblock
\href{https://doi.org/10.1016/j.neunet.2023.06.028}{doi:\nolinkurl{10.1016/j.neunet.2023.06.028}}


\bibitem[Zhang and Sheng(2024)]%
        {zhang2024neurosymbolicaiexplainabilitychallenges}
\bibfield{author}{\bibinfo{person}{Xin Zhang} {and} \bibinfo{person}{Victor~S. Sheng}.} \bibinfo{year}{2024}\natexlab{}.
\newblock \bibinfo{title}{Neuro-Symbolic {AI}: Explainability, Challenges, and Future Trends}.
\newblock
\showeprint[arxiv]{2411.04383}~[cs.AI]


\bibitem[Zhang et~al\mbox{.}(2020)]%
        {DBLP:conf/iclr/ZhangCYRLQS20}
\bibfield{author}{\bibinfo{person}{Yuyu Zhang}, \bibinfo{person}{Xinshi Chen}, \bibinfo{person}{Yuan Yang}, \bibinfo{person}{Arun Ramamurthy}, \bibinfo{person}{Bo Li}, \bibinfo{person}{Yuan Qi}, {and} \bibinfo{person}{Le Song}.} \bibinfo{year}{2020}\natexlab{}.
\newblock \showarticletitle{Efficient Probabilistic Logic Reasoning with Graph Neural Networks}. In \bibinfo{booktitle}{\emph{8th International Conference on Learning Representations (ICLR 2020)}}. \bibinfo{publisher}{OpenReview.net}.
\newblock


\bibitem[Zhao et~al\mbox{.}(2024)]%
        {DBLP:journals/pacmmod/ZhaoDKRT24}
\bibfield{author}{\bibinfo{person}{Hangdong Zhao}, \bibinfo{person}{Shaleen Deep}, \bibinfo{person}{Paraschos Koutris}, \bibinfo{person}{Sudeepa Roy}, {and} \bibinfo{person}{Val Tannen}.} \bibinfo{year}{2024}\natexlab{}.
\newblock \showarticletitle{Evaluating Datalog over Semirings: {A} Grounding-based Approach}.
\newblock \bibinfo{journal}{\emph{Proc. ACM Manag. Data}} \bibinfo{volume}{2}, \bibinfo{number}{2} (\bibinfo{year}{2024}), \bibinfo{pages}{90:1--90:26}.
\newblock
\href{https://doi.org/10.1145/3651591}{doi:\nolinkurl{10.1145/3651591}}


\bibitem[Zhao et~al\mbox{.}(2021)]%
        {DBLP:conf/aaai/ZhaoWSHSY21}
\bibfield{author}{\bibinfo{person}{Jianan Zhao}, \bibinfo{person}{Xiao Wang}, \bibinfo{person}{Chuan Shi}, \bibinfo{person}{Binbin Hu}, \bibinfo{person}{Guojie Song}, {and} \bibinfo{person}{Yanfang Ye}.} \bibinfo{year}{2021}\natexlab{}.
\newblock \showarticletitle{Heterogeneous Graph Structure Learning for Graph Neural Networks}. In \bibinfo{booktitle}{\emph{Proceedings of the 35th AAAI Conference on Artificial Intelligence (AAAI 2021)}}. \bibinfo{publisher}{AAAI Press}, \bibinfo{address}{Washington, DC, USA}, \bibinfo{pages}{4697--4705}.
\newblock
\href{https://doi.org/10.1609/aaai.v35i5.16600}{doi:\nolinkurl{10.1609/aaai.v35i5.16600}}


\bibitem[Zhou et~al\mbox{.}(2020)]%
        {zhou2020graph}
\bibfield{author}{\bibinfo{person}{Jie Zhou}, \bibinfo{person}{Ganqu Cui}, \bibinfo{person}{Shengding Hu}, \bibinfo{person}{Zhengyan Zhang}, \bibinfo{person}{Cheng Yang}, \bibinfo{person}{Zhiyuan Liu}, \bibinfo{person}{Lifeng Wang}, \bibinfo{person}{Changcheng Li}, {and} \bibinfo{person}{Maosong Sun}.} \bibinfo{year}{2020}\natexlab{}.
\newblock \showarticletitle{Graph neural networks: {A} review of methods and applications}.
\newblock \bibinfo{journal}{\emph{AI Open}}  \bibinfo{volume}{1} (\bibinfo{year}{2020}), \bibinfo{pages}{57--81}.
\newblock
\href{https://doi.org/10.1016/j.aiopen.2021.01.001}{doi:\nolinkurl{10.1016/j.aiopen.2021.01.001}}


\end{thebibliography}

\appendix

\section{Proofs for Section \ref{sec:CompareNRPs}}

As mentioned in the main body, if the transformations include element-wise product, $\oplus$ combination can simulate $\odot$ combination, although the resulting programs are not guaranteed to remain zero-ary, monadic, or frontier-guarded:
\begin{restatable}{proposition}{propConcatSimulatesProd}
\label{prop:NRM_concat_simulates_prod}
    For every NRP $\Pi_1$ with $\odot$ combination there exists an equivalent    
    $\Pi_2$ with only $\oplus$ combination.    
    Moreover:
    \begin{enumerate}
        \item If \,\,$\Pi_1$ is disjunction-free the same holds for $\Pi_2$; 
        \item For any function class $\mathcal{F}$ 
    containing $\reluFFNs$, if $\Pi_1$ only uses transformations from $\mathcal{F}$ then $\Pi_2$ uses transformations from $\mathcal{F}\cup \calF_\times$. 
    \end{enumerate}
\end{restatable}
\begin{proof}
    Let $\Psi$ in $\Pi_1$:
    \[ R(\tup x)\angs{\tsum}\;\Leftarrow_{\odot}\;
        R_1(\tup x_1)\land\dots\land R_\ell(\tup x_\ell)
    \]
    Substitute in the following rules:
    \begin{align*}
        R'(\tup x_1, \dots, \tup x_\ell)\;               & \Leftarrow_{\oplus}\;
        R_1(\tup x_1)\land\dots\land R_\ell(\tup x_\ell)                                                  \\
        R''(\tup x_1, \dots, \tup x_\ell)\angs{\mu}      & \Leftarrow R'(\tup x_1, \dots, \tup x_\ell) \\
        R(\tup x)\angs{\tsum}                           & \Leftarrow R''(\tup x_1, \dots, \tup x_\ell)
    \end{align*}
    where $\mu$ applies element-wise multiplication over the concatenation of vectors by multiplying subsets of the embedding with transformations in $\calF_\times$ and possibly using affine transformations to obtain embeddings of the right size. Note that if $\Pi_1$ is disjunction-free, so is $\Pi_2$. However, if $\Pi_1$ is zero-ary, monadic or frontier-guarded the same does not necessarily hold for $\Pi_2$.
\end{proof}

\section{Proofs for Section~\ref{sec:zero-ary}}
\thmQueryAlgs*
    \begin{proof}
    {[only if]}:
        Let $(\{F_1, \ldots, F_k\},X)$ be a non-adaptive query algorithm for $\mathcal{C}$. We construct a corresponding
         NRP $\Pi$. In the definition of $\Pi$ below, we make use of $\oplus$ as combination function. However, by Proposition~\ref{prop:NRM_prod_simulates_concat}, the program can be rewritten to only use $\odot$.
        
        The NRP $\Pi$ uses 
        IDBs $S_i[0]\angs{1}$ and $R_i[0]\angs{1}$ for each $i\leq k$, as well
        as IDBs $R_{all}[0]\angs{k}$ and
         $\text{Ans}[0]\angs{1}$.
        As rules, it contains:
        \begin{itemize}
            \item for each
        $i\leq k$, a conjunction rule
        of the form 
        \[S_i()\angs{\tsum}\from \left( \bigwedge_{\text{$R(a_1, \ldots, a_n)$ a fact of $F_i$}}R(x_{a_1},\ldots,x_{a_n})\right)\land\text{True}_1(),\]
        \item A disjunction rule $R_i()\angs{\tsum}\from S_i() \lor \text{True}_0()$,
        \item a rule $R_{all}()\from_\oplus R_1()\land\cdots\land R_k()$,
        \item a rule $\text{Ans}()\angs{\mu}\from R_{all}()$, where $\mu$ is the indicator function of $X$, i.e.,  $\mu(\tup{x})=1$ for $\tup{x}\in X$ and $\mu(\tup{x})=0$ for $\tup{x}\not\in X$. 
        \end{itemize}
        By construction, 
        an $R_{\text{all}}()$ fact is derived on every input database $D$, and its embedding vector is the 
        homomorphism count vector $(|\Hom(F_1,D)|,\ldots,|\Hom(F_k,D)|)$.
        Note that the disjunctive rules ensure that an
        $R_i$-fact is derived even when $|\Hom(F_i,D)|=0$.
        It follows that, with answer IDB $\text{Ans}$ and the default acceptance policy ($x>0$), $\Pi$ accepts the same databases as the non-adaptive query algorithm.

        [if]: Suppose $\mathcal{C}$ is defined by a gated zero-ary NRP $\Pi=(\Psi_1,\dots,\Psi_n)$.
Let $F_1,\dots,F_k$
be the distinct canonical databases obtained as follows: for each conjunctive rule of $\Pi$, collect the EDB-atoms occurring in its body and take their canonical database.
We show that there exists a set $X\subseteq \mathbb N^k$ such that, for every input database $D$, $D\in \mathcal{C}$ iff $(|\Hom(F_1,D)|,\dots,|\Hom(F_k,D)|)\in X$.

We will in fact prove a stronger claim. For a database $D$, let:
$$\mathbf h(D):=(|\Hom(F_1,D)|,\dots,|\Hom(F_k,D)|)$$
we then prove, by induction along the rule order of $\Pi$:
\\
\\
\noindent
\emph{Claim.}
For every IDB $R$ of $\Pi$ of arity $0\angs{d_R}$, there exists a function
$g_R:\mathbb N^k\to \{\bot\}\cup \mathbb R^{d_R}$
such that for every input database $D$ and $\tup e \in \reals^{d_R}$,
\begin{align*}
  R()\angs{\tup e}\in \Pi(D) &\iff g_R(\mathbf{h}(D)) = \tup e.
\end{align*}

Note that, since $R$ has content arity $0$, its e-relation is either empty or consists of the single e-fact $R()\langle \mathbf{e}\rangle$.
Thus, the above claim says that both the existence of $R()$ and, if present, its embedding are determined by $\mathbf h(D)$.
The proof is by induction over the rules $\Psi_i$, with cases corresponding to each of the three rule types.
If $\Psi_i$ is a disjunctive rule $R()\angs{\alpha} \Leftarrow R_1() \lor \cdots \lor R_\ell()$, then all $R_j$ with $1 \leq j \leq \ell$ must be IDBs, and $g_R$ can be obtained using $\alpha$ and the functions $g_{R_1}, \ldots, g_{R_\ell}$ which exist by the induction hypothesis. A similar argument applies if $\Psi_i$ is a transformation rule.
Now suppose that $\Psi_i$ is a conjunction rule of the
form 
\[
R()\langle \alpha\rangle \Leftarrow_o
\beta_1 \land \cdots \land \beta_m \land B_1() \land \cdots \land B_t(),
\]
where $\beta_1,\dots,\beta_m$ are the EDB-atoms in the body and $B_1(),\dots,B_t()$ are the zero-ary IDB-atoms in the body.
Let $F_\Psi$ be the canonical database of the EDB-atoms $\beta_1,\dots,\beta_m$.
By construction, $F_\Psi$ is one of $F_1,\dots,F_k$,
and by the induction hypothesis, the claim holds for $B_1, \ldots, B_t$.

For any input database $D$, 
whether an $R()$ fact is derived
is determined by whether the facts $B_1(),\dots,B_t()$ are present
and whether there is a homomorphism from $F_\Psi$ to $D$. The 
former, by induction hypothesis, is determined by \(\tup{h}(D)\), while
the latter is determined by \(\tup{h}(D)\)
because $F_\Psi$ is among $F_1, \ldots, F_k$.
Suppose now that all $B_1(),\dots,B_t()$ are present and there is a homomorphism from \(F_\Psi\) to \(D\).
Then the number of homomorphisms from the body of $\Psi_i$ is exactly $|\Hom(F_\Psi,D)|$.
Moreover, each such homomorphism contributes the same vector $c$ to the aggregation, which, moreover, is determined by $\tup{h}(D)$:
the EDB-atoms have embedding dimension $0$, so they contribute only the unique $0$-dimensional vectors, while each $B_j()$ contributes its fixed embedding $\tup{e}_{B_j}$, which is determined by $\mathbf h(D)$ by the induction hypothesis. 
Therefore, if $R()$ is derived, its embedding is
\[
\alpha(\multiset{\underbrace{c,\dots,c}_{|\Hom(F_\Psi,D)|\text{ times}}}).
\]
which is determined by $\tup{h}(D)$.
    \end{proof}

\section{Proofs for Section~\ref{sec:DHN}}
\subsection{Translation from DHN to Monadic NRP}
\label{sec:appdx_DHN_to_NRP}

We show that embeddings produced by homomorphism-queries, DHN layers and DHNs are also computed by monadic NRPs. Recall from Section~\ref{sec:DHN} that $\varepsilon(D^a)$ is the element-embedded database $({D'}^{a},\lambda')$ consisting of a classical database $D'$ that contains all facts in $D$ without their embeddings, a monadic relation $\adom$ that contains all elements in the active domain of $D$ and an embedding function such that for monadic relations $P_1,\dots,P_n$ and $\arity(P_i)=1\angs{d_i}$:
$$\lambda'(a) = \displaystyle\bigoplus_{1 \leq i \leq n} \begin{cases}
    (1)\oplus \tup e &\text{ if } P_i(a)\angs{\tup e} \in D\\
    (0)\oplus\tup 0^{d_i} &\text{ otherwise}
\end{cases}$$
As explained in Section~\ref{sec:DHN} we apply monadic NRPs to databases $D^\downarrow$ obtained from $D$ after dropping embeddings for relations with content arity larger than $1$.
\begin{lemma}
    \label{lem:Homquery_as_NRP}
    Let $(F^x,\mu)$ be a homomorphism query. There exists a monadic NRP $\Pi$ with monadic relation $R$ such that for each e-database $D$ and $a \in \adom(D)$:
    $$
    R(a)\angs{\tup e} \in \Pi(D^\downarrow) \text{ iff } \res((F^x,\mu),\varepsilon(D^a))=\tup e
    $$
    Moreover, if $\mu$ has functions in class $\calF$, including affine transformations, $\Pi$ has transformations in $\calF$.
\end{lemma}
\begin{proof}
    Let $\mu = \{\mu^y \mid y \in \adom(F)\}$. We construct conjunction and disjunction rules to obtain the element-embedding as a single embedding in the e-database:
\[\begin{array}{llll}
    \text{for $i = 1,\dots,n$:}  &\text{Zero}_i(x) &\from& \adom(x) \land \text{True}_{\tup 0^{d_i+1}}()\\ 
    \text{for $i = 1,\dots,n$:}  &P'_i(x) &\fromconcat& \text{True}_{(1)}() \land P_i(x)\\     
     \text{for $i = 1,\dots,n$:}   &P''_i(x)\angs{\tsum} &\from& P'_i(x) \lor \text{Zero}_i(x)\\
        &P_{\lambda'}(x) &\fromconcat& P''_1(x) \land \dots \land P''_n(x) \land \adom(x)
    \end{array}\]
    where $\text{True}_{\tup 0^{d_i+1}}()$ has an all-zero embedding of size $d_i+1$, $\text{True}_{(1)}$ has embedding $(1)$ and $\oplus$ combination is simulated with $\odot$ combination and affine transformations using Proposition~\ref{prop:NRM_prod_simulates_concat}.
For each $y \in \adom(F)$ add a transformation rule:
    \begin{align*}
        P^y(x)\angs{\mu^y} \from P_{\lambda'}(x)
    \end{align*}
Now, for each monadic relation $Q[1]\angs{d}$ in $\calS$ with $d > 0$ we add the
transformation rule $Q^{\circ}(y)\angs{\mu_0} \from Q(y)$, where $\mu_0 \colon
\mathbb{R}^d \to \mathbb{R}^0$; for relations with embedding dimension $0$ we let $Q^\circ := Q$. We then add the conjunction rule
\begin{align*}
    H(x)\angs{\tsum} \fromprod \bigwedge_{y \in \adom(F)} P^y(y)
    \land \bigwedge_{Q(\tup y) \in F} Q^\circ(\tup y)
\end{align*}
Here, the first conjunct ensures the rule multiplies the outputs of the $\mu^y$, and the second conjunct restricts the interpretations of $\tup y$ to homomorphisms from $F$. Now if there is a homomorphism from $F^x$ to $D^a$ then $H(a)$ has embedding:
    \begin{align*}
        \tsum \bigmultiset{ \prod_{y \in \adom(F)} \mu^y(\lambda'(h(y))) \mid h \in \Hom(F^x, D^a)} = \res((F^x,\mu), \varepsilon(D^a))
    \end{align*}
However, if there is no homomorphism then $\res((F^x, \mu),\varepsilon(D^a))$ is an all-zero embedding $\tup 0^{d'}$, where $d' \geq 0$ is the output dimension of the query, while $H(a)$ is not derived. We thus add the rules:
\[\begin{array}{lll}
        \text{Zero}(x) &\from& \text{True}_{\tup 0^{d'}}() \land \adom(x)\\
        R(x)\angs{\tsum} &\from& H(x) \lor \text{Zero}(x)
\end{array}\]
    Then the embedding of $R$ is the output of $(F^x,\mu)$.
\end{proof}

\begin{lemma}
    \label{lem:DHNLayer_as_NRP}
    Let $\calL = (((F^\bullet_1, \mathbf{\mu}_1), \dots, (F^\bullet_m, \mathbf{\mu}_m)), \rho)$ be a DHN layer. There exists a monadic NRP $\Pi$ with monadic relation $R$ such that for each e-database $D$ and $a \in \adom(D)$:
    \begin{align*}
        \calL(\varepsilon(D))(a) = \tup e \text{ iff } R(a)\angs{\tup e} \in \Pi(D^\downarrow)
    \end{align*}
    Moreover, if the transformations in $\calL$ are in function class $\calF$, including affine transformations, $\Pi$ has transformations in $\calF$.
\end{lemma}
\begin{proof}
    By Lemma~\ref{lem:Homquery_as_NRP} there exist NRPs $\Pi_1, \dots, \Pi_m$ producing relations $Q_1, \dots, Q_m$ computing the $m$ homomorphism queries in $\calL$. We let $\Pi$ contain all rules in $\Pi_1, \dots, \Pi_m$ and add a conjunction rule and transformation rule:
    \begin{align*}
        Q(x) \fromconcat Q_1(x) \land \dots \land Q_m(x) \qquad
        R(x)\angs{\rho} \from Q(x)
    \end{align*}
    where $\oplus$ combination is simulated with $\odot$ combination and affine transformations using Proposition~\ref{prop:NRM_prod_simulates_concat}.
\end{proof}

\begin{lemma}
    \label{lem:DHN_as_NRP}
    Let $\calN = (\calL_1, \dots, \calL_n)$ be a DHN. There exists a monadic NRP $\Pi$ with monadic relation $R$ such that for each e-database $D$ and $a \in \adom(D)$:
    $$
    \calN(\varepsilon(D))(a)=\tup e \text{ iff } R(a)\angs{\tup e} \in \Pi(D^\downarrow)
    $$
    Moreover, if $\calN$ has transformations in class $\calF$, including affine transformations, $\Pi$ has transformations in $\calF$.  
\end{lemma}
\begin{proof}
    By Lemma~\ref{lem:DHNLayer_as_NRP} there exists an NRP $\Pi_1$ so that the embedding of $P_1$ in $\Pi_1(D^\downarrow)$ is $\calL_1(\varepsilon(D))$. We now apply induction over the sequence of layers. Suppose the embedding of $P_i$ in $\Pi_i(D^\downarrow)$ is $\calL_i \circ \dots \circ \calL_1(\varepsilon(D))$, where we write $(\calL \circ \calL')(D,\lambda)$ denoting $\calL(D,\calL'(D,\lambda))$. Then:
    $$
    \calL'_{i+1} (\varepsilon(\Pi_i(D^\downarrow))) = \calL_{i+1} \circ \dots \circ \calL_1(\varepsilon(D))
    $$
    where $\calL'_{i+1}$ applies $\calL_{i+1}$ to the part of the element-embedding representing $P_i$, and ignores the rest of the element-embedding. By Lemma~\ref{lem:DHNLayer_as_NRP} there is a program $\Pi'_{i+1}$ so that for every e-database $D$, the embedding of $P_{i+1}$ in $\Pi'_{i+1}(D^\downarrow)$  is $\calL'_{i+1}(\varepsilon (D))$. Then if $\Pi_{i+1}$ contains all rules in $\Pi'_{i+1}$ and $\Pi_i$:
    \begin{align*}
       P_{i+1}(a)\angs{\tup e} \in \Pi_{i+1}(D^\downarrow) &\text{ iff } P_{i+1}(a)\angs{\tup e} \in \Pi'_{i+1}(\Pi_i(D^\downarrow))\\
                     &\text{ iff } \calL'_{i+1}(\varepsilon(\Pi_i(D)))(a) = \tup e\\
                     &\text{ iff } \calL_{i+1} \circ \dots \circ \calL_1(\varepsilon(D^\downarrow))(a) = \tup e
    \end{align*}
    We then let $\Pi = \Pi_n$ and $R = P_n$.
\end{proof}

Lemma~\ref{lem:DHN_as_NRP} shows that every total unary query computed by a DHN with transformations in $\calF$, including affine transformations, is computed by a monadic NRP with transformations in $\calF$. If a DHN has transformations in $\calF \cup \calF_\times$, then since element-wise product transformations can be simulated with $\odot$ combination, using Proposition~\ref{prop:NRPproductCombSimulatesProductTrans}, there exists a monadic NRP that computes the same query and only uses transformations in $\calF$. This proves one direction of Theorem~\ref{Thm:NRPMatchesDHN}.

\subsection{Translation from Monadic NRP to DHN}

Monadic NRPs have IDBs with content arity $0$ or $1$, and body atoms with arity $0\angs{d},1\angs{d}$ or $k\angs{0}$ for any $d,k \geq 0$. We first show that unary queries computed by monadic NRPs can also be computed without using zero-ary relations. In this section we assume for simplicity that the input relational schema $\scs$ does not contain zero-ary relations.
\begin{lemma}
\label{lem:remove_0ary_from_monadic}
    Let $\calQ$ be a unary embedded query computed by a monadic NRP over input schema $\scs$ without zero-ary relations, then $\calQ$ is computed by a monadic NRP where all IDBs are monadic, and all body atoms are monadic or have embedding dimension $0$.
\end{lemma}
\begin{proof}
    Let $\Pi_1$ be a monadic NRP. We construct monadic NRP $\Pi_2$ that computes the same monadic relations without zero-ary IDBs and without zero-ary body atoms with embedding dimension $>0$. In every disjunction or transformation rule in $\Pi_1$, we substitute each zero-ary atom $R()$ by monadic atom $R(x)$. In a conjunction rule in $\Pi_1$, if the head is monadic with variable $x$, substitute all zero-ary body atoms $R()$ by $R(x)$ matching the head variable. If the head is zero-ary substitute each zero-ary atom in the body or head $R()$ by a monadic atom $R(x)$, where $x$ is a single new variable used in all these monadic atoms. 
    Finally, in every rule that is now unsafe since $x$ occurs in the head but not in the body, add a body atom $\adom(x)$ where $\arity(\adom)=1\angs{0}$, and $\adom(a)$ is derived for every $a \in \adom(D)$. We show by induction that for every e-database $D$:
    \begin{enumerate}
        \item For every monadic $R$ and $a \in \adom(D)$, $R(a)\angs{\vec e} \in \Pi_1(D)$ iff $R(a)\angs{\vec e} \in \Pi_2(D)$.
    \item For every zero-ary $R$, $R()\angs{\vec e} \in \Pi_1(D)$ iff $R(a)\angs{\vec e} \in \Pi_2(D)$ for every $a \in \adom(D)$.
    \end{enumerate}
    Let $R^i$ be derived by rule $\Psi_i$ in $\Pi_1$ such that (1) and (2) hold for all relations in the body. If $\Psi_i$ is a disjunction rule:
    $$R^i(\tup x)\angs{\tsum}\from R_1(\tup x)\lor\dots\lor R_\ell(\tup x)$$
    then note that if $R^i$ is zero-ary or monadic, the same holds for all relations in the body of $\Psi_i$, so that applying either (1) or (2) to the relations in the body, (1) and (2) are preserved. The same argument applies if $\Psi_i$ is a transformation rule. Now suppose $\Psi_i$ is a conjunction rule:
    $$R^i(\tup x)\angs{\tsum}\fromprod R_1(\tup x_1)\land\dots\land R_\ell(\tup x_\ell)$$
    If $R^i$ is monadic with variable $x$ we replaced every zero-ary body atom $R_j()$ by an atom $R_j(x)$. Since these atoms have the same embeddings by (2), (1) is preserved. If $R^i$ is zero-ary, then we replaced the head by $R^i(x)$ and all zero-ary atoms $R_j()$ by $R_j(x)$. Since by (2) for every element $a \in \adom(D)$, the embedding of $R_j()$ in $\Pi_1(D)$ equals the embedding of $R_j(a)$ in $\Pi_2(D)$, the same holds for $R^i$. 
    
    By (1), $\Pi_1$ and $\Pi_2$ compute the same unary embedded queries.
\end{proof}

We now show that the embeddings produced by conjunction rules, disjunction rules and transformation rules of monadic NRPs are also computed by DHNs, where the DHN outputs an extra Boolean to indicate whether a database relation is derived. Using Lemma~\ref{lem:remove_0ary_from_monadic} we only translate rules with monadic IDBs and body atoms that are monadic or have embedding dimension $0$.
\begin{lemma}
    \label{lem:conjunction_rule_as_DHN}
    Let $\Psi$ be a conjunction rule that derives $R$ with $\arity(R)=1\angs{d}$ and has body atoms that are monadic or have embedding dimension $0$. There exists a DHN layer $\calL$ so that for each e-database $D$ and $a \in \adom(D)$:
    $$
    \calL(\varepsilon(D))(a) = \begin{cases}
        (1)\oplus \tup e &\text{ if } R(a)\angs{\tup e} \in \Psi(D)\\
        (0)\oplus\tup 0^{d} &\text{ if } R(a)\not\in \Psi(D)        
    \end{cases}
    $$
    Moreover, $\calL$ has $\reluFFNs$ and $\calF_\times$ as transformations, and the only monadic relation in patterns of homomorphism queries of $\calL$ is $\adom$.
\end{lemma}
\begin{proof}
    Let $\Psi$:
    \begin{align*}
        R(x)\angs{\tsum}\fromprod R_1(x_1)\land\dots\land R_m(x_m) \land R_{m+1}(\tup x_{m+1}) \land \dots\land R_{m+\ell}(\tup x_{m+\ell})
    \end{align*}
    where $R_1, \dots, R_m$ are monadic and $R_{m+1}, \dots R_{m+\ell}$ have embedding dimension $0$. For $i=1,\dots,m$ let $\arity(R_i)=1\angs{d_i}$.
    Suppose w.l.o.g. that $R_1, \dots R_m$ are the first $m$ monadic relations in the signature, i.e. in $\varepsilon(D)$ the first $\sum_{1 \leq i \leq m} (1 + d_i)$ values of the element-embedding $\lambda(a)$ are $\lambda_1(a)\oplus \dots \oplus \lambda_m(a)$, where $\lambda_i(a)$ is $(1)\oplus\tup e$ if $R_i(a)\angs{\tup e}\in D$ and $(0)\oplus\tup 0^{d_i}$ if $R_i(a) \not\in D$.
Let $\mathit{body}(\Psi)$ be the set of body atoms of $\Psi$. Let $F^x$ be the pointed database with a fact `$\adom(y)$' for every body variable $y$, 
and the non-monadic atoms in $\mathit{body}(\Psi)$ as facts without embeddings. 
Let $\mu = \{\mu^y: y \in \adom(F^x)\}$ where:
    \begin{align*}
        \mu^{y}(\lambda(v)) = \begin{cases}
            \prod_{R_j(y) \in \mathit{body}(\Psi)} \lambda_j(v) \text{ if $R_j(y) \in \mathit{body}$ for some $1\leq j\leq m$}\\
            (1)\oplus\tup 1^d \text{ otherwise}
        \end{cases}
    \end{align*}    
    so that $\mu^y$ performs element-wise product over the embeddings of monadic body atoms with variable $y$. Here embeddings of different dimensions are handled in the same way as by $\odot$ combination in NRPs, using affine transformations and element-wise multiplications over subspaces to produce an embedding of size $d$, where if one of the embeddings $\lambda_j(v)$ starts with $0$ the whole product is $\tup 0^{d}$.
    Let $H^a$ be the set of homomorphisms from the body of $\Psi$ to $D^{a}$. Then the output of homomorphism query $(F^x, \mu)$ is:
    \begin{align*}
        \res((F^x, \mu), \varepsilon(D))(a) & = \sum_{\pi \in \Hom(F^x,D^a)}\prod_{y \in \adom(F)} \mu^{y}(\lambda(\pi(y)))                                                                             \\
                                       & = \left(|H^a|,\sum_{\pi \in H^a} \prod \multiset{\tup e_{\pi,i} \mid i=1,\dots m \text{ and } R_i(\pi(x_i))\angs{\tup e_{\pi,i}}\in \varepsilon(D)}\right)
    \end{align*}
    Note that the first value of the output is obtained by summing over homomorphisms in $\Hom(F^x,D^a)$, which are maps from body variables to elements in $D$ preserving all non-monadic body atoms in $\Psi$, and for each such map counting $1$ if it preserves all monadic atoms and $0$ otherwise, yielding a total sum $|H^a|$. The remaining values are obtained by summing over the products of embeddings for each homomorphism in $H^a$, since homomorphisms in $\Hom(F^x,D^a)\setminus H^a$ that do not preserve all monadic body atoms yield product $\tup 0^{d}$. Applying a $\reluFFN$ $\rho$ that bounds the first value to a maximum $1$, the DHN layer $\calL = ((F^x, \mu),\rho)$ derives $(1)\oplus \tup e$ if $R(a)\angs{\tup e} \in \Psi(D)$ and $(0) \oplus \tup 0^{d}$ if $R(a) \not\in \Psi(D)$.
\end{proof}

\begin{lemma}    \label{lem:disjunction_rule_as_DHN}
    Let $\Psi$ be a disjunction rule that derives $R$ with arity $1\angs{d}$. There exists a DHN layer $\calL$ so that for each e-database $D$ and $a \in \adom(D)$:
    $$
    \calL(\varepsilon(D))(a) = \begin{cases}
        (1)\oplus \tup e &\text{ if } R(a)\angs{\tup e} \in \Psi(D)\\
        (0)\oplus \tup 0^{d} &\text{ if } R(a) \not\in \Psi(D)
    \end{cases}
    $$
    Moreover, $\calL$ has $\reluFFNs$ as transformations,  and the only monadic relation in patterns of homomorphism queries of $\calL$ is $\adom$.
\end{lemma}
\begin{proof}
    Let $\Psi$:
    \begin{align*}
        R(x)\angs{\tsum}\from R_1(x)\lor\dots\lor R_m(x)
    \end{align*}
    As in the lemma above suppose $R_1, \dots, R_m$ with $\arity(R_i)=1\angs{d_i}$ are the first $m$ monadic relations in the signature so that in the element-embedding $\lambda$ of $\varepsilon(D)$ the first $\sum_{1\leq i \leq m} (1 + d_i)$ values are $\lambda_1(a)\oplus \dots \oplus \lambda_m(a)$, where $\lambda_i(a)$ is $(1)\oplus\tup e$ if $R_i(a)\angs{\tup e}\in D$ and $(0)\oplus\tup 0^{d_i}$ if $R_i(a) \not\in D$.

    Let $F^x = \{\adom(x)\}$, let $\mu^x$ take the element-wise sum over the $\lambda_i$ and then upper bound the first value at $1$.  The DHN layer $\calL = ((F^x,\{\mu^x\}), I)$ where $I$ is the identity transformation then derives $(1)\oplus\tup e$ if $R(a)\angs{\tup e} \in \Psi(D)$ and $(0)\oplus \tup 0^{d}$ if $R(a) \not\in \Psi(D)$.
\end{proof}

\begin{lemma}
    \label{lem:transformation_rule_as_DHN}
    Let $\Psi$ be a transformation rule that derives $R$ with arity $1\angs{d}$. There exists a DHN layer $\calL$ so that for each e-database $D$ and $a \in \adom(D)$:
    $$
    \calL(\varepsilon(D))(a) = \begin{cases}
        (1)\oplus \tup e  &\text{ if } R(a)\angs{\tup e} \in \Psi(D)\\
        (0)\oplus \tup 0^{d} &\text{ if } R(a) \not\in\Psi(D)
    \end{cases}
    $$
    Moreover, if the transformation function of $\Psi$ is in a class $\calF$, including $\reluFFNs$, then $\calL$ has transformations in $\calF$. Further, the only monadic relation in patterns of homomorphism queries of $\calL$ is $\adom$.
\end{lemma}
\begin{proof}
    Let $\Psi$:
    \begin{align*}
        R(x)\angs{\mu} \from R'(x)
    \end{align*}
    Let $\arity(R')=1\angs{d'}$ and suppose w.l.o.g. that $R'$ is the first monadic relation in the signature, so that the first $1+d'$ values in the element-embedding $\lambda(a)$ of $\varepsilon(D)$ are $\lambda_1(a)$, which is $(1)\oplus \tup e$ if $R'(a)\angs{\tup e} \in D$ and $(0)\oplus \tup 0^{d'}$ if $R'(a) \not\in D$. Let $\calL = ((F^x, \{\mu^x\}),I)$ where $F^x = \{\adom(x)\}$ and, with $\tup c := \mu(\tup 0^{d'})$:
$$
\mu^x\big((x_0)\oplus \tup e\big) := (x_0)\oplus \big(\mu(\tup e) + x_0\cdot \tup c - \tup c\big).
$$   
    Then $\calL(\varepsilon(D))(a)$ is $(1) \oplus\tup e$ if $R(a)\angs{\tup e} \in \Psi(D)$ and $(0) \oplus \tup 0^{d}$ if $R(a) \not\in \Psi(D)$.
\end{proof}

\begin{lemma}
    \label{lem:stack_DHNs}
    Let $\calN_1, \dots \calN_n$ be DHNs with input dimension $d$ and output dimension $d_i$. There exists a DHN $\calN^\oplus$ with input dimension $d$ and output dimension $\sum_{1 \leq i \leq n} d_i$ such that for each element-embedded database $(D,\lambda)$ and $a \in \adom(D)$:
    \begin{align*}
        \calN^\oplus(D,\lambda)(a) & = \displaystyle\bigoplus_{1 \leq i \leq n} \calN_i(D, \lambda)(a)
    \end{align*}
\end{lemma}
\begin{proof}
    The first layer of $\calN^\oplus$ copies the input embedding of each node $n$ times. Then, given DHN layers $\calL_i=(\calF_i, \rho_i)$ for each $\calN_i$, $\calN^\oplus$ has a DHN layer $(\calL, \rho)$ where $\calL$ has all homomorphism queries of the $\calF_i$ ordered by $i$, and $\rho$ applies each $\rho_i$ to the feature space corresponding to the output of $\calN_i$.
\end{proof}
We can now prove Theorem~\ref{Thm:NRPtoDHN}:
\ThmNonTotalQueriesDHN*
\begin{proof}
Let $\calQ$ be computed by a monadic NRP. By Lemma~\ref{lem:remove_0ary_from_monadic} $\calQ$ is computed by a monadic NRP $\Pi$ where all IDBs are monadic, and all body atoms are monadic or have embedding dimension $0$. Let $\Pi = (\Psi_1,\dots,\Psi_n)$, let each $\Psi_i$ derive $R_i$ with arity $1\angs{d_i}$ and given e-database $D$ let $D_0 = D$ and for $i = 1,\dots, n$, $D_i = D_{i-1} \cup \Psi_i(D_{i-1})$. We show that for each $i = 1, \dots, n$ there exists a DHN $\calN_i$ so that $\calN_i(\varepsilon(D))$ is the element-embedding of $D_{i}$. By Lemmas~\ref{lem:conjunction_rule_as_DHN}, \ref{lem:disjunction_rule_as_DHN} and~\ref{lem:transformation_rule_as_DHN}, there is a DHN layer $\calL_1$ so that $\calL_1(\varepsilon(D))(a)$ is $(1) \oplus \tup e$ if $R_1(a)\angs{\tup e} \in \Psi_1(D)$ and $(0)\oplus \tup 0^{d_1}$ if $R_1(a) \not\in \Psi_1(D)$. By Lemma~\ref{lem:stack_DHNs} there then exists a DHN $\calN_1$ so that $\calN_1(\varepsilon(D))$ is the element-embedding of $\varepsilon(D_1)$.

Now let $i<n$ and suppose $\calN_i(\varepsilon(D))$ is the element-embedding of $\varepsilon(D_i)$. By Lemmas~\ref{lem:conjunction_rule_as_DHN}, \ref{lem:disjunction_rule_as_DHN} and~\ref{lem:transformation_rule_as_DHN}, there exists a DHN layer $\calL_{i+1}$ so that $\calL_{i+1}(\varepsilon(D_i))(a)$ is $(1)\oplus \tup e$ if $R_{i+1}(a)\angs{\tup e} \in \Psi_{i+1}(D_i)$ and $(0)\oplus \tup 0^{d_{i+1}}$ if $R_{i+1}(a) \not\in \Psi_{i+1}(D_i)$.
Since the only monadic relation in patterns of homomorphism queries in $\calL_{i+1}$ is $\adom$ (so that all relations used in patterns of homomorphism queries to $\varepsilon(D_i)$ are also in $\varepsilon(D) = (D',\lambda')$), and since $\calN_i(\varepsilon(D))$ is the element-embedding of $\varepsilon(D_i)$:
$$
\calL_{i+1}(\varepsilon(D_i)) = \calL_{i+1}(D', \calN_i(\varepsilon(D))) = L_{i+1} \circ \calN_i(\varepsilon(D))
$$
Thus, $\calL_{i+1} \circ \calN_i(\varepsilon(D))(a)$ is $(1)\oplus \tup e$ if $R_{i+1}(a)\angs{\tup e} \in \Psi_{i+1}(D_i)$ and $(0)\oplus \tup 0^{d_{i+1}}$ if $R_{i+1}(a) \not\in \Psi_{i+1}(D_i)$. Again using Lemma~\ref{lem:stack_DHNs}, there exists a DHN $\calN_{i+1}$ so that $\calN_{i+1}(\varepsilon(D))$ is the element-embedding of $D_{i+1}$. After $\calN_n$ we add a final layer that selects the part of the element-embedding representing the output relation of $\Pi$. This shows that $\calN(\varepsilon(D))(a)$ is $(1)\oplus \tup e$ if $R_n(a)\angs{\tup e} \in \Pi(D)$ and  $\calN(\varepsilon(D))(a) = (0)\oplus \tup 0^d$ otherwise. Queries are computed by monadic NRPs on databases $D^\downarrow$ obtained by dropping embeddings of relations with content arity larger than $1$. Since $\varepsilon(D) = \varepsilon(D^\downarrow)$:
\begin{align*}
    \calN(\varepsilon(D))(a) = \calN(\varepsilon(D^\downarrow))(a) = \begin{cases}
            (1)\oplus \tup e & \text{ if } a\angs{\tup e} \in \Pi(D^\downarrow) \\
            (0)\oplus\tup 0^{d} &\text{ if for all } \tup e \in \reals^d,  a\angs{\tup e} \not\in\Pi(D^\downarrow)
        \end{cases}
\end{align*}
\end{proof}
Theorem~\ref{Thm:NRPtoDHN} also shows that for every \emph{total} unary embedded query $\calQ$ computed by a monadic NRP with transformations in $\calF$, there exists a DHN with transformations in $\calF \cup \calF_\times$ that computes $\calQ$. One simply removes the first value of the element-embedding. Together with the results of section~\ref{sec:appdx_DHN_to_NRP}, this proves Theorem~\ref{Thm:NRPMatchesDHN}.

\section{Proofs for Section~\ref{sec:guarded}}
\subsection{Separation of Frontier-Guarded NRPs and Monadic NRPs}
\thmGuardedvsMonadic*




\RenewCommandCopy{\theHtheorem}{\thetheorem}
\RenewCommandCopy{\theHlemma}{\thelemma}
\RenewCommandCopy{\theHdefinition}{\thedefinition}
\RenewCommandCopy{\theHcorollary}{\thecorollary}



	The flat query $\mathcal{Q}_M$ over the schema \(\{E[2]\langle0\rangle,P[1]\langle0\rangle\}\) is given by the first-order formula
  \[\varphi(x) = \exists y\big( E(x,y) \land P(y) \land \neg\exists z (E(x,z)\land E(y,z))\big).\]


\begin{lemma}
 The flat query $\mathcal{Q}_M$ can be expressed by a simply-gated frontier guarded  $\relu$-FFN-based NRP.
\end{lemma}

\begin{proof}
  The following rules compute for each pair of neighboring values \(x,y\) how many common neighbors they have. 
\begin{align*}
  N'(x,y)\angs{\tsum} &\Leftarrow E(x,y)\land E(x,z)\land E(y,z)\land\mathrm{True}_1()\\
  N(x,y)\angs{\tsum} &\Leftarrow N'(x,y)\lor \mathrm{True}_0(x,y)
\end{align*}	
Here we use a guarded version of the $\mathrm{True}_0(x, y)$ shorthand via the transformation rule
\[
\mathrm{True}_0(x, y) \angs{\mu_0} \Leftarrow E(x, y)
\]
where $\mu_0$ is the constant zero function.

Next, we use a transformation rule 
\begin{align*}
  T(x,y)\angs{\mu} &\Leftarrow N(x,y).
\end{align*}
with the function \(\mu(e) = \relu(1-e)\), 
which, on non-negative integer inputs $e$, returns $1$ if  $e = 0$ and $0$ otherwise.
Last but not least, the program has to aggregate this information over all neighbors that satisfy \(P\):
\begin{align*}
	\mathrm{Res}(x)\angs{\tsum} &\Leftarrow E(x,y)\land T(x,y)\land P(y)
\end{align*}
The resulting program derives \(\mathrm{Res}(v)\langle n\rangle\) with $n\geq 1$ if \(v\) is an answer to $\mathcal{Q}_M$, and derives either \(\mathrm{Res}(v)\langle 0\rangle\) or no $\mathrm{Res}(v)$-fact at all otherwise.
With the usual acceptance policy, it therefore expresses $\mathcal{Q}_M$.
It is also easy to see that the program is frontier guarded.
\end{proof}

In what follows, we prove that the above query is not expressible by a monadic NRP. 

\begin{center}
\fbox{\parbox{0.9\linewidth}{\it\centering To simplify the exposition, in the remainder of this section of the appendix, all input schemas and databases will be quietly assumed to consist of relations of arity $k\angs{0}$ with $k>0$, as is indeed the case for the above query }}
\end{center}




\subsubsection{Pattern-WL and Monadic NRPs}

To prove that $\mathcal{Q}_M$ is not expressible by a monadic NRP, we use the connection to DHNs and the notion of pattern-Weisfeiler-Leman (WL) test recently introduced for DHNs by Luo and Virtema \cite{Luo2026}.
After defining this test, we show how to restrict it to patterns that are connected and have no articulation point.
Afterward, we introduce a construction that, given two pointed databases indistinguishable by one round of the pattern-WL test, creates two databases that are indistinguishable by the pattern-WL test for multiple rounds.
Finally, we build two databases that are indistinguishable by one round of the pattern WL test, but of which only one satisfies \(\mathcal{Q}_M\).
We can thus conclude that $\mathcal{Q}_M$ is not expressible by monadic NRPs.


  Let \(D_1^v\), \(D_2^u\) be two pointed databases.
  An injective  function \(h\colon \adom(D_1) \to \adom(D_2)\) is an \emph{isomorphic embedding} from $D_1^v$ to $D_2^u$ if
 \(h(v) = u\),
 and \(R(\mathbf{x})\in D_1\) if and only if \(R(h(\mathbf{x}))\in D_2\).
For two pointed databases $D_1^v$, $D_2^u$ we use $\emb(D_1^v, D_2^u)$ to denote the set of isomorphic embeddings from $D_1^v$ to $D_2^u$.
%
Fix a countable set of \emph{colors} $\mathcal{C}$. A \emph{coloring function} $\col$ is a function from pointed databases to $\mathcal{C}$.

\begin{definition}[Pattern-WL]
  Let \(\mathcal{P} = \{P_1^x, \ldots, P_n^x\}\) be a finite set of pointed databases which we call \emph{patterns}, and $\ell \in \mathbb{N}$.
  The $\ell$-round \(\mathcal{P}\)-WL algorithm takes as input a database $D$ and computes a series of coloring functions $\col^\mathcal{P}_{i}$ as follows.
   Initially, for all $v \in \adom(D)$, \[\col^\mathcal{P}_0(D^v) = \mathrm{HASH}(\{R\mid R(v,\ldots, v) \in D\}).\]
   Then, for all $i$ with $0 \leq i < \ell$ and $v \in \adom(D)$,
          \begin{align*}
            \col^\mathcal{P}_{i+1}(D^v) = \mathrm{HASH}\big(&\col^\mathcal{P}_{i}(D^v), \\
                                    &\multiset{(\col^\mathcal{P}_i\circ\, h^\bullet)\mid h\in \emb(P_1^x,D^v)},\\
                                   &\ldots\\
                                   &\multiset{(\col^\mathcal{P}_i\circ\, h^\bullet)\mid h\in \emb(P_n^x,D^v)},\\
                                   &\multiset{\col^\mathcal{P}_i(D^u)\mid u\in\mathrm{adom}(D)} \big),
          \end{align*}
          where $\mathrm{HASH}$ is an injective function with codomain $\mathcal{C}$ and  $(\col^\mathcal{P}_i \circ h^\bullet)$ is shorthand for the tuple
          $(\col^\mathcal{P}_i(D^{h(x_1)}), \ldots, \col^\mathcal{P}_i(D^{h(x_m)}))$ for any isomorphic embedding $h$ from some $P^x \in \mathcal{P}$ to $D^v$, $x_1, \ldots, x_m$ a fixed enumeration of $\adom(P)$.

  Two pointed databases \(D_1^v, D_2^u\) are \emph{distinguishable} by the \(\ell\)-round \(\mathcal{P}\)-WL algorithm if \(\col^\mathcal{P}_\ell(D_1^v)\neq \col^\mathcal{P}_\ell(D_2^u)\).
\end{definition}

Although written as a composition, \((\col^{\mathcal P}_i\circ h^\bullet)\) should simply be read as the color profile of the embedding $h$: it lists, in the fixed order $x_1,\ldots,x_m$, the current colors of $h(x_1),\ldots,h(x_m)$. 




\begin{definition} \label{def:p-wl-invariance}
Let $\mathcal{P}$ be a finite set of patterns and $\ell \in \mathbb{N}$.
A $1\langle d\rangle$-ary embedded query $\mathcal{Q}$ over databases 
is \emph{invariant} under the \(\ell\)-round \(\mathcal{P}\)-WL if 
for all pointed databases \(D_1^v\) and~\(D_2^u\),  \(\col^\mathcal{P}_\ell(D_1^v) = \col^\mathcal{P}_\ell(D_2^u)\), implies that  \(v\langle \mathbf{e} \rangle \in \mathcal{Q}(D_1) \) if and only if \(u \langle \mathbf{e} \rangle \in \mathcal{Q}(D_2)\),  for every $\mathbf{e}$.
\end{definition}
It is well known that homomorphism counts and isomorphic embedding counts are closely related to each other.
This is also the case for the \(\mathcal{P}\)-WL test, where we can also use homomorphisms instead of isomorphic embeddings.
The following lemma is an extension of the homomorphism count theorem by Lov\'asz~\cite{Lovasz2012}.
Let $P^\bullet$ and $R^\bullet$ be patterns. With $\sur(P^\bullet, R^\bullet)$ we denote the set of surjective homomorphisms from $P^\bullet$ to $R^\bullet$. When forming sets of patterns below, we tacitly choose one representative of each pointed-isomorphism class.



\begin{lemma}\label{lem:hom_count_dbs}
  Let \(P^\bullet\) be a pattern over a schema $\mathcal{S}$ and \(\mathcal{R}\)
   the set of all patterns $R^\bullet$ over $\mathcal{S}$ with $\sur(P^\bullet, R^\bullet) \neq \emptyset$. Then, for all coloring functions $\col$ and pointed databases $D_1^\bullet$, $D_2^\bullet$ over $\mathcal{S}$,
  \begin{enumerate}
      \item   ~\\[-7mm]
      $$
      \begin{array}{rcl}
  \multiset{( \col \circ h^\bullet)\mid h\in\hom(R^\bullet,D_1^\bullet)} &=& \multiset{( \col \circ h^\bullet)\mid h\in\hom(R^\bullet,D_2^\bullet)}
  \text{ for all } R^\bullet \in \mathcal{R} \\[1mm]
  \multicolumn{3}{c}{\text{implies that} \qquad\qquad\qquad\qquad~}\\[1mm]
  \multiset{( \col \circ h^\bullet)\mid h\in\emb(P^\bullet,D_1^\bullet)} &=& \multiset{( \col \circ h^\bullet)\mid h\in\emb(P^\bullet,D_2^\bullet)}.
  \end{array}
    $$
  \item 
   ~\\[-7mm]
      $$
      \begin{array}{rcl}
  \multiset{( \col \circ h^\bullet)\mid h\in\emb(R^\bullet,D_1^\bullet)} &=& \multiset{( \col \circ h^\bullet)\mid h\in\emb(R^\bullet,D_2^\bullet)}  \text{ for all } R^\bullet \in \mathcal{R} \\[1mm]
 \multicolumn{3}{c}{\text{implies that} \qquad\qquad\qquad\qquad~}\\[1mm]
  \multiset{( \col \circ h^\bullet)\mid h\in\hom(P^\bullet,D_1^\bullet)} &=& \multiset{( \col \circ h^\bullet)\mid h\in\hom(P^\bullet,D_2^\bullet)}.
 \end{array}
 $$
  \end{enumerate}
%
\end{lemma}

\begin{proof}
  The proof is an extension of the proof of Lov\'asz homomorphism count theorem~\cite{Lovasz2012} to pointed databases and coloring functions.
  Let \(P^\bullet\) be a pattern, $\mathcal{R}$ as in the lemma, $\col$ a coloring function and $D_1^\bullet$, $D_2^\bullet$ pointed databases.
  
  We start with $(2)$. It suffices to show that there is a function that, for any pointed  database $D^\bullet$, takes 
  all multisets
$\multiset{( \col \circ h^\bullet)\mid h\in\emb(R^\bullet,D^\bullet)}$ for $R^\bullet \in \mathcal{R}$ as input
and outputs
$\multiset{( \col \circ h^\bullet)\mid h\in\hom(P^\bullet,D^\bullet)}$.

  For this, let \(\tup h \) be the vector that contains for each pattern \(R^\bullet\in \mathcal{R}\) with $|\adom(R)| = m$ and \(\tup c \in \mathcal{C}^m\) the entry \(\mathbf{h}_{R^\bullet, \tup c} = |\{h\in\hom(R^\bullet,D^\bullet)\mid (\col\circ h^\bullet) = \tup c\}|\).
  Furthermore, let \(\mathbf{e}\) be the vector that contains for each pattern \(R^\bullet\in \mathcal{R}\) with $|\adom(R)| = m$ and \(\tup c \in \mathcal{C}^m\) the entry \(\mathbf{e}_{R^\bullet,\tup c} = |\{h\in\emb(R^\bullet,D^v)\mid (\col\circ h^\bullet) = \tup c\}|\).
  Since \(P^\bullet\in \mathcal{R}\) and the multisets of the lemma are uniquely defined by \(\mathbf{h}\) and \(\mathbf{e}\) and vice versa, it suffices to show: \(\mathbf{h}\) is uniquely defined by \(\mathbf{e}\).
  
  Following the arguments by~\citet{Lovasz2012}, every homomorphism \(h\in \hom(R^\bullet, D^\bullet)\) for some $R^\bullet \in \mathcal{R}$ can be decomposed into a surjective homomorphism \(h_1\in \sur(R^\bullet,S^\bullet)\) for some \(S^\bullet\in\mathcal{R}\) followed by an embedding \(h_2\in \emb(S^\bullet, D^\bullet)\).
  This decomposition is unique up to automorphisms of \(S^\bullet\).
  
  When considering colors,
  \[\mathbf{h}_{R^\bullet, \tup c} = \sum_{\substack{S^\bullet\in \mathcal{R}, S^\bullet\neq R^\bullet\\ h\in\sur(R^\bullet, S^\bullet)\\ c' \in \mathcal{C}^{|\adom(S)|}\\\tup c = \tup c'\circ h}} \dfrac{1}{|\mathrm{Aut}(S^\bullet)|}\mathbf{e}_{S^\bullet, \tup c'} + \mathbf{e}_{R^\bullet, \tup c},\]
  where \(\mathrm{Aut}(S^\bullet)\) is the set of automorphisms of \(S^\bullet\).
  Notice, for \(S^\bullet\neq R^\bullet\) we counteract the multitude of decomposition of homomorphisms by dividing by the number of automorphisms.
  For \(S^\bullet = R^\bullet\), we do not consider all of \(\sur(R^\bullet, R^\bullet)\) but only the identity function. 
  This concludes the proof of $(2)$.

  For $(1)$, we proceed similarly. Notice that there exists a linear function \(f\) that maps \(\mathbf{e}\) to \(\mathbf{h}\).
  Let \(\preceq\) be an order on \(\mathcal{R}\) that satisfies
  \(R^\bullet \preceq S^\bullet\) if there exists a surjective homomorphism from \(R^\bullet\) to \(S^\bullet\). 
  This order is well-defined, since if there are surjective homomorphisms from \(R^\bullet\) to \(S^\bullet\) and from \(S^\bullet\) to \(R^\bullet\), then \(R^\bullet\) and \(S^\bullet\) are isomorphic. 
  Let \(\mathbf{e}\) and \(\mathbf{h}\) be ordered according to \(\preceq\) and let \(M\) be the matrix representing \(f\).
  Then \(M\) is a triangular matrix, since \(\mathbf{h}_{R^\bullet,c}\) only depends on \(\mathbf{e}_{R^\bullet, \tup c}\) and \(\mathbf{e}_{S^\bullet, \tup c'}\) where \(R^\bullet \preceq S^\bullet\) and \(S^\bullet\not\preceq R^\bullet\).
  Moreover, all diagonal entries are \(1\).
  Since triangular matrices with all diagonal entries being nonzero can be inverted, \(\mathbf{e}\) is uniquely defined by \(\mathbf{h}\).
  %
%
\end{proof}

We now use Lemma~\ref{lem:hom_count_dbs} to prove that we can use homomorphisms in the \(\mathcal{P}\)-WL test.
We denote this alternative version by \(\mathcal{P}\)-WL\(^\mathrm{hom}\) and use $\col_i^{\mathcal{P}, \hom}$ for its coloring functions.
Invariance under \(\mathcal{P}\)-WL\(^\mathrm{hom}\) is defined in the obvious way mirroring Definition~\ref{def:p-wl-invariance}.

\begin{lemma}\label{lem:wl-emb-eq-wl-hom}
    Let $\ell \in \mathbb{N}$ and $\mathcal{Q}$ be a $1\langle d\rangle$-ary embedded query.
    Then the following two statements are equivalent:
  \begin{enumerate}
          \item There exists a finite set of patterns \(\mathcal{P}\) such that
          $\mathcal{Q}$ is invariant under \(\ell\)-round \(\mathcal{P}\)-WL.
          \item There exists a finite set of patterns \(\mathcal{R}\) such that 
          $\mathcal{Q}$ is invariant under  \(\ell\)-round \(\mathcal{R}\)-WL\(^\mathrm{hom}\).
  \end{enumerate}
\end{lemma}
\begin{proof}
  For $(1) \Rightarrow (2)$, let $\mathcal{P}$ be a finite set of patterns and $\mathcal{Q}$ be a $1\langle d\rangle$-ary embedded query that is invariant under $\ell$-round $\mathcal{P}$-WL.
  Let $\mathcal{S}$ be the input schema of $\mathcal{Q}$.
  Then set $\mathcal{R}$ to be the set of patterns $R^\bullet$ over $\mathcal{S}$ such that $\sur(P^\bullet, R^\bullet) \neq \emptyset$ for some $P^\bullet \in \mathcal{P}$. We may assume $\mathcal{R}$ to be finite by choosing only one representative per isomorphism class.
  
  We show that for all pointed databases $D_1^v$, $D_2^u$ over $\mathcal{S}$,
  $\col^{\mathcal{R},\hom}_\ell(D_1^v) = \col^{\mathcal{R},\hom}_\ell(D_2^u)$
  implies that 
  $\col^{\mathcal{P}}_\ell(D_1^v) = \col^{\mathcal{P}}_\ell(D_2^u)$.
  We show this via induction on $\ell$. The case $\ell = 0$ is immediate from the definition of $\col^{\mathcal{P}}_0$ and $\col^{\mathcal{R}, \hom}_0$.
  Now assume that the implication holds for $\ell$ and let $D_1^v$, $D_2^u$ be pointed databases such that
  $\col^{\mathcal{R},\hom}_{\ell + 1}(D_1^v) = \col^{\mathcal{R},\hom}_{\ell + 1}(D_2^u)$.
  Then, by definition of $\col^{\mathcal{R}, \hom}_{\ell + 1}$ and injectivity of $\mathrm{HASH}$,
  \begin{enumerate}[label=(\roman*)]
      \item $\col^{\mathcal{R},\hom}_{\ell}(D_1^v) = \col^{\mathcal{R},\hom}_{\ell}(D_2^u)$,
      \item $\multiset{(\col^{\mathcal{R}, \hom}_\ell \circ\, h^\bullet)\mid h\in \hom(R^\bullet,D_1^v)} = \multiset{(\col^{\mathcal{R}, \hom}_\ell \circ\, h^\bullet)\mid h\in \hom(R^\bullet,D_2^u)}$ for every $R^\bullet \in \mathcal{R}$, and
      \item $\multiset{\col^{\mathcal{R}, \hom}_\ell(D_1^w)\mid w\in\mathrm{adom}(D_1)} = \multiset{\col^{\mathcal{R}, \hom}_\ell(D_2^w)\mid w\in\mathrm{adom}(D_2)}$.
  \end{enumerate}
  Point~(i) combined with the induction hypothesis yields $\col^{\mathcal{P}}_{\ell}(D_1^v) = \col^{\mathcal{P}}_{\ell}(D_2^u)$.
  Now let $P^\bullet \in \mathcal{P}$. By construction of $\mathcal{R}$, $\mathcal{R}$ contains all $R^\bullet$ over $\mathcal{S}$ with $\sur(P^\bullet, R^\bullet) \neq \emptyset$. Hence it follows from Lemma~\ref{lem:hom_count_dbs}~(1) and Point~(ii) that
  \[ \multiset{(\col^{\mathcal{R}, \hom}_\ell \circ\, h^\bullet)\mid h\in \emb(P^\bullet,D_1^v)} = \multiset{(\col^{\mathcal{R}, \hom}_\ell \circ\, h^\bullet)\mid h\in \emb(P^\bullet,D_2^u)}.\]
  Applying the induction hypothesis yields
  $\multiset{(\col^{\mathcal{P}}_\ell \circ\, h^\bullet)\mid h\in \emb(P^\bullet,D_1^v)} = \multiset{(\col^{\mathcal{P}}_\ell \circ\, h^\bullet)\mid h\in \emb(P^\bullet,D_2^u)}$.
  Additionally, Point~(iii) and the induction hypothesis implies that 
  $\multiset{\col^{\mathcal{P}}_\ell(D_1^w)\mid w\in\mathrm{adom}(D_1)} = \multiset{\col^{\mathcal{P}}_\ell(D_2^w)\mid w\in\mathrm{adom}(D_2)}$.
  Overall, it follows that
  $\col^{\mathcal{P}}_{\ell+ 1}(D_1^v) = \col^{\mathcal{P}}_{\ell + 1}(D_2^u)$, as desired.

  We continue to show that $\mathcal{Q}$ is invariant under $\ell$-round $\mathcal{R}$-WL$^{\hom}$.
  Let $D_1^v$ and $D_2^u$ be pointed databases over $\mathcal{S}$
  such that 
  $\col^{\mathcal{R},\hom}_\ell(D_1^v) = \col^{\mathcal{R},\hom}_\ell(D_2^u)$.
  Then $\col^{\mathcal{P}}_\ell(D_1^v) = \col^{\mathcal{P}}_\ell(D_2^u)$.
  As $\mathcal{Q}$ is invariant under $\ell$-round $\mathcal{P}$-WL, 
  it follows that for every $\tup e$, $v\langle \tup e \rangle \in \mathcal{Q}(D_1)$ if and only if $u \langle \tup e \rangle \in \mathcal{Q}(D_2)$, as required.

  The proof of $(2) \Rightarrow (1)$ is analogous, using Lemma~\ref{lem:hom_count_dbs}~(2).
\end{proof}

We now show that every NRP $\Pi$ is invariant under $\mathcal{P}$-WL for some suitable  $\mathcal{P}$. This holds for
$2n$ rounds, with $n$ the number of rules in  $\Pi$
because in the worst case, simulating a rule requires two \(\mathcal P\)-WL rounds: one to encode the matches of its EDB pattern at each database element, and another to make these encodings available  through the global color multiset when the head variable does not occur in that pattern.

\begin{lemma}\label{lem:nrp-wl-invariance}
	For every monadic NRP \(\Pi\) such that $\mathcal{Q}_\Pi$ is $1\langle d \rangle$-ary, there exists a finite set of patterns \(\mathcal{P}\) 
    and an \(\ell\in \mathbb{N}\) 
    such that \(\mathcal{Q}_\Pi\) is invariant under \(\ell\)-round \(\mathcal{P}\)-WL.
\end{lemma}


\begin{proof}
  Let \(\Pi'\) be a monadic NRP over schema $\mathcal{S}$.
  By \Cref{lem:remove_0ary_from_monadic}, \(\Pi'\) can be transformed into an equivalent monadic NRP \(\Pi\) without zero-ary IDBs.
   By \Cref{lem:wl-emb-eq-wl-hom} it suffices to show that \(\mathcal{Q}_\Pi\) is invariant under \(k\)-round \(\mathcal{P}\)-WL\(^\mathrm{hom}\) for some finite set \(\mathcal{P}\) and some \(k\in\mathbb{N}\).
  
  We define a finite set of patterns \(\mathcal{P}\):
  for each conjunction rule \[T(x)\langle \alpha\rangle \Leftarrow_o R_1(x_1)\land\cdots\land R_n(x_n)\land S_1(\tup{y}_1)\land\cdots\land S_m(\tup{y}_m)\] in $\Pi$, where \(m\geq 1\) and each \(S_i\) is an EDB and each \(R_i\) is an IDB, add the pattern \(P^z\) to $\mathcal{P}$ that contains the values \(\bigcup_{i = 1}^m \tup y_i\) and the facts \(S_i(\tup y_i)\) for each \(1\leq i\leq m\).
  If \(x\in \bigcup \tup{y_i}\), choose \(z = x\), otherwise set \(z\) to some arbitrary variable in \(\bigcup \tup{y_i}\).
%
  %
  The set \(\mathcal{P}\) does not contain any pattern for conjunction rules with only IDBs, disjunction rules or transformation rules.


For any unary relation symbol $R$, we define \(\Pi_{R}(D^v)\) to be the function that returns \(\tup e\) if 
$R(v)\langle \tup e \rangle \in \Pi(D)$
and returns $\bot$ if there is no $\tup e$ with $R(v)\langle \tup e \rangle \in \Pi(D)$.
Further, let $\Pi = (\Psi_1, \ldots, \Psi_n)$ where $\Psi_i$ produces $R_i$ for all $i$ with $1 \leq i \leq n$. Set $\mathcal{S}_0 = \mathcal{S}$
and $\mathcal{S}_{i + 1} = \mathcal{S}_i \cup \{R_{i + 1}\}$ for all $i$ with $0 \leq i < n$.
  We  show by induction on \(\ell\) that for all unary relations $R$ in $\mathcal{S}_\ell$ the function \(\Pi_{R}\) is \(2\ell\)-round \(\mathcal{P}\)-WL\(^\mathrm{hom}\) invariant.

  Fix a database \(D\) over $\mathcal{S}$ and let \(D_\ell\) be the database that extends \(D\) with all IDB relations that are derived by $\Psi_1, \ldots, \Psi_\ell$.
  For the induction start, i.e. $\ell = 0$ without any rules applied, \(R\) must be an EDB relation and since we restrict ourselves to classical databases, \(R\) has embedding dimension \(0\).
  Thus, 
  \begin{align*}
    \Pi_{R}(D^v) &= \begin{cases}
      ()& \text{ if } R(v)\in D\\\bot &\text{ otherwise.}
    \end{cases}
  \end{align*}
  By injectivity of $\operatorname{HASH}$,
$\operatorname{col}^{\mathcal P,\mathrm{hom}}_0(D^v)$ determines whether
$R(v) \in D$, and hence determines $\Pi_R(D^v)$.
  Therefore, \(\Pi_{R}\) is \(0\)-round \(\mathcal{P}\)-WL{\(^\mathrm{hom}\)} invariant.

  In the induction step we show that \(\Pi_{R}\) is \(2\ell+2\)-round \(\mathcal{P}\)-WL$^{\hom}$ invariant.
  We consider the different types of rules.
  In the first case, let $\Psi_{\ell + 1}$ be a conjunction rule of the form
  \[R(x)\langle\alpha\rangle \Leftarrow_o R_1(x_1)\land\cdots\land R_n(x_n)\land R_1'(x_1')\land\cdots\land R_{n'}'(x_{n'}')\land S_1(\tup y_1)\land\cdots S_m(\tup y_m),\]
  where \(m\geq 1\), \(x\in \bigcup\tup{y_i}\), \(R_1,\ldots,R_n\) are IDB relations with \(x_i\in\bigcup\tup{y_i}\), \(R_1',\ldots,R_n'\) are IDB relations with \(x_i'\not\in \bigcup\tup{y_i}\), and \(S_1,\ldots,S_m\) are EDB relations.
  Let \(P^x\) be the pattern defined by the EDB relations in the body and let \(Q^x\) be defined by all facts in the body of this rule.
  By definition,
  \[\Pi_R(D^v) = \begin{cases}
    \bot &\text{ if } \hom(Q^x, D^v_\ell) =\emptyset\\
    \alpha(\multiset{o(\Pi_{R_1}(D^{h(x_1)}),\ldots,\Pi_{R_{n'}'}(D^{h(x_{n'}')}))\mid h\in \hom(Q^x, D^v_{\ell})})&\text{ otherwise. }
  \end{cases}
\]
We omit the \(S_i\) in the second line, since EDB relations have embedding arity \(0\). They are of course checked by homomorphisms from \(Q^x\) to \(D^v_\ell\).
%
%
Each homomorphism in \(\hom(Q^x,D^v_\ell)\) can be uniquely decomposed into a homomorphism \(h\in \hom(P^x, D^v_\ell) = \hom(P^x,D^v)\) such that \(\Pi_{R_i}(h(x_i))\neq \bot\) for all \(1\leq i\leq n\) and a mapping \(g\colon \{x_1',\ldots x_{n'}'\}\to \mathrm{adom}(D)\) such that \(\Pi_{R'_i}(g(x_i'))\neq \bot\) for all \(1\leq i\leq n'\).
Let \(\mathfrak{G}\) be the set of all these functions \(g\).
Thus, we can rewrite the definition of \(\Pi_R(D^v)\) as follows:
\begin{align*}
  \Pi_R(D^v) = \alpha'(\multiset{o'(\Pi_{R_1}(D^{h(x_1)})\ldots, \Pi_{R_n}(D^{h(x_n)}), \Pi_{R_1'}(D^{g(x_1')}),\ldots,\Pi_{R_{n'}'}(D^{g(x_{n'}')}))&\\
                                           \mid h\in\hom(P^x,D^v),\ g\in \mathfrak{G}&}),
\end{align*}
where \(o'\) returns \(\bot\) if one of the arguments is \(\bot\) and \(\alpha'\) ignores all elements \(\bot\) and returns \(\bot\) if the multiset is empty.
The multiset can be decomposed such that there exists a function \(\beta\) that reconstructs the multiset and applies \(o'\) and \(\alpha'\) and
\begin{align*}
  \Pi_R(D^v) = \beta(\multiset{(\Pi_{R_1}(D^{h(x_1)})\ldots, \Pi_{R_n}(D^{h(x_n)})\mid h\in\hom(P^x,D^v))},&\\
  \multiset{(\Pi_{R_1'}(D^{g(x_1')}),\ldots,\Pi_{R_{n'}'}(D^{g(x_{n'}')}))\mid g\in \mathfrak{G}}\phantom{,}&).
\end{align*}
By induction hypothesis, each \(R_i\) is \(2\ell\)-round  \(\mathcal{P}\)-WL$^{\hom}$ invariant and thus also \(2\ell+1\)-round invariant.
Therefore, the multiset \(\multiset{(\Pi_{R_1}(D^{h(x_1)})\ldots, \Pi_{R_n}(D^{h(x_n)})\mid h\in\hom(P^x,D^v))}\) is uniquely defined by \(\multiset{(\col^{\mathcal{P},\hom}_{2\ell+1}\circ\,h^\bullet)\mid h\in \hom(P^x,D^v)}\), which is part of \(\col^{\mathcal{P},\hom}_{2\ell+2}(D^v)\) since \(P^x\in \mathcal{P}\).
By the induction hypothesis, for every $1 \leq i \leq n'$ and
$u \in \operatorname{adom}(D)$, the color
$\operatorname{col}^{\mathcal P,\mathrm{hom}}_{2\ell+1}(D^u)$
determines $\Pi_{R_i'}(D^u)$, including whether it is $\bot$.
Consequently, the global multiset
\(
  \multiset{    \col^{\mathcal P,\hom}_{2\ell+1}(D^u)
    \mid u \in \mathrm{adom}(D)
  }
\)
determines the entire multiset
\(
  \multiset{
    (
      \Pi_{R_1'}(D^{g(x_1')}),
      \ldots,
      \Pi_{R_{n'}'}(D^{g(x_{n'}')})
    )
    \mid g \in \mathfrak G
  }.
\)
%
This concludes the case of conjunction rules with EDB relations and \(x\in \bigcup \tup{y_i}\).

The proof for conjunction rules where \(x\not\in \bigcup\tup{y_i}\) is similar.
  Each \(h\in\hom(Q^x,D^v_\ell)\) can be uniquely decomposed into an \(h'\in\hom(P^\bullet, D^w)\) for some value \(w\in\mathrm{adom}(D)\), a homomorphism \(h''\) that checks facts of the form \(R(x)\) at \(v\), and again a function \(g\) that checks all \(R_i'(x_i')\) with \(x_i'\neq x\).
  By induction hypothesis, the second part is encoded in \(\col^{\mathcal{P},\hom}_{2\ell+1}(D^v)\), and we have seen in the first case that the third part is also encoded in \(\col^{\mathcal{P},\hom}_{2\ell+2}(D^v)\).
  We have also seen that the first part is encoded in \(\col^{\mathcal{P},\hom}_{2\ell+1}(D^w)\) for each \(w\in\mathrm{adom}(D)\).
  This information is then stored in \(\multiset{\col^{\mathcal{P},\hom}_{2\ell+1}(D^w)\mid w\in \mathrm{adom}(D)}\), which is part of \(\col^{\mathcal{P},\hom}_{2\ell+2}(D^v)\).

The case of conjunction rules without EDB relations is analogous to the second type of conjunction rules, but \(h'\) is vacuous.

  The induction step for disjunction and transformation rules is analogous.
  Since these rules only use one variable \(x\), the embedding of \(R(v)\langle\tup e\rangle\) only depends on the embeddings \(R_i(v)\langle \tup e_i\rangle\).
  Thus, the patterns derived from conjunction rules suffice.

  Taking $R$ to be the output relation shows that $Q_\Pi$ is invariant under
$2n$-round $\mathcal P\text{-}\mathrm{WL}^{\hom}$, where $n$ is the
number of rules in $\Pi$. The claim now follows from Lemma~\ref{lem:wl-emb-eq-wl-hom}.
\end{proof}

\subsubsection{Biconnected patterns}

We now show that we can strengthen Lemma~\ref{lem:nrp-wl-invariance} by using only patterns of a restricted form, as defined next.
  The \emph{Gaifman graph} of a database \(D\) is the undirected graph \(G = (V,E)\) with \(V = \adom(D)\) and
    \((v,u)\in E\) if and only if there is a fact \(R(a_1,\ldots,a_n)\) in \(D\) such that \(v,u\in \{a_1,\ldots,a_n\}\).
A database $D$ is \emph{connected} if the Gaifman graph of $D$ is connected.

\begin{definition}[Decomposable Database]\label{def:decomposition_articulation_direct}
  Let \(D\) be a database.
  A value \(x\in \adom(D)\) is an \emph{articulation point} if there exist \(V_1,V_2\subseteq \adom(D)\) with \(V_1\setminus\{x\}\ne\varnothing\) and \(V_2\setminus\{x\}\ne\varnothing\)
  such that
  \begin{enumerate}
    \item \(V_1\cup V_2 = \adom(D)\) and \(V_1\cap V_2 = \{x\}\) and
    \item every fact \(R(\mathbf{a})\in D\) satisfies \(\mathbf{a}\subseteq V_1\) or \(\mathbf{a}\subseteq V_2\).
  \end{enumerate}
  We call $D$ \emph{decomposable} if it contains an articulation point and \emph{non-decomposable} otherwise.
  A pointed database $D^\bullet$ is decomposable if $D$ is decomposable.
\end{definition}
For binary schemas, connectedness and non-decomposability is equivalent to biconnectivity.
The main result of this section shows that monadic NRPs are  invariant under $\mathcal{P}$-WL already for a suitable set $\mathcal{P}$ of connected and non-decomposable patterns.

\begin{lemma}\label{cor:wl-connected-and-non-decomposable}
    For each monadic NRP $\Pi$ such that $\mathcal{Q}_\Pi$ is $1 \langle d \rangle$-ary, there exists a finite set of connected and non-decomposable patterns $\mathcal{P}$ and an $\ell \in \mathbb{N}$ such that $\mathcal{Q}_\Pi$ is invariant under $\ell$-round $\mathcal{P}$-WL.
\end{lemma}

To show this lemma, we first apply Lemma~\ref{lem:nrp-wl-invariance} and then argue that for every finite set of patterns \(\mathcal{P}\) and every $\ell \in \mathbb{N}$, there exists a finite set of connected and non-decomposable patterns \(\mathcal{R}\) and an $\ell' \in \mathbb{N}$ such that if two pointed databases \(D_1^\bullet\) and \(D_2^\bullet\) are $\ell'$-round $\mathcal{R}$-WL indistinguishable, then they are also $\ell$-round $\mathcal{P}$-WL indistinguishable.

We begin with showing this for sets of connected patterns.
Concretely, we show that disconnected patterns in \(\mathcal{P}\) can be replaced with their connected components without decreasing the expressivity of \(\mathcal{P}\)-WL.
Since \(\mathcal{P}\)-WL needs one round to spread global information, twice the number of rounds might be necessary to achieve the same distinguishing power. We start with considering only a single pattern. 

\begin{lemma}\label{lem:ccclosure_is_sufficient}
  Let \(P^x\) be a pattern over schema $\mathcal{S}$. There exists a finite set of connected patterns $\mathcal{R}$ over $\mathcal{S}$ such that for all coloring functions $\col$ and pointed databases $D_1^v$, $D_2^u$ over $\mathcal{S}$: if for all $R^\bullet \in \mathcal{R}$,
\[
\multiset{(\col\circ h^\bullet)\mid h\in \emb(R^\bullet, D_1^v)} = \multiset{(\col\circ h^\bullet)\mid h\in \emb(R^\bullet, D_2^u) }
\]
\hspace*{6.7cm}and
\[
\multiset{(\col\circ h^\bullet)\mid h\in \emb(R^\bullet, D_1^w),\ w\in\adom(D_1)} = \multiset{(\col\circ h^\bullet)\mid h\in \emb(R^\bullet, D_2^w),\ w\in\adom(D_2)},
\]
\hspace*{6.7cm}then
\begin{equation*}
\multiset{(\col\circ h^\bullet) \mid h\in \emb(P^x, D_1^v)} = \multiset{(\col\circ h^\bullet) \mid h\in \emb(P^x, D_2^u)}.
\tag{$*$}
\end{equation*}

\end{lemma}
\begin{proof}
  Let \(P^x\) be a pattern over schema $\mathcal{S}$. 
  The main challenge of this proof is that isomorphic embeddings from the connected components of \(P^x\) to some pointed database \(D^v\) cannot be composed to an embedding of \(P^x\) into \(D\), as unions of embeddings need not be injective.

  
  Let \(\mathcal{R}'\) be the set of all pointed databases $R^\bullet$ over $\mathcal{S}$ such that $\sur(P^\bullet,R^\bullet) \neq \emptyset$. By Lemma~\ref{lem:hom_count_dbs} it is the case that
  \[
\multiset{(\col\circ h^\bullet)\mid h\in \hom(R^\bullet, D_1^v)} = \multiset{(\col\circ h^\bullet)\mid h\in \hom(R^\bullet, D_2^u) }
  \]
  for every $R^\bullet \in \mathcal{R}'$ implies $(*)$.
  
  Let $R^x \in \mathcal{R}'$, and let $S_0^x, S_1^\bullet, \ldots, S_n^\bullet$
  be the connected components of $R^x$ with $S_0^x$ the connected component with $x \in \adom(S_0)$ and the distinguished value of all $S_i^\bullet$ being chosen arbitrarily.
  Every homomorphism in $\hom(R^\bullet, D_1^v)$ can be uniquely decomposed into homomorphisms $h_0 \in \hom(S_0^x, D_1^v)$ and $h_i \in \hom(S_i^\bullet, D_1^{w_i})$ for some $w_i \in \adom(D_1)$ for $1 \leq i \leq n$ and the same holds for homomorphisms into $D_2$. Take $\mathcal{R''}$ to be the union of all such connected components of all $R^\bullet$ in $\mathcal{R'}$. Then
    \[
\multiset{(\col\circ h^\bullet)\mid h\in \hom(R^\bullet, D_1^v)} = \multiset{(\col\circ h^\bullet)\mid h\in \hom(R^\bullet, D_2^u) }
  \] and
  \[
\multiset{(\col\circ h^\bullet)\mid h\in \hom(R^\bullet, D_1^w),\ w\in\adom(D_1)} = \multiset{(\col\circ h^\bullet)\mid h\in \hom(R^\bullet, D_2^w),\ w\in\adom(D_2)}
\]
for all $R^\bullet \in \mathcal{R}''$ implies $(*)$.
As a last step, take $\mathcal{R}$ to be the union of all pointed databases $R^\bullet$ over $\mathcal{S}$ such that $\sur(P^\bullet, R^\bullet) \neq \emptyset$ for some $P^\bullet \in \mathcal{R''}$.
As surjective homomorphisms preserve connectedness, all pointed databases in $\mathcal{R}$ are connected.
Furthermore, it follows using Lemma~\ref{lem:hom_count_dbs} that
    \[
\multiset{(\col\circ h^\bullet)\mid h\in \emb(R^\bullet, D_1^v)} = \multiset{(\col\circ h^\bullet)\mid h\in \emb(R^\bullet, D_2^u) }
  \] and
  \[
\multiset{(\col\circ h^\bullet)\mid h\in \emb(R^\bullet, D_1^w),\ w\in\adom(D_1)} = \multiset{(\col\circ h^\bullet)\mid h\in \emb(R^\bullet, D_2^w),\ w\in\adom(D_2)}
\]
for all $R^\bullet \in \mathcal{R}$ implies $(*)$, as desired. \qedhere

\end{proof}

Now for the announced transition from unrestricted patterns to connected patterns.

\begin{lemma}\label{lem:wl-remove-nonconnected}
	Let \(\mathcal{P}\) be a finite set of patterns and $\ell \in \mathbb{N}$. 
  There exists a finite set of connected patterns \(\mathcal{R}\)
  such that 
  if two pointed databases \(D_1^\bullet\) and \(D_2^\bullet\)  are $2\ell$-round $\mathcal{R}$-WL indistinguishable, then they are also $\ell$-round $\mathcal{P}$-WL indistinguishable.
\end{lemma}

\begin{proof}
    Let $\mathcal{R}$ be the union of the sets of connected patterns that exist by Lemma~\ref{lem:ccclosure_is_sufficient} for each $P^\bullet \in \mathcal{P}$.
    Clearly, $\mathcal{R}$ contains only connected patterns.
  We will now show the lemma by induction on \(\ell\).
  For the induction start there is nothing to do, since the initial colors do not depend on the set of patterns.
  For the induction step, assume \(\col_{2\ell+2}^{\mathcal{R}}(D_1^v) = \col_{2\ell+2}^{\mathcal{R}}(D_2^u)\).
  We have to show that \(\col_{\ell+1}^\mathcal{P}(D_1^v) = \col_{\ell+1}^\mathcal{P}(D_2^u)\).
  As $\mathrm{HASH}$ is injective, it follows that
\begin{align*}
  \col_{2\ell+1}^{\mathcal{R}}(D_1^v) &= \col_{2\ell+1}^{\mathcal{R}}(D_2^u),\\
  \multiset{(\col_{2\ell+1}^{\mathcal{R}}\circ\, h^\bullet)\mid h\in\emb(R^x,D_1^v)} &= \multiset{(\col_{2\ell+1}^{\mathcal{R}}\circ\, h^\bullet)\mid h\in \emb(R^x,D_2^u)} \text{ for each } R^x\in \mathcal{R},\\
  \multiset{(\col_{2\ell+1}^{\mathcal{R}}(D_1^w))\mid w\in\mathrm{adom}(D_1)} &= \multiset{(\col_{2\ell+1}^{\mathcal{R}}(D_2^w))\mid w\in\mathrm{adom}(D_2)}.
\end{align*}
In turn, it follows from the first equality that \(\col_{2\ell}^{\mathcal{R}}(D_1^v) = \col_{2\ell}^{\mathcal{R}}(D_2^u)\). The induction hypothesis then yields \({\col_{\ell}^\mathcal{P}(D_1^v) = \col_{\ell}^{\mathcal{P}}(D_2^u)}\).

    Let now $P^x$ be a pattern in $\mathcal{P}$ and let $\mathcal{R}'$ be the set of connected patterns constructed in Lemma~\ref{lem:ccclosure_is_sufficient}.
  The above equalities imply
  \[\multiset{(\col_{2\ell}^{\mathcal{R}}\circ\, h^\bullet)\mid h\!\in\!\emb(R^\bullet,D_1^w), w\!\in\!\mathrm{adom}(D_1)} \!=\! \multiset{(\col_{2\ell}^{\mathcal{R}}\circ\, h^\bullet)\mid h\!\in\!\emb(R^\bullet,D_2^w), w\!\in\!\mathrm{adom}(D_2)}\] for all \(R^\bullet\in \mathcal{R}'\).
  By induction hypothesis,
  \[\multiset{(\col_{\ell}^{\mathcal{P}}\circ\, h^\bullet)\mid h\!\in\!\emb(R^\bullet,D_1^w), w\!\in\!\mathrm{adom}(D_1)} \!=\! \multiset{(\col_{\ell}^{\mathcal{P}}\circ\, h^\bullet)\mid h\!\in\! \emb(R^\bullet,D_2^w), w\!\in\! \mathrm{adom}(D_2)}\] for all \(R^\bullet\in \mathcal{R}'\).
  Analogously, by the second equation above and by the induction hypothesis
  \[\multiset{(\col_{\ell}^{\mathcal{P}}\circ\, h^\bullet)\mid h\in\emb(R^\bullet,D_1^v)} = \multiset{(\col_{\ell}^{\mathcal{P}}\circ\, h^\bullet)\mid h\in \emb(R^\bullet,D_2^u)}\] 
  for all \(R^\bullet\in \mathcal{R}'\).
  From \Cref{lem:ccclosure_is_sufficient} follows
 %
  \[\multiset{(\col_{\ell}^{\mathcal{P}}\circ\, h^\bullet)\mid h\in\emb(P^x,D_1^v)} = \multiset{(\col_{\ell}^{\mathcal{P}}\circ\, h^\bullet)\mid h\in \emb(P^x,D_2^u)}.\]
  The equality for the multisets of the third type:
  \[\multiset{\col_{\ell}^\mathcal{P}(D_1^w)\mid w\in\mathrm{adom}(D_1)} = \multiset{\col_{\ell}^\mathcal{P}(D_2^w)\mid w\in\mathrm{adom}(D_2)}\]
  follows easily from the induction hypothesis.
  Therefore, \(\col_{\ell+1}^\mathcal{P}(D_1^v) = \col_{\ell+1}^\mathcal{P}(D_2^u)\).
\end{proof}

We next show that it suffices to use patterns that are connected and non-decomposable.
This is a generalization of a result by \citet{barcelo2021graph}.
We first need some terminology around decomposability and articulation points.

Let $D$ be a database. The \emph{distance} of two values \(v\) and \(u\) in  \(D\) is their distance in the Gaifman graph of \(D\),
that is, the length of the shortest path between \(v\) and \(u\) in the Gaifman graph of \(D\).
The distance is \(\infty\) if there is no path  between \(v\) and \(u\).

The \emph{decomposition} of a pointed database  \(P^v\) is the set of  pointed databases \(Q^u\) such that \(Q\) is a maximal connected and non-decomposable induced subinstance of \(P\) and \(u\in\adom(Q)\) is a value with minimal distance from \(v\).
We say that \(P^v\) is \emph{\(n\)-decomposable} if its decomposition contains \(n+1\) patterns.
Thus $P^v$ is \(0\)-decomposable if and only if it is non-decomposable.

The following notions for undirected graphs are taken from \cite{Diestel2025}.
Let \(G = (V,E)\) be an undirected graph.
A \emph{cutvertex} in $G$ is a vertex \(v \in V\) that separates two vertices $x,y \in V \setminus \{v\}$  that belong to the same connected component of $G$. 
A \emph{block} is a maximal induced subgraph that does not contain a cutvertex.
The intersection of two blocks has cardinality of at most \(1\).
The \emph{block graph} of \(G\) is an undirected graph that contains a vertex for each cutvertex and for each block. It has an edge between a cutvertex \(v\) and a block \(b\) if \(v\) is in \(b\).
The following lemma makes precise the relation between articulation points and cutvertices.

\begin{lemma} \label{lem:cutvertex is articulation point}
Let $D$ be a connected database and $x \in \adom(D)$.
  The value $x$ is an articulation point if and only if $x$ is a cutvertex in the Gaifman graph of $D$.
\end{lemma}
\begin{proof}
  Let \(D\) be a connected database, and let \(x\) be an articulation point 
    and let \(V_1,V_2\) be the two connected components that \(x\) separates.
  By definition, there is no fact in \(D\) that contains values \(v_1 \in V_1, v_2\in V_2\) such that \(v_1\neq x\) and \(v_2\neq x\).
  Thus, in the Gaifman graph of $D$, every path from a value in
$V_1 \setminus \{x\}$ to a value in $V_2 \setminus \{x\}$ contains $x$.
Since both sets are nonempty, $x$ is a cutvertex.

  For the converse, let \(x\) be a cutvertex in the Gaifman graph of \(D\).
  Let \(y_1,y_2\) be two vertices that are separated by \(x\).
  Let \(V_2 = \{x\} \cup \{z\mid x\text{ separates \(z\) from \(y_1\)}\}\) and let \(V_1 = \{x\}\cup(\adom(D)\setminus V_2)\).
  Therefore, \(\{x\}\subsetneq V_1,V_2\subseteq \adom(D)\), \(V_1\cup V_2 = \adom(D)\) and \(V_1\cap V_2 = \{x\}\).
  Assume to the contrary, that there exists a fact \(R(\mathbf{a})\) such that \(\mathbf{a}\not\subseteq V_1\) and \(\mathbf{a}\not\subseteq V_2\).
  Thus, \(\mathbf{a}\) contains values \(a_1\not\in V_2\) and \(a_2\not\in V_1\).
  Hence, \(a_1\neq x\) and \(a_2\neq x\).
  By definition of Gaifman graph, there exists an edge between \(a_1\) and \(a_2\).
  This implies that there exists a path from \(a_2\) to \(a_1\) that does not visit \(x\).
  By definition of \(V_1\) there exists a path from \(a_1\) to \(y_1\) that does not visit \(x\).
  Thus, \(x\) does not separate \(a_2\) from \(y_1\) and \(a_2\in V_1\), which is a contradiction to the assumption \(a_2\not\in V_1\).
\end{proof}

In what follows we  freely use the term \emph{block} for maximal induced subinstances of databases that are connected and non-decomposable.
The following is a consequence of Lemma~\ref{lem:cutvertex is articulation point} and the fact that the block graph of the Gaifman graph of \(P^v\) is a tree, see \cite{Diestel2025}.
\begin{corollary}
  Let \(P^v\) be a connected pattern and let \(Q\) be a block of \(P^v\).
	Then there is a unique value in $\adom(Q)$ with minimal distance from \(v\).
\end{corollary}

\begin{definition}[Decomposition Step]\label{def:articulation_decomposition_step}
	Let \(P^v\) be a connected decomposable pattern.
  A \emph{decomposition step} consists of choosing an articulation point \(x\) and witnessing sets \(V_1,V_2\subseteq \adom(P)\).
  The results of this decomposition step are the subinstances  \(P_1^{v_1}\) and \(P_2^{v_2}\) induced by \(V_1\) and \(V_2\), respectively, with $v_i = v$ if $v \in V_i$ and $v_i = x$ otherwise.
\end{definition}
One can easily verify that if \(v\not\in V_i\), then the articulation point \(x\) is the unique value in \(V_i\) that has the minimal distance from \(v\).

When we replace a decomposable pattern \(P\) with its decomposition \(\{P_1,\ldots,P_n\}\), we have to deal with the following problem.
To compose isomorphic embeddings \(h_1,\ldots,h_n\)
of the decomposition patterns
that agree on common values to an isomorphic embedding of \(P\), we need that if \(h_i(x) = h_k(y)\) then \(x = y\) and that they are an articulation point in $P$. We solve this
problem by replacing 
 \(P\) not only with its decomposition, but also including additional patterns.
We  introduce these additional patterns next and prove that they are at most \(n\)-decomposable if \(P\) is \(n+1\)-decomposable.

\begin{definition} \label{def:sticky-situation}
Let \(P^v\) be a decomposable pattern over a schema $\mathcal{S}$, \(P_1^\bullet\) and \(P_2^\bullet\)  the result of a decomposition step at value~\(w\).
Then $\glueclosure_{\mathcal{S}}(P_1^\bullet, P_2^\bullet, w)$ is the set of all databases $R^z$ over $\mathcal{S}$,  up to isomorphism, such that there exists a surjective homomorphism $h$ from $P$ to $R$ and
\begin{enumerate}
    \item $h(v) = z$,
    \item the restriction of $h$ to $\adom(P_1)$ is an isomorphic embedding,
    \item the restriction of $h$ to $\adom(P_2)$ is an isomorphic embedding,
    \item $h$ is not an isomorphism.
\end{enumerate}
\end{definition}
Observe that the set $\glueclosure_{\mathcal{S}}(P_1^\bullet, P_2^\bullet, w)$
is finite because we require the homomorphism $h$ to be surjective. Since $P^v$ need not use all relations in $\mathcal{S}$, the definition
of this set depends on $\mathcal{S}$.
Finally, observe that
Point~(4)   implies that there are $v_1 \in \adom(P_1)$ and $v_2 \in \adom(P_2)$ with $v_1 \neq w$ and $v_2 \neq w$ with $h(v_1) = h(v_2)$ or that there is a fact $S(\tup v) \in R^z$ with $h(v_1), h(v_2) \in \tup v$.

\begin{lemma}\label{lem:glue-closure-n-decomposable}
  Let \(P^v\) be a connected and \(n+1\)-decomposable pattern and \(P_1^\bullet\), \(P_2^\bullet\) be the result of a decomposition step at value \(w\).
  Then all patterns in \(\glueclosure_{\mathcal{S}}(P_1^\bullet,P_2^\bullet,w)\)  are connected and at most \(n\)-decomposable.
\end{lemma}

\begin{proof}
  Let \(R^\bullet \in \glueclosure_{\mathcal{S}}(P^\bullet_1,P^\bullet_2,w)\) and \(h\) be a surjective homomorphism from $P^v$ to $R^z$ satisfying Points~(1) to~(4) of Definition~\ref{def:sticky-situation} that witnesses this.

  Every block $Q$ of $P$ must lie entirely in $P_1$ or entirely in $P_2$. As $h$ is an embedding on $P_1$ and $P_2$, $Q$ must be mapped by $h$ into a single block of $R$.
  As $h$ is surjective, every block of $R$ must contain the image of at least one block of $P$.
  This implies that $R$ is at most $n + 1$-decomposable.

  As $h$ is not an isomorphism, there must be $v_1 \in \adom(P_1) \setminus \{w\}$ and $v_2 \in \adom(P_2) \setminus \{w\}$ such that
  $h(v_1) = h(v_2)$ or such that there is a fact $S(\tup x) \in R$ with $h(v_1), h(v_2) \in \tup x$. This necessarily merges two blocks of $P$.

  Let $G$ be the Gaifman graph of $P$. As $P$ is connected, there is a path $\pi$ from $v_1$ to $v_2$ in $G$. As $v_1 \in \adom(P_1)$ and $v_2 \in \adom(P_2)$,
  $v_1$ and $v_2$ are separated by $w$ and $\pi$ passes through $w$.
  Now consider the Gaifman graph $G'$ of $R$. The $h$-image of $\pi$ is a 
  walk from $h(v_1)$ to $h(v_2)$ in $G'$. If there is a fact $S(\tup x) \in R$ with $h(v_1), h(v_2) \in \tup x$, then this walk together with $S(\tup x)$ 
  induces a cycle, meaning that $h(v_1)$ and $h(v_2)$ lie in the same block of $R$.
  In the case $h(v_1) = h(v_2)$ then the $h$-image of $\pi$ forms a closed walk.
  If this walk contains a cycle, that cycle again witnesses that $h(v_1)$ and $h(v_2)$ lie in the same block of $R$.
  If this walk contains no cycle, then some distinct vertices occurring along \(\pi\) must themselves have been identified by \(h\); such an identification necessarily occurs across a separation between two blocks along \(\pi\), and therefore again merges two blocks of \(P\) in \(R\).
  In both cases, two blocks of $P$ are mapped to one block of $R$, which implies that $R$ is at most $n$-decomposable.    \qedhere

\end{proof}

\begin{lemma}\label{lem:decomposition-step-indistinguishable}
  Let \(\mathcal{P}\) be a finite set of connected patterns over schema $\mathcal{S}$ and \(P^v\in \mathcal{P}\) be an \(n+1\)-decomposable pattern.
  Then there exists a finite set of connected patterns \(\mathcal{Q}\), all of which are at most \(n\)-decomposable, such that
      if two pointed databases \(D_1^\bullet\) and \(D_2^\bullet\) over $\mathcal{S}$ are $2\ell$-round \((\mathcal{P}\setminus\{P\})\cup \mathcal{Q}\)-WL indistinguishable, then they are also $\ell$-round $\mathcal{P}$-WL indistinguishable.
\end{lemma}
\begin{proof}
  Let \(P_1^\bullet\) and \(P_2^\bullet\) be the result of  a decomposition step of \(P^v\) at value \(w\) according to \Cref{def:articulation_decomposition_step}.
  We define \(\mathcal{Q}\) as \(\{P_1^\bullet,P_2^\bullet\}\cup\glueclosure_{\mathcal{S}}(P_1^\bullet,P_2^\bullet,w)\) for the schema $\mathcal{S}$. 
  By \Cref{lem:glue-closure-n-decomposable} and choice of $P_1^\bullet,P_2^\bullet$, all patterns in \(\mathcal{Q}\) are at most \(n\)-decomposable.
  We now prove the implication by induction on \(\ell\).
  We abbreviate \((\mathcal{P}\setminus\{P\})\cup \mathcal{Q}\) with \(\mathcal{R}\).
Let \(D_1^a\) and \(D_2^b\) be two databases.
For \(\ell = 0\), that is, \(D_1^a\) and \(D_2^b\) are \(0\)-round \(\mathcal{R}\)-WL indistinguishable, there is nothing to do since the initial color does not depend on the patterns.
For the induction step, assume that \(\col_{2\ell+2}^{\mathcal{R}}(D_1^a) = \col_{2\ell+2}^{\mathcal{R}}(D_2^b)\), thus by injectivity of $\mathrm{HASH}$,
\begin{enumerate}
	\item \(\col_{2\ell+1}^{\mathcal{R}}(D_1^a) = \col_{2\ell+1}^{\mathcal{R}}(D_2^b)\),
  \item
  \(\multiset{(\col_{2\ell+1}^{\mathcal{R}}\circ\,h^\bullet)\mid h\in \emb(R^\bullet, D^a_1)} = \multiset{(\col_{2\ell+1}^{\mathcal{R}}\circ\, h^\bullet)\mid h\in \emb(R^\bullet, D^b_2)}\)\label{eq:same_mulisets_of_decomposed_pattern} for each \(R^\bullet\in \mathcal{R}\),
  \item \(\multiset{\col_{2\ell+1}^{\mathcal{R}}(D_1^y)\mid y\in\mathrm{adom}(D_1)} = \multiset{\col_{2\ell+1}^{\mathcal{R}}(D_2^y)\mid y\in\mathrm{adom}(D_2)}\).
\end{enumerate}
By induction hypothesis and since \(\col_{2\ell+1}^{\mathcal{R}}(D_1^a) = \col_{2\ell+1}^{\mathcal{R}}(D_2^b)\) implies \(\col_{2\ell}^{\mathcal{R}}(D_1^a) = \col_{2\ell}^{\mathcal{R}}(D_2^b)\),
\begin{itemize}
	\item \(\col_\ell^{\mathcal{P}}(D_1^a) = \col_\ell^{\mathcal{P}}(D_2^b)\),
  \item \(\multiset{(\col_\ell^{\mathcal{P}}\circ \, h^\bullet)\mid h\in \emb(R^\bullet, D^a_1)} = \multiset{(\col_\ell^{\mathcal{P}}\circ\, h^\bullet)\mid h\in \emb(R^\bullet, D^b_2)}\) for each \(R^\bullet\in \mathcal{P} \setminus \{P^v\}\), 
  \item \(\multiset{\col_{\ell}^\mathcal{P}(D_1^y)\mid y\in\mathrm{adom}(D_1)} = \multiset{\col_{\ell}^{\mathcal{P}}(D_2^y)\mid y\in\mathrm{adom}(D_2)}\).
\end{itemize}
It remains to show that \(\multiset{(\col_\ell^{\mathcal{P}}\circ \, h^\bullet)\mid h\in \emb(P^v, D^a_1)} = \multiset{(\col_\ell^{\mathcal{P}}\circ\, h^\bullet)\mid h\in \emb(P^v, D^b_2)}\).
Recall that $P^v$ was decomposed at value $w$ into $P_1^\bullet$ and $P_2^\bullet$. Without loss of generality we assume that 
the distinguished value of \(P_1^\bullet\) is \(v\) and the distinguished values of \(P_2^\bullet\) is \(w\), as we can swap $P_1^\bullet$ and $P_2^\bullet$ freely. Note that it may be the case that $v = w$.
We start with the equality
\[\multiset{(\col_{2\ell+1}^{\mathcal{R}}\circ\,h^\bullet)\mid h\in \emb(P_1^v, D^a_1)} = \multiset{(\col_{2\ell+1}^{\mathcal{R}}\circ\, h^\bullet)\mid h\in \emb(P_1^v, D^b_2)},\]
which is part of Point~(2).
In these tuples we can repeat \(\col_{2\ell+1}^{\mathcal{R}}(D^{h(w)})\) which is already encoded in the tuple \((\col_{2\ell+1}^{\mathcal{R}}\circ\, h^\bullet)\).
Therefore,
\({\multiset{(\col_{2\ell+1}^{\mathcal{R}}\circ\, h^\bullet, \col_{2\ell+1}^{\mathcal{R}}(D_1^{h(w)})) \mid h\!\in\! \emb(P^v_1, D^a_1)}}\) is equal to \({\multiset{(\col_{2\ell+1}^{\mathcal{R}}\circ\, h^\bullet, \col_{2\ell+1}^{\mathcal{R}}(D_2^{h(w)}))\mid h\!\in\! \emb(P^v_1, D^b_2)}}\).\\
Information about the isomorphic embeddings of \(P_2\) at \(h(w)\) are encoded in \(\col_{2\ell+1}^{\mathcal{R}}(D_i^{h(w)})\), therefore
\begin{align*}
	&\multiset{(\col_{2\ell+1}^{\mathcal{R}}\circ\, h^\bullet, \multiset{(\col_{2\ell}^{\mathcal{R}}\circ\, g^\bullet)\mid g\in \emb(P^w_2, D_1^{h(w)})}) \mid h\in \emb(P^v_1, D^a_1)} \\
  = {} &\multiset{(\col_{2\ell+1}^{\mathcal{R}}\circ\, h^\bullet, \multiset{(\col_{2\ell}^{\mathcal{R}}\circ\, g^\bullet)\mid g\in \emb(P^w_2, D_2^{h(w)})})\mid h\in \emb(P^v_1, D^b_2)}.
\end{align*}
 Flattening these multisets preserves the equality:
\begin{align*}
	&\multiset{(\col_{2\ell+1}^{\mathcal{R}}\circ\, h^\bullet, \col_{2\ell}^{\mathcal{R}}\circ\, g^\bullet)\mid h\in \emb(P^v_1, D^a_1), g\in \emb(P^w_2, D_1^{h(w)})} \\
  = {} &\multiset{(\col_{2\ell+1}^{\mathcal{R}}\circ\, h^\bullet,\col_{2\ell}^{\mathcal{R}}\circ\, g^\bullet)\mid h\in \emb(P^v_1, D^b_2), g\in \emb(P^w_2, D_2^{h(w)})}.
\end{align*}
By the induction hypothesis then
\begin{align*}
	&\multiset{(\col_{\ell}^{\mathcal{P}}\circ\, h^\bullet, \col_{\ell}^{\mathcal{P}}\circ\, g^\bullet)\mid h\in \emb(P^v_1, D^a_1), g\in \emb(P^w_2, D_1^{h(w)})} \\
  = {} &\multiset{(\col_{\ell}^{\mathcal{P}}\circ\, h^\bullet,\col_{\ell}^{\mathcal{P}}\circ\, g^\bullet)\mid h\in \emb(P^v_1, D^b_2), g\in \emb(P^w_2, D_2^{h(w)})}.
\end{align*}
This is useful, but not all combinations of $h\in \emb(P^v_1, D^a_1), g\in \emb(P^w_2, D_1^{h(w)})$ correspond to isomorphic embeddings from $P^v$ to $D^a_1$ (and the same for $D_2^b$). Now using the patterns in $\glueclosure_{\mathcal{S}}(P_1, P_2, w)$, we identify the multiset $\mathcal{B}_i$ of tuples such that
\[
\multiset{(\col_{\ell}^{\mathcal{P}}\circ\, h^\bullet, \col_{\ell}^{\mathcal{P}}\circ\, g^\bullet)\mid h\in \emb(P^v_1, D^\bullet_i), g\in \emb(P^w_2, D_i^{h(w)})} \setminus \mathcal{B}_i
\]
corresponds, after deleting one of the two identical entries corresponding to $w$
and applying the appropriate  permutation to the remaining entries, to
\[
\multiset{(\col_\ell^{\mathcal{P}}\circ \, h^\bullet)\mid h\in \emb(P^v, D^\bullet_i)}.
\]

For this, let $R^\bullet \in \glueclosure_{\mathcal{S}}(P_1^v,P_2^w,w)$. With $\sur^*(P^v, R^\bullet)$ we denote the surjective homomorphisms from $P^v$ to $R^\bullet$ that additionally satisfy the conditions of Definition~\ref{def:sticky-situation}. 
Let $T_R \subseteq \sur^*(P^v, R^\bullet)$ be a set containing one representative from each
$R^\bullet$-automorphism equivalence class of $\sur^*(P^v, R^\bullet)$.
Every $f \in \sur^*(P^v, R^\bullet)$ induces a function $\pi_f \colon \mathcal{C}^{|\adom(R)|} \to \mathcal{C}^{|\adom(P_1)| + |\adom(P_2)|}$ that maps a color tuple for $\adom(R)$ to a color tuple for $\adom(P_1)$ and $\adom(P_2)$, that is
\[
    \pi_f( \tup c ) = ( (c_{f(x^1_1)}, \ldots, c_{f(x^1_{|\adom(P_1)|})}), (c_{f(x^2_1)}, \ldots, c_{f(x^2_{|\adom(P_2)|})})).
\]
where $\tup c = (c_{y_1}, \ldots, c_{y_{|\adom(R)|}})$, and $x^1_1, \ldots, x^1_{|\adom(P_1)|}$ is a fixed enumeration of $\adom(P_1)$, and $x^2_1, \ldots, x^2_{|\adom(P_2)|}$ is a fixed enumeration of $\adom(P_2)$.

Then
\[
\mathcal{B}_i = \bigcup_{\substack{R^\bullet \in \glueclosure_{\mathcal{S}}(P_1^\bullet, P_2^\bullet, w)\\ f \in T_R}} \multiset{ \pi_f(\col^\mathcal{P}_\ell \circ\, e^\bullet )\mid e \in \emb(R^\bullet, D^\bullet_i)}.
\]
To see this, observe that first, every $f \in T_R$ and $e \in \emb(R^\bullet, D^\bullet_i)$
induces a pair $(h, g)$ that does not correspond to an embedding of $P^v$.
Let $f' \in T_R$ and $e' \in \emb(R^\bullet, D^\bullet_i)$. If $f, e$ yield the same pair as $f', e'$, then $e(f(y)) = e'(f'(y))$ for all $y \in \adom(P)$.
Then $e^{-1} \circ e'$ is an automorphism of $R^\bullet$ that converts $f$ into $f'$. As $T_R$ includes only one surjective homomorphism per automorphism equivalence class, $f = f'$. Then $e = e'$ directly follows. Moreover, $R^\bullet$ is determined up to isomorphism by the pair $(h,g)$,
since it is isomorphic to the subdatabase of $D_i$ induced by
$\mn{img}(h\cup g)$. Hence, this representation is unique.

For the other direction, let $(h, g)$ be a pair with $h\in \emb(P^v_1, D^\bullet_i), g\in \emb(P^w_2, D_i^{h(w)})$ that does not correspond to an embedding of $P^v$.
Set $q = h \cup g$, and let $A$ be the restriction of $D_i$ to $\mn{img}(q)$ with point $q(v)$. As $q$ is then a surjective homomorphism from $P^v$ to $A$ that satisfies the conditions of Definition~\ref{def:sticky-situation}, there is an $R^z \in  \glueclosure_{\mathcal{S}}(P_1^\bullet,P_2^\bullet,w)$ such that $R^z$ is isomorphic to $A^{q(v)}$. The surjective homomorphism $q$ can then be uniquely (using the same argument as above) decomposed into an $f \in T_R$ and $e \in \emb(R^\bullet, D^\bullet_i)$
such that
\[
    \pi_f(\col_\ell^{\mathcal{P}} \circ e^\bullet) = (\col^\mathcal{P}_\ell \circ h^\bullet, \col^\mathcal{P}_\ell \circ g^\bullet).
\]
Now, as by Point~(2) and the induction hypothesis
\begin{align*}
	&\multiset{(\col_{\ell}^{\mathcal{P}}\circ\, h^\bullet)\mid h\in \emb(R^\bullet, D^a_1)} 
  = \multiset{(\col_{\ell}^{\mathcal{P}}\circ\, h^\bullet)\mid h\in \emb(R^\bullet, D^b_2)},
\end{align*}
for all \(R^\bullet \in \glueclosure_{\mathcal{S}}(P_1^\bullet, P_2^\bullet, w)\),
it follows that $\mathcal{B}_1 = \mathcal{B}_2$.
Therefore,
\[\multiset{(\col_\ell^{\mathcal{P}}\circ \, h^\bullet)\mid h\in \emb(P^v, D^a_1)} = \multiset{(\col_\ell^{\mathcal{P}}\circ\, h^\bullet)\mid h\in \emb(P^v, D^b_2)}.\]

\end{proof}

Repeated application of Lemma~\ref{lem:decomposition-step-indistinguishable} yields a set of connected and non-decomposable patterns as required. Together with Lemma~\ref{lem:nrp-wl-invariance} and Lemma~\ref{lem:wl-remove-nonconnected} this establishes Lemma~\ref{cor:wl-connected-and-non-decomposable}.





\subsubsection{\texorpdfstring{Databases that are \(\ell\)-round \(\mathcal{P}\)-WL indistinguishable for large \(\ell\)}{Databases that are l-round P-WL indistinguishable for large l}}

Let $D$ be a database. The 1-hop neighborhood of a value $v \in \adom(D)$ is the set $S_v = \{ u \mid R(\tup a) \in D, \{v, u\} \subseteq \tup a \}$.
With $\mathcal{N}_1(D^v)$ we denote the subinstance of $D$ induced by the 1-hop neighborhood of $v$, that is
$\{ R(\bar a) \in D \mid \bar a \subseteq S_v \}$.

The following is the main lemma of this subsection:

\begin{lemma}\label{thm:1-round-indist-extension-to-l-round-indist}
	Let \(\mathcal{P}\) be a finite set of connected and non-decomposable patterns and let \(D_1^v,D_2^u\) be two pointed
  databases which are \(1\)-round \(\mathcal{P}\)-WL indistinguishable.
  Then for each \(\ell \in \mathbb{N}\) there exist pointed databases \(\hat{D}_1^v, \hat{D}_2^u\), such that
  \begin{itemize}
      \item $\mathcal{N}_1(D_1^v)$ is isomorphic to $\mathcal{N}_1(\hat D_1^v)$,
      \item $\mathcal{N}_1(D_2^u)$ is isomorphic to $\mathcal{N}_1(\hat D_2^u)$,
      \item \(\hat{D}_1^v \) and \(\hat{D}_2^u\) are \(\ell\)-round \(\mathcal{P}\)-WL indistinguishable.
  \end{itemize}
\end{lemma}

Intuitively, $\hat D$ is the depth-$\ell$ truncation of a natural infinite saturation of $D$ in which the 1-hop neighborhood of every value except (copies of) $u$ and $v$ is determined by its $\col^\mathcal{P}_0$-value.
This allows us to construct highly indistinguishable databases while preserving the (non-)satisfaction of $\mathcal{Q}_M$, since $\mathcal{Q}_M$ only depends on the \(1\)-hop neighborhood of \(x\).
To show the lemma we rely on all patterns in $\mathcal{P}$ being connected and non-decomposable. We introduce some new terminology.

Let $D$ be a database and $v \in \adom(D)$. With $\mathcal{B}(D^v)$ we denote the set of 
pointed databases $C^v$ where $C$ is a block of $D$ and $v \in \adom(C)$.

\begin{lemma}\label{lem:biconnected-embs}
    Let $P^x$ be a connected and non-decomposable pattern and $D^v$ a database. Then
    \[
        \emb(P^x, D^v) = \bigcup_{B^v \in \mathcal{B}(D^v) }\emb(P^x, B^v)
    \]
\end{lemma}
\begin{proof}
Immediately by Definition of isomorphic embeddings, $\emb(P^x, B^v) \subseteq \emb(P^x, D^v)$ for every $B^v \in \mathcal{B}(D^v)$.
Now let $h \in \emb(P^x, D^v)$ and assume that there is no $B^v \in \mathcal{B}(D^v)$ such that $h \in \emb(P^x, B^v)$.
  We will now argue that then \(P\) must be decomposable. By the assumption above, there must be a $B^v \in \mathcal{B}(D^v)$ such that $\mn{img}(h) \cap \adom(B)$ contains at least $v$ and one more value $v'$, and
  $\mn{img}(h) \setminus \adom(B)$ contains at least one value $v''$. As $v'$ and $v''$ are connected but lie in different blocks of $D$, there must be an articulation vertex $y$ in $D$ that separates $v'$ and $v''$
  As $h$ is an embedding and $P$ is connected, $h^{-1}(y)$ must exist, and it must separate $h^{-1}(v')$ and $h^{-1}(v'')$ in $P$. Hence, $h^{-1}(y)$ is an articulation point in $P$ and $P$ is decomposable.
\end{proof}

Let $\mathcal{P}$ be a finite set of connected and non-decomposable patterns, $\ell \in \mathbb{N}$ and $D_1^v$, $D_2^u$ be pointed databases that are $1$-round $\mathcal{P}$-WL indistinguishable.
To obtain the desired databases as in Lemma~\ref{thm:1-round-indist-extension-to-l-round-indist},
we first take the disjoint union of $D_1$ and $D_2$, and obtain a database $D$ such that $D^v$ and $D^u$ are $1$-round $\mathcal{P}$-WL indistinguishable. We construct a sequence of databases
$E_0 \subseteq E_1 \subseteq \cdots$ such that $E_k^v$ and $E_k^u$ are $k$-round $\mathcal{P}$ indistinguishable for all $k \geq 1$. Then setting $\hat D_1^v = E^v_\ell$ and $\hat D_2^u = E^u_\ell$ satisfies the lemma.


We begin by setting $E_0 = D$. To construct $E_{k + 1}$ from $E_k$ we extend $E_k$ with certain new values and facts. As new values will be copies of values 
in $D$, we maintain functions $\tau_k \colon \adom(E_k) \to \adom(D)$ that associate copies with their original value. Initially, we set $\tau_0$ to be the identity on $\adom(D)$

Now, to obtain $E_{k + 1}$ from $E_k$, do the following:
for each $w \in \adom(E_k) \setminus \adom(E_{k - 1})$ and $w' \in \adom(D)$ such that
\begin{enumerate}
    \item $\tau_k(w) \notin \{ v, u\}$ and $w' \notin \{v, u\}$,
    \item $\tau_k(w) \neq w'$,
    \item $\col^\mathcal{P}_0(E_k^w) = \col^\mathcal{P}_0(D^{w'})$,
\end{enumerate}
add a fresh disjoint copy of $D$ and glue $w$ and the copy of $w'$ together.
Here, $\tau_{k + 1}(w) = \tau_k(w)$, but $\tau_{k + 1}$ is extended to the other values of $D$ as obvious.
Additionally, we take $E_{-1}$ to be $\emptyset$.

\begin{lemma}
$E_\ell^v$ and $E_\ell^u$ are $\ell$-round $\mathcal{P}$-WL indistinguishable.
\end{lemma}
\begin{proof}
Let $F$ be a block of $E_\ell$.
Then $F$ is \emph{a copy of} a block $B$ of $D$, by which we mean that $F$ and $B$ are isomorphic and this isomorphism agrees with $\tau_\ell$ at every value except perhaps at the value that was glued to some other value with the same $\col_0^\mathcal{P}$.

We show that the following holds in $E_\ell$:
\begin{enumerate}
    \item $\col^\mathcal{P}_0(E_\ell^w) = \col^\mathcal{P}_0(D^{\tau_\ell(w)})$ for all $w \in \adom(E_\ell)$,
    \item every block of $E_\ell$ contains at most one value $w$ with $\tau_\ell(w) \in \{u, v\}$,
    \item For all $w \in \adom(E_\ell)$ with $\tau_\ell(w) \in \{u, v\}$, there is a bijection between
    $\mathcal{B}(E_\ell^w)$ and $\mathcal{B}(D^{\tau_\ell(w)})$
    such that for every $(F, B)$ in this bijection, $F$ is a copy of $B$.
    \item For all $w \in \adom(E_{\ell - 1})$ with $\tau_\ell(w) \notin \{u, v\}$, there is a bijection between
    $\mathcal{B}(E_\ell^w)$ and 
    \[
    \bigcup_{\substack{w' \in \adom(D) \setminus \{u, v\} \\ \col_0^\mathcal{P}(E^w_\ell) = \col_0^\mathcal{P}(D^{w'})} } \mathcal{B}(D^{w'})
    \]
    such that for every $(F, B)$ in this bijection, $F$ is a copy of $B$.

    \item for all $k$ with $0 \leq k \leq \ell$ and $w_1, w_2 \in \adom(E_{\ell - k})$ such that $\tau_\ell(w_1), \tau_\ell(w_2) \in \{u, v\}$ or both $\tau_\ell(w_1), \tau_\ell(w_2) \notin \{u, v\}$ and $\col^\mathcal{P}_0(E_\ell^{w_1}) = \col^\mathcal{P}_0(E_\ell^{w_2})$, it holds that \(\col^\mathcal{P}_k(E_\ell^{w_1}) = \col^\mathcal{P}_k(E_\ell^{w_2})\).
\end{enumerate}
The lemma then follows from Point~(5) for $k = \ell$ and $u, v \in \adom(E_0)$. Points~(1) to~(4) aid in the proof of Point~(5).

Point~(1) holds directly by construction of $E_\ell$ and $\tau_\ell$.
Point~(2) can easily be shown by induction over $\ell$, the fact that $u$ and $v$ are in disconnected components of $D$, and that the construction of $E_\ell$ only glues disjoint components together at single values.
Points~(3) and~(4) can similarly be shown by induction over $\ell$, the fact that the construction of $E_\ell$ does not glue anything to copies of $v$ and $u$, and that glueing at a single value only adds new blocks, but does not modify existing ones.

We prove Point~(5) by induction on \(k\). 
For \(k = 0\), observe that the case for $\tau_\ell(w_1), \tau_\ell(w_2) \notin \{u, v\}$ is immediate. Thus, let $\tau_\ell(w_1), \tau_\ell(w_2) \in \{u, v\}$.
As $D^v$ and $D^u$ are $1$-round $\mathcal{P}$-indistinguishable, $\col^\mathcal{P}_0(D^v) = \col^\mathcal{P}_0(D^u)$. By Point~(1) it follows that $\col^\mathcal{P}_0(E^{w_1}_\ell) = \col^\mathcal{P}_0(E^{w_2}_\ell)$.

For the induction step, let $w_1, w_2 \in \adom(E_{\ell - (k + 1)})$ such that $\tau_\ell(w_1), \tau_\ell(w_2) \in \{u, v\}$ or both $\tau_\ell(w_1), \tau_\ell(w_2) \notin \{u, v\}$ and $\col^\mathcal{P}_0(E^{w_1}_\ell) = \col^\mathcal{P}_0(E^{w_2}_\ell)$.
We show that \(\col^\mathcal{P}_{k + 1}(E_{\ell}^{w_1}) = \col^\mathcal{P}_{k + 1}(E_{\ell}^{w_2})\), for which it suffices to show that
\begin{enumerate}[label=(\roman*)]
    \item \(\col^\mathcal{P}_k(E_{\ell}^{w_1}) = \col^\mathcal{P}_k(E_{\ell}^{w_2})\),
    \item $\multiset{(\col^\mathcal{P}_k\circ\, h^\bullet)\mid h\in \emb(P^x,E_{\ell}^{w_1})} = \multiset{(\col^\mathcal{P}_k\circ\, h^\bullet)\mid h\in \emb(P^x,E_{\ell}^{w_2})}$ for every $P^x \in \mathcal{P}$, and
    \item $\multiset{ \col^\mathcal{P}_k(E_\ell^y) \mid y \in \adom(E_\ell)} = \multiset{\col^\mathcal{P}_k(E_\ell^y) \mid y \in \adom(E_\ell)}$.
\end{enumerate}

We begin with the case where \(\tau_\ell(w_1),\tau_\ell(w_2)\in \{v,u\}\).
Point~(i) follows from the induction hypothesis.
Point~(iii) is immediate.
For Point~(ii), let $P^x$ be a connected and non-decomposable pattern from $\mathcal{P}$. As $D^v$ and $D^u$ are $1$-round $\mathcal{P}$-WL indistinguishable and $\mathrm{HASH}$ is injective, 
\[
\multiset{(\col^\mathcal{P}_0\circ\, h^\bullet)\mid h\in \emb(P^x, D^v)} = \multiset{(\col^\mathcal{P}_0\circ\, h^\bullet)\mid h\in \emb(P^x, D^u)}.
\]
Now, by Lemma~\ref{lem:biconnected-embs}, $\emb(P^x, D^v)$ only depends on $\mathcal{B}(D^v)$.
By Point~(3), there is a bijection between $\mathcal{B}(D^{\tau_\ell(w_1)})$ and $\mathcal{B}(E^{w_1}_\ell)$ such that
all pairs in this bijection are isomorphic. The same holds for $\mathcal{B}(E^{w_2}_\ell)$. Therefore,
\[
\multiset{(\col^\mathcal{P}_0\circ\, h^\bullet)\mid h\in \emb(P^x, E_\ell^{w_1})} = \multiset{(\col^\mathcal{P}_0\circ\, h^\bullet)\mid h\in \emb(P^x, E_\ell^{w_2})}.
\]
Now, let $h \in \emb(P^x, E_\ell^{w_1})$ and $h' \in \emb(P^x, E_\ell^{w_2})$
such that $\col^\mathcal{P}_0(E_\ell^{h(y)}) = \col^\mathcal{P}_0(E_\ell^{h'(y)})$ for all $y \in \adom(P)$.
By construction of $E_\ell$, as $\tau_\ell(w_1), \tau_\ell(w_2) \in \{u, v\}$ and
$w_1, w_2 \in \adom(E_{\ell - (k + 1)})$ it follows that
$h(y), h'(y) \in \adom(E_{\ell - k})$ for all $y \in \adom(P)$.
Additionally, Point~(2) implies that $\tau_\ell(h(y)) \in \{u, v\}$ if and only if $\tau(h'(y)) \in \{u, v\}$ for all $y \in \adom(P)$. In fact, $\tau_\ell(h(y)), \tau_\ell(h'(y)) \in \{u, v\}$ holds only for $y = x$.
Using these observations it follows from the induction hypothesis that
$\col^\mathcal{P}_k(E_\ell^{h(y)}) = \col^\mathcal{P}_k(E_\ell^{h'(y)})$ for all $y \in \adom(P)$.
Therefore,
\[
\multiset{(\col^\mathcal{P}_k\circ\, h^\bullet)\mid h\in \emb(P^x, E_\ell^{w_1})} = \multiset{(\col^\mathcal{P}_k\circ\, h^\bullet)\mid h\in \emb(P^x, E_\ell^{w_2})},
\]
as desired. This concludes the first case.

In the second case, \(\tau(w_1),\tau(w_2)\not \in \{v,u\}\) and 
$\col^\mathcal{P}_0(E^{w_1}_\ell) = \col^\mathcal{P}_0(E^{w_2}_\ell)$.
Again, Point~(i) follows from the induction hypothesis and Point~(iii) is immediate.
For Point~(ii), let $P^x$ be a connected and non-decomposable pattern from $\mathcal{P}$ with $\adom(P) = \{x_1, \ldots, x_m\}$.  Observe that by applying Point~(4) to both $\mathcal{B}(E^{w_1}_\ell)$ and
$\mathcal{B}(E^{w_2}_\ell)$, we obtain a bijection between $\mathcal{B}(E^{w_1}_\ell)$ and
$\mathcal{B}(E^{w_2}_\ell)$ such that every pair $(B_1^\bullet, B_2^\bullet)$ in this bijection is isomorphic
and is the copy of the same block of $D$.
As $P^x$ is connected and non-decomposable, it follows by Lemma~\ref{lem:biconnected-embs} that
there is a bijection between $\emb(P^x, E_\ell^{w_1})$ and $\emb(P^x, E_\ell^{w_2})$ such that for every pair $(h, h')$ in this bijection, 
$\tau_\ell(h(x_j)) \in \{u, v\}$ if and only if $\tau_\ell(h'(x_j)) \in \{u, v\}$ for all $x_j \in \adom(P)$. 


Now let $h, h'$ be a pair in the bijection between $\emb(P^x, E_\ell^{w_1})$ and $\emb(P^x, E_\ell^{w_2})$ , and consider
$(h(x_1), \ldots, h(x_m))$ and $(h'(x_1), \ldots, h'(x_m))$.
As $h$ and $h'$ are isomorphic embeddings, it follows that
$\col^\mathcal{P}_0(E^{h(x_j)}_\ell) = \col^\mathcal{P}_0(E^{h'(x_j)}_\ell)$ for all $j$ with $1 \leq j \leq m$.
Additionally, by construction of $E_\ell$ it is the case that $h(x_j), h'(x_j) \in \adom(E_{\ell - k})$ for all $j$ with $1 \leq j \leq m$.
Hence, the induction hypothesis yields that
$\col^\mathcal{P}_k(E^{h(x_j)}_\ell) = \col^\mathcal{P}_k(E^{h'(x_j)}_\ell)$ for all $j$ with $1 \leq j \leq m$.
Therefore,
\[\multiset{(\col^\mathcal{P}_k\circ\, h^\bullet)\mid h\in \emb(P^x,E_{\ell}^{w_1})} = \multiset{(\col^\mathcal{P}_k\circ\, h^\bullet)\mid h\in \emb(P^x,E_{\ell}^{w_2})}.\]


This concludes the second case. \qedhere
\end{proof}

 Furthermore, it is easy to verify that $\mathcal{N}_1(D_1^v) = \mathcal{N}_1(E_\ell^v)$ and $\mathcal{N}_1(D_2^u) = \mathcal{N}_1(E_\ell^u)$.
 For this, observe that the construction of $E_\ell$ does not glue anything to $u$ and $v$, and while the construction glues copies to values in the $1$-hop neighborhood of $u$ and $v$, these copies are disjoint and do not add any facts involving only $1$-hop neighbors of $u$ or $v$.
 Lemma~\ref{thm:1-round-indist-extension-to-l-round-indist} follows.

\subsubsection{\texorpdfstring{A unary FO property which cannot be expressed by \(\mathcal{P}\)-WL invariant models}{A unary FO property which cannot be expressed by P-WL invariant models}}


We can now show that the query $\mathcal{Q}_M$ cannot be expressed by monadic NRPs.
For this we apply Lemma~\ref{cor:wl-connected-and-non-decomposable} in the following way.
Given some finite set of connected and non-decomposable patterns \(\mathcal{P}\) we first construct two databases \(D_1^v\) and \(D_2^u\) that are \(1\)-round \(\mathcal{P}\)-WL indistinguishable and such that $v \in \mathcal{Q}_M(D_1)$ and $u \notin \mathcal{Q}_M(D_2)$, and then use Lemma~\ref{thm:1-round-indist-extension-to-l-round-indist} to obtain $\ell$-round \(\mathcal{P}\)-WL indistinguishable databases.

For $n \geq 0$, let $T_n^x$ denote the pointed database over the schema $\{E[2]\langle0\rangle, P[1]\langle 0 \rangle\}$ with facts $E(x, y)$, $P(y)$, $E(x, z_i)$, and $E(y, z_i)$ for all $i$ with $1 \leq i \leq n$. 
Observe that $x \in \mathcal{Q}_M(T_0)$, but $x \notin \mathcal{Q}_M(T_n)$ for all $n \geq 1$.

Now, let $\mathcal{P}$ be a finite set of connected and non-decomposable patterns.
Let $n$ be the largest number such that $\mathcal{P}$ contains a pattern isomorphic to $T_n^x$, or $0$ if no $T_n^x$ occurs in $\mathcal{P}$.
Based on $n$, we construct the two desired databases:
\begin{itemize}
    \item $D_1$ contains the value $v$ and for each \emph{even} $k$ with $0 \leq k \leq n + 1$, $\binom{n + 1}{k}$ copies of $T_k^x$ with the copy of $x$ glued to $v$;
    In the case of $n = 0$, we additionally add a fact $E(v, y)$ to $D_1$, where $y$ is a fresh value.
    \item $D_2$ contains the value $u$ and for each \emph{odd} $k$ with $0 \leq k \leq n + 1$, $\binom{n + 1}{k}$ copies of $T_k^x$ with the copy of $x$ glued to $u$.
\end{itemize}

By construction, \(D_1\) contains a copy of \(T_0^x\).
Hence, \(v\) has a neighbor \(w\) that satisfies \(P\) and \(v\) and \(w\) do not have any common neighbors.
Therefore, one can easily verify that $v \in \mathcal{Q}_M(D_1)$ and $u \notin \mathcal{Q}_M(D_2)$.
It remains to show that both constructed databases are \(1\)-round \(\mathcal{P}\)-WL indistinguishable.

\begin{lemma}\label{lem:two_one_round_indistinguishable_dbs}
\(D_1^v\) and \(D_2^u\) are \(1\)-round \(\mathcal{P}\)-WL indistinguishable.
\end{lemma}
\begin{proof}
For this we analyze the possible isomorphic embeddings of patterns in $\mathcal{P}$ to $D_1$ and $D_2$. As we only consider 1-round indistinguishability, it suffices to consider patterns $P^x \in \mathcal{P}$ with 
$\emb(P^x,D_1^v) \neq \emptyset$ or $\emb(P^x,D_2^u) \neq \emptyset$.
Let $P^x$ be such a pattern. As $P^x$ is connected and non-decomposable, $\emb(P^x,D_1^v) \neq \emptyset$ or $\emb(P^x,D_2^u) \neq \emptyset$ implies by Lemma~\ref{lem:biconnected-embs} that $\emb(P^x, T_k^x) \neq \emptyset$ for some $k \geq 0$. The only non-decomposable subsets of $T_k^x$ that include $x$ are isomorphic to $T_{\ell}^x$ for $\ell \leq k$ or isomorphic to $\{ E(x, y)\}$.
Hence, $P^x$ must be either isomorphic to some $T_{\ell}^x$ or to $\{ E(x, y) \}$.


  We show that $|\emb(P^x, D^v_1)| = |\emb(P^x, D^u_2)|$ for all $P^x \in \mathcal{P}$.
  For patterns $P^x$ that are isomorphic to some $T^x_\ell$, we centrally use that for $\ell \leq k$,
  \[
    |\emb(T^x_\ell, T_k^x)| = \binom{k}{\ell}\ell!.
  \]
  Then,
  \[
  |\emb(T^x_\ell, D_1^v)| = \sum_{\substack{k=\ell\\ k \text{ is even}}}^{n+1} \binom{n+1}{k}\cdot\binom{k}{\ell} \cdot \ell!
  \]
  and 
    \[
  |\emb(T^x_\ell, D_2^u)| = \sum_{\substack{k=\ell\\ k \text{ is odd}}}^{n+1} \binom{n+1}{k}\cdot\binom{k}{\ell} \cdot \ell!
  \]
  The desired
\(|\emb(T_\ell^x, D_1^v)|  = |\emb(T_\ell^x, D_2^u)|\) then follows from the folklore fact (e.g. Table 169 in~\cite{Graham1990}) that for all $\ell \leq n$,
  \[\sum_{k=\ell}^{n+1} (-1)^k \binom{n+1}{k}\cdot\binom{k}{\ell} = 0.\]
   Observe that for $\ell = 0$ the fact that \(|\emb(T_0^x, D_1^v)|  = |\emb(T_0^x, D_2^u)|\) implies that the number of values $w$ in $\adom(D_1)$ such that $P(w) \in D_1$ is the same as the number of values $w'$ in $\adom(D_2)$ such that $P(w') \in D_2$.
   
  Similar arguments apply when $P_x$ is isomorphic to $\{E(x, y)\}$. Then,
  \[
  |\emb(P^x, D^v_1)| = \sum_{\substack{k=0\\ k \text{ is even}}}^{n+1} \binom{n+1}{k}\cdot k
  \]
  and 
    \[
  |\emb(P^x, D^u_2)| = \sum_{\substack{k=0\\ k \text{ is odd}}}^{n+1} \binom{n+1}{k}\cdot k.
  \]
  
  Equality follows from the fact that 
    \[
    \sum_{k=0}^{n+1} (-1)^k \binom{n+1}{k}\cdot k = \sum_{k=1}^{n+1}(-1)^k \binom{n+1}{k} \binom{k}{1} =  0,\]
    for all $n > 0$. For $n = 0$, the special case in the construction of $D^v_1$ ensures equality.
    Observe that by construction of $D_1$ and $D_2$, $|\emb(P^x, D^v_1)| = |\emb(P^x, D^u_2)|$ also implies that the number of values $w$ in $\adom(D_1)$ such that $P(w) \notin D_1$ is the same as the number of values $w'$ in $\adom(D_2)$ such that $P(w') \notin D_2$.

    These observations now imply that 
    \[ \multiset{\col^\mathcal{P}_0(D^x_1)\mid x\in\mathrm{adom}(D_1)}= \multiset{\col^\mathcal{P}_0(D^x_2)\mid x\in\mathrm{adom}(D_2)}.\]
    As in particular also $\col^\mathcal{P}_0(D^v_1) = \col^\mathcal{P}_0(D^u_2)$, we thus obtain that \(\col^\mathcal{P}_1(D_1^v) = \col^\mathcal{P}_1(D_2^u)\), as required.

\end{proof}


\begin{corollary}
The query $\mathcal{Q}_M$ cannot be expressed by gated monadic NRPs. 	
\end{corollary}
\begin{proof}
Let $(\Pi, A)$ be a gated monadic NRP over the schema of $\mathcal{Q}_M$. By Lemma~\ref{cor:wl-connected-and-non-decomposable} there is a finite set of connected and non-decomposable patterns $\mathcal{P}$ and $\ell \in \mathbb{N}$ such that $\mathcal{Q}_\Pi$ is invariant under $\ell$-round $\mathcal{P}$-WL.
However, it follows from Lemma~\ref{lem:two_one_round_indistinguishable_dbs} that there are pointed databases $D^v_1, D^u_2$ over the same schema that are $1$-round $\mathcal{P}$-WL indistinguishable and such that $v \in \mathcal{Q}_M(D_1)$ and $u \notin \mathcal{Q}_M(D_2)$.
By Lemma~\ref{thm:1-round-indist-extension-to-l-round-indist} there are then also databases 
$\hat D^v_1, \hat D^u_2$ over the same schema that are $\ell$-round $\mathcal{P}$-WL indistinguishable and such that $v \in \mathcal{Q}_M(\hat D_1)$ and $u \notin \mathcal{Q}_M(\hat D_2)$.
Therefore, $\mathcal{Q}_{(\Pi, A)}$ cannot be equivalent to $\mathcal{Q}_M$.
\end{proof}



\subsection{Translation from Frontier-Guarded NRP to Monadic NRP with Linear Order}
\thmGuardedNRP*
\begin{proof}
Let $\Pi=(\Psi_1, \ldots, \Psi_n)$ be a frontier-guarded NRP, 
and let $N$ be the maximum content arity of the EDB relations of $\Pi$.
It follows from $\Pi$ being frontier-guarded, by a straightforward inductive argument, that $\Pi$ can only derive
IDB e-facts whose content tuple is a projection of an EDB relation. 
The idea behind the construction of $\Pi'$ below is that we will 
represent each output tuple by the row-id of a corresponding EDB e-fact. This will allow us to simulate the rules of $\Pi$ by monadic rules that operate on row-ids. To make this work, we must
furthermore ensure that the representation of output tuples by row-ids is canonical. This is to avoid over-counting when aggregate functions are applied.

For every $k\leq N$, let $P_k$ be the finite set of all $k$-ary EDB-projections
$\rho=(R,(i_1,\ldots,i_k))$, where $R$ is an EDB relation of content arity $m$ and
$i_1,\ldots,i_k\in\{1,\ldots,m\}$. Fix an arbitrary linear order $\prec$ on every $P_k$. 
We start with the following rules, for
all $\rho=(R,(i_1,\ldots,i_k))\in P_k$
and  $\rho'=(Q,(j_1,\ldots,j_k)) \in P_k$ with $\rho'\prec\rho$:
\[
\begin{array}{lll}
\mathit{RowIDs}_\rho(x_1)\angs{\mu_0} &\from& R(x_1,\ldots,x_m)  \\
\mathit{SeenEarlier}_\rho(x_1)\angs{\mu_1} &\from& R(x_1,\ldots,x_m)\wedge R(y_1,\ldots,y_{m})\land y_1<x_1  \\
\mathit{Match}_{\rho,\rho'}(x_1)\angs{\mu_1} &\from&
R(x_1,\ldots,x_m)\wedge Q(y_1,\ldots,y_{m'}) 
\\
\mathit{Earlier}_\rho(x)\angs{\tsum}
&\from&
\mathit{RowIDs}_\rho(x)\lor \mathit{SeenEarlier}_\rho(x) \lor \bigvee_{\rho'\prec\rho}\mathit{Match}_{\rho,\rho'}(x) \\
\mathit{Least}_\rho(x)\angs{\mu} &\from& \mathit{Earlier}_\rho(x) \text{ where $\mu(e)=\mathrm{ReLU}(1-e)$}
\end{array}
\]
where, in the rule producing $\mathit{SeenEarlier}$ the variables in the body are identified so that $(x_{i_1},\ldots,x_{i_k})=(y_{i_1},\ldots,y_{i_k})$, and in the rule producing $\mathit{Match}$ rules,
the variables in the body are identified so that $(x_{i_1},\ldots,x_{i_k})=(y_{j_1},\ldots,y_{j_k})$.
Here, $\mu_0$ is the constant map with value $0$ and $\mu_1$ is the constant map with value $1$. 
We only use the facts about
$\tsum$ that 
$\tsum(\multiset{0,\ldots,0})=0$ and $\tsum(X)\geq 1$ for any multi-set $X$ consisting of zeroes and ones and containing at least one 1. We note that the first three rules above are each shorthand for a pair of rules (namely a conjunctive rule followed by a transformation rule), in the obvious way.

Note that each of the IDBs defined by the above rules is of arity $1\angs{1}$.
Furthermore, it follows from the construction of these rules that, for $\rho=(R,(i_1,\ldots, i_k))$,
$\mathit{Least}_\rho$ contains all row-ids of $R$,
and that the associated embedding
is either 0 or 1, with it being 1 precisely if the
$\rho$-projection of the row in question is not the $\rho$-projection of another row of $R$ with a smaller row-id, and also is 
not the $\rho'$-projection of any EDB-fact for any projection type 
$\rho'\prec\rho$. Thus, when $\mathit{Least}_\rho(a)\angs{1}$ holds and $\rho$-projection of the row with row-id $a$ is $(a_1, \ldots, a_k)$, then 
we can use $a$ as a canonical reference to the tuple $(a_1,\ldots,a_k)$.

Let $d_{\max}$ be the maximum embedding dimension of the IDBs of $\Pi$. For each
IDB $\textit{Least}_\rho$ as defined above 
and for each 
$d\leq d_{\max}$, it will be convenient to create an additional auxiliary IDB
$\mathit{Least}^{(d)}_\rho$ given by the transformation rule
\[ \mathit{Least}^{(d)}_\rho(x)\angs{\mu}\from \mathit{Least}_\rho(x)\]
where $\mu:\reals\to\reals^d$ simply maps $r$ to 
$\angs{r,\ldots, r}$.

We now construct,
for each $k\angs{d}$-ary IDB $T$ of $\Pi$ and for each $k$-ary projection $\rho=(R,(i_1,\ldots,i_k))$, a $1\angs{d}$-ary IDB $T_\rho$, maintaining the following invariant (over all databases $D$):
\begin{equation}
\{(a)\angs{\tup{e}}\mid T_\rho(a)\angs{\tup{e}}\in \Pi'(D)\} = \{a_1\angs{\tup{e}}\mid R(a_1,\ldots, a_m)\in D \text{ and } T(a_{i_1},\ldots, a_{i_k})\angs{\tup{e}}\in \Pi(D)\}
\tag{*}
\end{equation}
In particular, for $T=S$, this then implies the statement of the theorem.

The remaining rules of $\Pi'$ are constructed following the order of the rules $\Psi_1,\ldots,\Psi_n$ of $\Pi$.
Let $j\leq n$.
If $\Psi_j$ is a disjunction rule
\[
T(\tup{z})\angs{\alpha}\from U_1(\tup{z})\lor\cdots\lor U_\ell(\tup{z}),
\]
then, for each  $\rho\in P_k$, where $k$ is the content arity of $T$, it suffices to add
\[
T_\rho(x)\angs{\alpha}\from
U_{1,\rho}(x)\lor\cdots\lor U_{\ell,\rho}(x).
\]

Similarly, if $\Psi_j$ is a transformation rule
\[
T(\tup{z})\angs{\mu}\from U(\tup{z}),
\]
then, for every $\rho\in P_k$, where $k$ is the content arity of $T$, it suffices to add
\[
T_\rho(x)\angs{\mu}\from U_\rho(x).
\]
It remains to consider conjunction rules. Let
\[
\Psi_j:\quad
T(\tup{z})\angs{\tsum}\from_{\odot}
A_1(\tup{u}_1)\wedge\cdots\wedge A_t(\tup{u}_t)\wedge\phi
\]
where each $A_i$ is an IDB and 
where $\phi$ is a conjunction of EDB atoms. 
Let $k\angs{d}$ be the arity of $T$.
For all $k$-ary target projections $\rho=(R,(i_1,\ldots,i_k))\in P_k$, if there exists an $R$-atom $R(x_1,\ldots,x_m)$ whose $\rho$-projection
$(x_{i_1},\ldots,x_{i_k})$ equals $\tup{z}$, fix such an $R$-atom, where the $x_i$ with $i\not\in \{i_1,\dots,i_k\}$ are fresh variables. We now add rules to represent derivation of $T$ by derivation of a monadic relation $T_\rho$.
For each $s\leq t$, let $k_s$ be the content arity of $A_s$ and
let $\rho_s=(Q_s,(j^s_1,\ldots,j^s_{k_s}))$ range over $P_{k_s}$. 
For each such choice
$\bar{\rho}=(\rho_1,\ldots,\rho_t)$, define:
\[
\mathit{T}_{\rho,\bar{\rho}}(x_1)\angs{\tsum}\from_{\odot}
R(x_1,\ldots,x_m)\wedge
(\bigwedge_{s=1\ldots t} \psi_s)\wedge
\phi\]
where
$\psi_s=Q_s(y^s_1,\ldots,y^s_{m_s})\wedge \mathit{Least}^{(d)}_{\rho_s}(y^s_1)\wedge (A_s)_{\rho_s}(y^s_1)$
with the variables additionally identified so that $\tup{u}_s$ equals the projection $(y^s_{j^s_1},\ldots,y^s_{j^s_{k_s}})$,
and where the remaining variables $y^s_i$ with $i \notin \{j^s_1,\ldots,j^s_{k_s}\}$ are fresh and pairwise distinct, also across different $s$ and from the variables of the $R$-atom.
In case one or more of the relations $R$ and $Q_s$ are monadic EDBs with a positive embedding dimension,
we replace them by $1\angs{0}$-ary relations $R^\circ$ or $Q_s^\circ$, 
respectively, adding suitable transformation rules
of the form $R^{\circ}(x)\angs{\mu}\from R(x)$, 
in order to ensure that these atoms in the above rule body don't contribute an embedding.

Finally, we add the disjunctive rule
\[
T_\rho(x)\angs{\tsum}\from
\bigvee_{\bar{\rho}}\mathit{T}_{\rho,\bar{\rho}}(x).
\]

We  claim that, in this way, the invariant (*) is preserved. 

By the induction hypothesis, each IDB e-fact $A_s(\tup{c})\langle \tup{e}_s\rangle$ occurring in $\Pi(D)$ is represented in $\Pi'(D)$ by all row-ids of e-facts whose relevant projection is $\tup{c}$, with the same embedding $\tup{e}_s$. Among these representatives, exactly one is canonical: the least row-id for the first projection type, in the fixed order on projection types, that realizes $\tup{c}$. By construction, $Least_{\rho_s}$ has embedding $1^d$ on this canonical representative and embedding $0^d$ on all other representatives. 

Thus, for every row ID $r$ and projection $\rho$ so that $r$ identifies tuple $\bar b$, and for every homomorphism from the body of $\Psi_j$ that yields $T(\bar b)$, there is exactly one $\bar \rho$ that contributes a non-zero embedding to $T_{\rho, \bar \rho}(r)$,  namely by choosing the canonical representative for each IDB body atom. 
Conversely, every non-zero contribution to some $T_{\rho,\bar\rho}(r)$ determines such a homomorphism of the original rule. 
The auxiliary atoms \(R^{\circ}\) and \(Q_s^{\circ}\)
added to the rule body do not affect the
combined embedding, since they have
embedding dimension zero. 

Non-canonical representatives may still satisfy the translated body, but their contribution is zero because the product includes a $Least_{\rho_s}$-embedding equal to $0$. 

Consequently, the multiset of contributions summed into $T_\rho(r)$ in $\Pi'(D)$ is the same as the multiset of contributions summed into $T(\bar b)$ in $\Pi(D)$, up to additional zero vectors. Since\[\mathrm{sum}(\multiset{x,0,\ldots,0})=x,\] the two resulting embeddings are equal. The final disjunction over all choices of $\bar\rho$ merely ranges over the possible canonical projection types for the IDB body atoms, so these contributions are neither lost nor counted twice. 
Here we use associativity of sum, in the sense that summing first within the
relations \(T_{\rho,\bar\rho}\) and then summing over \(\bar\rho\) gives the same
result as summing over the union of all corresponding contribution multisets.
\end{proof}

\section{Proofs for Section \ref{sec:NRPandFOCrats}}

The proofs in this section establish a one-to-one correspondence between values in embeddings produced by NRPs with 
rational $\reluFFN$ transformations and terms in $\FOCrats$, to prove the equivalence for flat queries of Theorem~\ref{thm:NRP_and_FOC(1/2)_same_relations}.

\subsection{Translation from NRP to \texorpdfstring{$\FOCrats$}{FO+C(1/2)}}

Recall that an e-database $D$ over signature $\calS$ is represented by a weighted structure $\calA_D$ over signature $\sigma$, where for every $R[k]\angs{d} \in \calS$ there are $k$-ary $R,R_1,\dots,R_d$ in $\sigma$. For all $\tup a \in A^k$, $\sem{R}^{\calA_D}(\tup a)$ is $1$ if $R(\tup a) \in D$ and $0$ otherwise, and for $i = 1, \dots, d$, $\sem{R_i}^{\calA_D}(\tup a)$ is the $i$-th embedding value of $R(\tup a)$ in $D$, or $0$ if $R(\tup a) \not\in D$. We show that under this translation, the embeddings produced by NRP rules applied to database $D$ are represented by $\FOCrats$ terms interpreted on $\calA_D$. 

\begin{definition}
\label{def:representation_database_to_struct}
Let $R$ with arity $k\angs{d}$ be produced by NRP $\Pi$. We say a tuple $(\theta^R, \theta^R_1,\dots,\theta^R_d)$ of \FOCrats terms \emph{represents $R$} if for all e-databases $D$ and $\tup a \in \adom(D)^k$:
\begin{enumerate}
    \item If $R(\tup a)\angs{\tup e} \in \Pi(D)$ then $\left(\sem{\theta^R}^{\calA_D}(\tup a), \sem{\theta^R_1}^{\calA_D}(\tup a), \dots, \sem{\theta^R_d}^{\calA_D}(\tup a)\right)=(1)\oplus \tup e$;
    \item If $R(\tup a) \not\in \Pi(D)$ then $\left(\sem{\theta^R}^{\calA_D}(\tup a), \sem{\theta^R_1}^{\calA_D}(\tup a), \dots,  \sem{\theta^R_d}^{\calA_D}(\tup a)\right)=(0)\oplus \tup 0^{d}$.
\end{enumerate}
\end{definition}

\begin{restatable}{lemma}{lemFNNinFOCrats}
    \label{lem:FNN in FOCrats}
    Let $\mathfrak{F}$ be a $\reluFFN$ with
    rational parameters, input dimension $p$ and output dimension $q$. Let $\theta'_1, \dots \theta'_p$ be \FOCrats terms over k variables. Then there exist \FOCrats terms $\theta_1, \dots, \theta_q$ such that for each weighted structure $\calA$ and $\tup a \in A^k$:
    \begin{align*}
        \mathfrak{F}(\sem{\theta'_1}^{\calA}(\tup a), \dots, \sem{\theta'_p}^{\calA}(\tup a)) = ( \sem{\theta_1}^{\calA}(\tup a), \dots, \sem{\theta_q}^{\calA}(\tup a))
    \end{align*}
\end{restatable}
\begin{proof}
    Multiplication and addition with 
    rational constants, as well as sums over a fixed number of terms, are available in \FOCrats. Further, $\relu(\sem{\theta}^\calA(\tup a))$ equals $\sem{\max(\theta,0)}^\calA(\tup a)$.
\end{proof}

We use $\bar{y}$ to denote a \emph{variable list}, i.e., a tuple of pairwise distinct variables. Given a head atom $R(\tup x)$ over variable tuple $\tup x = (x_1,\dots,x_k)$, we will represent $R$ with \FOCrats terms over $k$ distinct variables $\bar{y} = (y_1,\dots,y_k)$. Since the same variable can have multiple occurrences in NRP variable tuples, we multiply by the formula:
$\theta^{id}_{\tup x}(\bar{y}) := \displaystyle \prod_{1 \leq i<j \leq k \,:\, x_i=x_j} \mathds{1}_{y_i=y_j}$, to ensure that interpretations of $\bar{y}$ match interpretations of $\tup x$. 

\begin{restatable}{lemma}{lemRepresentTransformation}
    \label{lem:represent transformation}
    Let $\Psi$ be a transformation rule, with $\mu$ a 
    rational parameter $\reluFFN$:
    $$
    R(\tup x)\angs{\mu} \from R'(\tup x)
    $$
    If 
    $R'$ is represented in 
    \FOCrats, the same holds for $R$.
\end{restatable}    
\begin{proof}
    Let $(\theta^{R'}(\bar y), \theta^{R'}_1(\bar y), \dots, \theta^{R'}_{d'}(\bar y))$ represent $R'$. By Lemma~\ref{lem:FNN in FOCrats} there are terms $\theta^R_1,\dots,\theta^R_d$ giving the output of $\mu$. 
    Then     
    $(\theta^{R'}(\bar{y})\cdot \theta^{id}_{\tup x}(\bar{y}), \theta^{R'}(\bar{y})\cdot\theta^{R}_1(\bar{y}) \cdot \theta^{id}_{\tup x}(\bar{y}), \dots, \theta^{R'}(\bar{y})\cdot\theta^{R}_d(\bar{y}) \cdot \theta^{id}_{\tup x}(\bar{y}))$ represents $R$, where we multiply each term by $\theta^{R'}(\bar{y})$ to ensure all terms in the tuple have value $0$ when $\theta^{R'}(\bar{y})$ has interpretation $0$.
\end{proof}

\begin{restatable}{lemma}{lemRepresentConjunction}
    \label{lem:represent conjunction}
    Let $\Psi$ be a conjunction rule:
    \begin{align*}
        R(\tup x)\angs{\tsum} \fromprod R_1(\tup x_1) \land \dots \land R_m(\tup x_m)
    \end{align*}
    If 
    $R_1,\dots,R_m$ are represented in 
    $\FOCrats$, the same holds for $R$.    
\end{restatable}
\begin{proof}
Let each $R_i$ with arity $k_i\angs{d_i}$ be represented by $(\theta^{R_i}(\bar{y}_i),\theta^{R_i}_1(\bar{y}_i),\dots,\theta^{R_i}_{d_i}(\bar{y}_i))$. Let $\bar{z} := \bar{y}_1\oplus \dots \oplus \bar{y}_m$ be the list of all free variables of terms representing the body atoms, let $\tup w := \tup x \oplus \tup x_1 \oplus \dots \oplus \tup x_m$ and let $\bar{y}$ be a fresh variable list with $|\bar{y}| = |\tup x|$.
    Then
    \begin{align*}
        \theta^R(\bar{y}) & =
        \min\left(\sum(\bar{z}).\left(\prod_{1 \leq i \leq m} \theta^{R_i}(\bar{y}_i)\cdot \theta^{id}_{\tup w}(\bar{y},\bar{z})\right),1\right)
    \end{align*}
    Here, multiplication by the identity term ensures that we only count over interpretations to $\bar{y}, \bar{y}_1,\dots, \bar{y}_m$ that match the occurrences of the same variable in the tuple $\tup x \oplus \tup x_1 \dots \oplus \tup x_m$, and $\min(\theta,1)=-\max(-\theta,-1)$.
    For each index $j$ of the output embedding such that $j \leq d_i$ (where $d_i$ is the embedding dimension of $R_i$) for each $i=1 \dots m$:
    \begin{align*}
        \theta^R_j(\bar{y}) & = \sum(\bar{z}).\left(\prod_{1 \leq i \leq m} \theta^{R_i}_j(\bar{y}_i)\cdot \theta^{id}_{\tup w}(\bar{y},\bar{z})\right)
    \end{align*}    
    If $j > d_i$ for some $i$, the corresponding $\theta^{R_i}_j$ is substituted by $\theta^{R_i}$, so that if $\theta^{R_i}(\bar{y}_i)$ is $1$ the associated embedding doesn't impact the product, while if $\theta^{R_i}(\bar{y}_i)$ is $0$ the product is $0$. 
\end{proof}

\begin{restatable}{lemma}{lemRepresentDisjunction}
    \label{lem:represent disjunction}
    Let $\Psi$ be a disjunction rule:
    \begin{align*}
        R(\tup x)\angs{\tsum} \from R_1(\tup x) \lor \dots \lor R_m(\tup x)
    \end{align*}
    If 
    $R_1,\dots,R_m$ are represented in 
    $\FOCrats$, the same holds for $R$. 
\end{restatable}
\begin{proof}
    Again, let each $R_i$ with arity $k_i\angs{d_i}$ be represented by
    $(\theta^{R_i}(\bar{y}_i),\theta^{R_i}_1(\bar{y}_i),\dots,\theta^{R_i}_{d_i}(\bar{y}_i))$.
    Let $\bar{z} := \bar{y}_1\oplus \dots \oplus \bar{y}_m$ be the list of all free
    variables of terms representing the body atoms, let
    $\tup w := \tup x^{\oplus (m+1)}$ be the concatenation of $m+1$ copies of $\tup x$,
    and let $\bar{y}$ be a fresh variable list with $|\bar{y}| = |\tup x|$. Then
    \begin{align*}
        \theta^R(\bar{y}) & =
        \min\left(\sum(\bar{z}).\left(\left(\theta^{R_1}(\bar{y}_1) + \dots + \theta^{R_m}(\bar{y}_m)\right) \cdot \theta^{id}_{\tup w}(\bar{y},\bar{z})\right),1\right)
    \end{align*}
    and for each index $j$ of the output embedding:
    \begin{align*}
        \theta^R_j(\bar{y}) & =
\sum(\bar{z}).\left(\left(\theta^{R_1}_j(\bar{y}_1) + \dots + \theta^{R_m}_j(\bar{y}_m)\right) \cdot \theta^{id}_{\tup w}(\bar{y},\bar{z})\right)
    \end{align*}
    As before, multiplication by the identity term ensures that we only count over
    interpretations of $\bar{y},\bar{y}_1,\dots,\bar{y}_m$ that match the occurrences of
    the same variable in $\tup w$; since every body atom has the same variable tuple
    $\tup x$ as the head, this forces $\bar{y}_i$ to agree with $\bar{y}$ for each $i$.
\end{proof}

This yields one direction of Theorem~\ref{thm:NRP_and_FOC(1/2)_same_relations}: every flat query $\calQ$ computed by a simply-gated \reluFFN-based NRP is defined in \FOCrats.
\begin{proof}
Let $(\Pi,P)$ compute $\calQ$ with $\Pi = (\Psi_1, \ldots, \Psi_n)$ and $P: \reals^d \to \{0,1\}$ being a $d$-ary acceptance policy that maps $\tup e \in \reals^d$ to $0$ or $1$ with a Boolean combination of inequalities of the form $x_i>0$ with $i = 1,\dots,d$. Call the relation that $\Psi_n$ produces $R$.
Now, by definition, $\tup a \in \calQ(D)$ if and only if for some $\tup e \in \reals^d$, $R(\tup a)\angs{\tup e} \in \Pi(D)$ and $P(\tup e)=1$.

Also by definition, each EDB $S \in \scs$ with arity $k_i\angs{d_i}$ is represented by relations $(S, S_1, \dots, S_{d_i})$ over $\sigma$. Applying Lemmas~\ref{lem:represent transformation}, \ref{lem:represent conjunction} and~\ref{lem:represent disjunction} each IDB relation produced by a rule $\Psi_i$ is also represented in \FOCrats. Thus, if $\arity(R) = k\angs{d}$ there exists a tuple of $k$-ary terms $(\theta^R, \theta^R_1,\dots, \theta^R_d)$ that represents $R$. Let $\phi = \theta^R>0 \wedge \phi_P$, where $\phi_P$ is the Boolean combination of $P$ with $\theta^R_i>0$ substituted for $x_i>0$. Then $\phi$ defines $\calQ$.
\end{proof}

\subsection{Translation from \texorpdfstring{\FOCrats}{FO+C(1/2)} to NRP}

We show that for every \FOCrats term $\theta$, there exists an NRP that computes the interpretation of $\theta$ as embedding.

\begin{restatable}{lemma}{lemWeightedFOCratstoNRP}
\label{lem:FOCrats_terms_to_NRP}
Let $\theta$ be an \FOCrats term. There exists an NRP $\Pi$ with rational $\reluFFN$ transformations producing relation $R^\theta$ so that, if $\theta$ has $k\geq 0$ free variables, for each e-database $D$ and $\tup a \in \adom(D)^k$:
$$
R^\theta(\tup a)\angs{\sem{\theta}^{\calA_D}(\tup a)} \in \Pi(D)
$$
\end{restatable}

\begin{proof}
    We perform induction over term construction. Throughout, $\bar y,\bar y_1,\bar y_2,\bar z$ denote variable lists (tuples of pairwise distinct variables), and for a tuple $\tup x$ we write $var(\tup x)$ for the variable list of variables occurring in $\tup x$, in order of first occurrence.

    In the induction start, $\theta$ is of the form $q$ for some $q \in \rats$, or $R(\tup x)$, or $\mathds{1}_{x = y}$.
    If $\theta$ is a rational $q \in \rats$, $\Pi$ produces the $0\angs{1}$-ary relation $\text{True}_{q}$ (as described in Section~\ref{sec:nrp}).
    Suppose $R[k]\angs{d} \in \scs$ is represented over $\sigma$ (following Definition~\ref{def:representation_database_to_struct}) by $(R, R_1,\dots, R_{d})$ and $\theta := R(\tup x)$. Let $\bar y := var(\tup x)$ be $(y_1,\dots,y_{k'})$, where $k' \leq k$ is the number of free variables of $\theta$. We define $\Pi$ as follows:
    \begin{align*}
        \text{Zero}_{k}(\tup x) &\from \adom(x_1) \land \dots \land \adom(x_{k}) \land \text{True}_0()\\
        P(\tup x)\angs{\mu_1} &\from R(\tup x)\\
        Q(\tup x)\angs{\tsum} &\from P(\tup x) \lor \text{Zero}_{k}(\tup x)\\
        R^\theta(\bar y)\angs{\tsum} &\fromconcat Q(\tup x)
    \end{align*}
    Here $\text{True}_0()$ has embedding $(0)$, and $\mu_1$ produces the constant embedding $(1)$. 
    Note that $Q(\tup x)$ may contain repeated variables, so that $R^{\theta}(\bar y)$ has embedding $(1)$ exactly for those tuples $\tup a$ for which $R(\tup x[\bar y \mapsto \tup a])$ holds. Similarly, if $\theta := R_j(\tup x)$, encoding the $j$-th embedding value of database relation $R$, let $\Pi$ be:
    \begin{align*}
        \text{Zero}_{k}(\tup x) &\from \adom(x_1) \land \dots \land \adom(x_{k}) \land \text{True}_0()\\
        P(\tup x)\angs{\mu} &\from R(\tup x)\\
        Q(\tup x)\angs{\tsum} &\from P(\tup x) \lor \text{Zero}_{k}(\tup x)\\
        R^\theta(\bar{y})\angs{\tsum} & \fromconcat Q(\tup x)
    \end{align*}
    where now $\mu$ selects the $j$-th value of the embedding. If $\theta$ is $\mathds{1}_{x=y}$ with $x \neq y$, let $\Pi$ be:
    \begin{align*}
        \text{Zero}_2(x,y) &\from \adom(x) \land \adom(y) \land \text{True}_0()\\
        \mathit{Eq}(x,x)\angs{\tsum} &\from \adom(x) \land \text{True}_1()\\
        R^\theta(x,y)\angs{\tsum} &\from \mathit{Eq}(x,y) \lor \text{Zero}_2(x,y)
    \end{align*}
    If $x = y$, then $\mathds{1}_{x=x}$ is equivalent to the constant $1$ over $\adom(x)$, and $\Pi$ is $R^\theta(x)\angs{\tsum} \from \adom(x) \land \text{True}_1()$.

    For the induction step, suppose the embeddings of relations $R^{\theta_1}(\bar y_1), R^{\theta_2}(\bar y_2)$ produced by $\Pi$ are the interpretations of terms $\theta_1(\bar y_1), \theta_2(\bar y_2)$, where $\bar y_1,\bar y_2$ are the lists of free variables of $\theta_1,\theta_2$. 
    If $\theta(\bar y)$ is $\theta_1(\bar y_1) + \theta_2(\bar y_2)$, we add the following conjunction rule and transformation rule to $\Pi$, with $\bar y$ containing the variables in $\bar y_1$ and $\bar y_2$ without duplicates:
    \begin{align*}
        R'^{\theta}(\bar y)\fromconcat R^{\theta_1}(\bar y_1)\land R^{\theta_2}(\bar y_2) \qquad   R^{\theta}(\bar y)\angs{\mu} \from R'^{\theta}(\bar y)
    \end{align*}
    where $\mu$ sums the two indices of the embedding. If $\theta(\bar y)$ is $\theta_1(\bar y_1) \cdot \theta_2(\bar y_2)$ we add a single conjunction rule with product combination:
    \begin{align*}
        R^{\theta}(\bar y) \fromprod R^{\theta_1}(\bar y_1) \land R^{\theta_2}(\bar y_2)
    \end{align*}
    If $\theta(\bar y)$ is a summing term $\sum(\bar z).\theta_2(\bar y_2)$ with free variables $\bar y$ then we add the following conjunction rule:
    \begin{align*}
        R^{\theta}(\bar y)\angs{\tsum} \from R^{\theta_2}(\bar y_2) \wedge \bigwedge_{z \in \bar z} \adom(z)
    \end{align*}
    Finally, if $\theta(\bar y)$ is $\max(\theta_1(\bar y_1), \theta_2(\bar y_2))$ we add a conjunction and transformation rule:
    \begin{align*}
        R'^\theta(\bar y) \fromconcat R^{\theta_1}(\bar y_1) \land R^{\theta_2}(\bar y_2) \qquad
        R^\theta(\bar y)\angs{\mu} \from R'^\theta(\bar y)
    \end{align*}
    where $\mu((e_1,e_2))=\max(e_1,e_2)=e_1+\relu(e_2-e_1)$.
\end{proof}

Lemma~\ref{lem:FOCrats_terms_to_NRP} yields the second direction of Theorem~\ref{thm:NRP_and_FOC(1/2)_same_relations}: every flat query defined in \FOCrats is computed by a simply-gated \reluFFN-based NRP.

\begin{proof}
    Let $\calQ$ be a $k$-ary flat query defined by $\phi(\bar y)$, which is a combination of inequalities $\theta_i > \theta_j$ using negation and conjunction. Let $\theta_1,\dots,\theta_m$ be the terms in inequalities of $\phi$, each with free variables $\bar{y}_i$. Using Lemma~\ref{lem:FOCrats_terms_to_NRP} there exists for each term $\theta_i$ in $\phi$ a relation $R_i$ computed by a \reluFFN-based NRP $\Pi_i$ so that, with $k_i = |\bar{y}_i|$, for each e-database $D$ and $\tup a \in \adom(D)^{k_i}$, $R_i(\tup a)\angs{\sem{\theta_i}^{\calA_D}(\tup a)}\in \Pi_i(D)$. 

Let $\Pi$ be the union of these programs $\Pi_i$, together with the following conjunction and transformation rules:
\begin{align*}
    R'(\bar y)
    &\fromconcat
    R_1(\bar{y}_1) \land \cdots \land R_m(\bar{y}_m),
    \\
    R(\bar y)\angs{\mu}
    &\from
    R'(\bar y).
\end{align*}
Here, $\mu: \reals^m \to \reals^{m^2}$ maps
$\tup e = (e_1,\dots,e_m)$ to the vector containing all pairwise differences
$e_i - e_j$ for $1 \leq i,j \leq m$.
Let $P$ be the acceptance policy obtained from $\phi$ by replacing each inequality
$\theta_i > \theta_j$ with the corresponding value inequality
$e_i - e_j > 0$. Then $(\Pi,P)$ computes $\calQ$.
\end{proof}

\section{Proofs for Section~\ref{sec:FOC_TC0}}

\subsection{First-order Logic with Counting}
We recall the syntax and semantics of \FOC as defined in \cite{grohe2024descriptive}. \FOC has element variables $x,x',x_1$ and number variables $y,y',y_1$.
\paragraph{Syntax}
\[\begin{array}{llll}
        \text{Terms: }    &
        \theta            & ::= & 0 \mid 1 \mid y \mid \theta + \theta' \mid \theta \cdot \theta' \mid \#(x_1,\dots,x_k, y_1 < \theta_1,\dots, y_\ell < \theta_\ell).\phi\\
        \text{Formulas: } &
        \phi              & ::= & R(x_1,\dots,x_k) \mid x_1=x_2 \mid \neg \phi \mid \phi \wedge \psi \mid \theta \leq \theta'
    \end{array}\]
Here, in $\#(x_1,\dots,x_k, y_1 < \theta_1,\dots, y_\ell < \theta_\ell).\phi$ the bounding terms $\theta_i$ may not depend on $y_j$ with $j\geq i$.
\paragraph{Semantics}

\FOC is interpreted over Boolean (unweighted) $\sigma$-structures  $\calA$ together with the nonnegative integers as numerical domain $\mathbb{N}$. We define values $\sem{\phi}^{\calA}(\tup a, \tup b) \in \{0,1\}$ and $\sem{\theta}^{\calA}(\tup a, \tup b) \in \mathbb{N}$ for formulas $\phi$ and terms $\theta$, where $\tup a \in A^k$ and $\tup b \in \mathbb{N}^{k'}$ if the formula or term has $k$ free structure variables and $k'$ free number variables.
\begin{itemize}
    \item $\sem{0}^{\calA}$ =  $0$, $\sem{1}^{\calA}$ = $1$ and $\sem{y}^{\calA}(b)$  = $b$.

    \item $\sem{\theta + \theta'}^{\calA}(\tup a, \tup b)$ = $\sem{\theta}^{\calA}(\tup a, \tup b) + \sem{\theta'}^{\calA}(\tup a, \tup b)$.

    \item $\sem{\theta \cdot \theta'}^{\calA}(\tup a, \tup b)$ = $\sem{\theta}^{\calA}(\tup a, \tup b) \cdot \sem{\theta'}^{\calA}(\tup a, \tup b)$.

    \item $\sem{\#(x_1, \ldots, x_k, y_1<\theta_1,\dots,y_\ell<\theta_\ell).\phi}^{\calA}(\tup a,\tup b)$
    is the number of tuples:
$(\tup a', \tup b') \in A^k \times \nats^\ell$
such that 
$\sem{\phi}^{\calA}(\tup a, \tup a', \tup b, \tup b') = 1$ and for $i=1,\dots,\ell$, with $\tup b'_{<i}$ containing the first $i-1$ elements and $b'_i$ being the $i$-th element of $\tup b'$:~ $b'_i < \sem{\theta_i}^{\calA}(\tup a, \tup a', \tup b, \tup b'_{<i})$.
    \item $\sem{R}^\calA(\tup a) = R^\calA(\tup a)$.
    \item $\sem{x_1=x_2}^\calA(a_1,a_2)$ is $1$ if $a_1=a_2$ and $0$ otherwise.
    \item $\sem{\neg \phi}^\calA(\tup a,\tup b) = 1-\sem{\phi}^\calA(\tup a,\tup b)$, and $\sem{\phi \wedge \psi}^\calA(\tup a,\tup b) = \sem{\phi}^\calA(\tup a,\tup b)\cdot \sem{\psi}^\calA(\tup a,\tup b)$.
    \item $\sem{\theta \leq \theta'}^{\calA}(\tup a, \tup b)$ is $1$ if $\sem{\theta }^{\calA}(\tup a, \tup b) \leq \sem{\theta'}^{\calA}(\tup a, \tup b)$ and $0$ otherwise.    
\end{itemize}

\subsection{Proof of Theorems \ref{thm:FOC_is_FOCrats_is_TCO} and \ref{thm:FOCrats_contained_in_TCO}}

An ordered (weighted) $\sigma$-structure is a $\sigma \cup \{\leq\}$-structure $\calA$ where $\leq^\calA$ is a linear order on the domain of $\calA$. Over ordered structures, number variables can be represented by element variables. We give an explicit proof.
\begin{lemma}
    \label{lem:FOC without number variables on ordered structures}
    Every \FOC-formula without free number variables is equivalent, over ordered Boolean structures, to a \FOC-formula without number variables.
\end{lemma}
\begin{proof}
        For a $\FOC$-formula $\phi$ without free number variables and a univariate
    polynomial $p$, we will say that $\phi(\tup x)$ is
    \emph{uniformly $p$-bounded} if for every Boolean structure $\calA$ and $\tup a \in A^{|\tup x|}$, all numerical variables are bounded by $p(|A|)$. As a first step, we will show that every \FOC-formula without free number variables is equivalent, over Boolean structures, to a uniformly $p$-bounded formula. Let $\phi$ be a \FOC-formula without free number variables.
    We may assume without loss of generality that each number variable is only bound by one $\#$ binder (by renaming variables as needed). We can then associate
    to each number variable $y$ a unique corresponding
    bounding term $\theta_{y}$. Consider the
    directed graph whose nodes are the number variables and with an edge $y\to y'$ if $y'$ occurs in the bounding term $\theta_{y}$. It follows from the definition of \FOC that this graph has no cycles.

    We can thus associate a rank to each variable, where
    the leaves of the graph have the lowest rank.
    We now compute by induction on the rank, for each number variable $y$ a polynomial $p_{\theta_y}(z)$ so that $p_{\theta_y}(|A|)$ bounds the interpretations of $\theta_y$ for each structure $\calA$.   
    \begin{itemize}
        \item $p_0(z)=0$,
        \item $p_1(z)=1$,
        \item $p_{y}(z)=p_{\theta_y}(z)$
        \item
              $p_{\theta+\theta'}(z)=p_\theta(z)+ p_{\theta'}(z)$
              and $p_{\theta\cdot\theta'}(z)=p_\theta(z)\cdot p_{\theta'}(z)$
        \item $p_{\#(x_1,\ldots, x_k,y_1< \theta_1, \ldots, y_\ell<\theta_\ell).\phi}(z)=z^k\cdot p_{\theta_1}(z)\cdot \ldots \cdot p_{\theta_\ell}(z)$.
    \end{itemize}
    Let $p$ be a polynomial that dominates $p_{\theta_y}$ for each
    number variable $y$ occurring in $\phi$.
    Let $(\phi)^p$ be obtained from $\phi$ as follows, by induction, where `$\operatorname{ord}$' is an abbreviation of $\#x.x=x$, representing object domain size, and $\theta < \theta'$ abbreviates $\theta+1 \leq \theta'$:
    \begin{itemize}
        \item $(\psi)^p = \psi$ for all atomic formulas and terms,
        \item $(\cdot)^p$ commutes with all operators except $\#$.
        \item $(\#(x_1,\ldots, x_k,y_1< \theta_1, \ldots, y_\ell<\theta_\ell).\psi)^p= \#(x_1,\ldots, x_k,y_1<p(\ord), \ldots, y_\ell<p(\ord)).(y_1< (\theta_1)^p\land \cdots\land y_\ell<(\theta_\ell)^p \land (\psi)^p) $
    \end{itemize}
    Then $\phi$ is equivalent to $(\phi)^p$ over all Boolean structures, and $(\phi)^p$ is uniformly $p$-bounded.

    Now, we prove the main statement.
    Let $\phi$ be an \FOC-formula without free number variables. By the above, we may assume that $\phi$ is uniformly $p$-bounded for some univariate polynomial $p$.
    Let us first restrict attention to structures with domain size at least two. We can then
    assume without loss of generality that $p(x)=x^N$ for some $N\in\mathbb{N}$. We simulate number variables by length-$N$ tuples of element variables. Formally, for each number variable
    $y_i$, let $\tup x^i=x^i_1, \ldots, x^i_{N}$ be
    corresponding fresh element variables. We write
    $\textbf{x}^i<_{\mathrm{lex}}\textbf{x}^j$ for
    \[\bigvee_{k=1,\ldots,N} \Big(x^i_{k}<x^j_k\land \bigwedge_{k'<k} x^i_{k'}=x^j_{k'}\Big) \]
    We now inductively translate $\phi$ to an equivalent $(\phi)^\dagger$ without number variables as follows, where $\tup x'$ is another $N$-tuple of fresh element variables:
    \begin{itemize}
        \item $(y_i)^\dag = \#(\tup x').(\tup x'<_{\mathrm{lex}} \tup x^i)$,
        \item $(\chi)^\dag = \chi$ for all other atomic formulas/terms,
        \item $(\cdot)^\dag$ commutes with all operators other than $\#$,
        \item $\big(\#(x_1, \ldots, x_n, y_{i_1}<p(\ord), \ldots, y_{i_\ell}<p(\ord)).\psi\big)^\dag=\#(x_1, \ldots, x_n,\tup x^{i_1},\dots,\tup x^{i_\ell}).(\psi)^\dag$
    \end{itemize}
    This argument assumed that we restrict attention to structures of domain size at least two. Over structures of domain size 1, it is trivial that $\phi$ is equivalent to a formula without number variables, since there are only finitely many isomorphism types of such structures.
    Furthermore, we can test with a \FOC-formula whether the domain size of a structure is at least two.    
\end{proof}

\begin{corollary}
\label{cor:FOC_to_FOCrats_formulas}
Every $\FOC$ formula is equivalent over ordered Boolean structures to an $\FOCrats$ formula.
\end{corollary}
\begin{proof}
Lemma~\ref{lem:FOC without number variables on ordered structures} states that every formula or term in \FOC is equivalent to a formula or term without number variables on ordered structures. We show by induction on formula construction that for every formula or term  $\chi$ in \FOC without number variables there exists an equivalent term $(\chi)^\dagger$ in \FOCrats. The cases where $\chi$ is a relational atom, equality, $0$ or $1$ are immediate. Further:
    \begin{itemize}
        \item $(\cdot)^\dagger$ commutes with addition and multiplication;
        \item $(\neg \phi)^\dagger = 1-(\phi)^\dagger$;
        \item $(\phi \land \psi)^\dagger = (\phi)^\dagger \cdot (\psi)^\dagger$;
        \item $(\theta \leq \theta')^\dagger = \min(\max((\theta')^\dagger+1-(\theta)^\dagger,0),1)$;
        \item $(\#(x_1,\dots,x_n).\phi)^\dagger = \sum(x_1,\dots,x_n).(\phi)^\dagger$
    \end{itemize}
    where $\min(\theta,\theta') := -\max(-\theta,-\theta')$. Every $\FOC$ formula $\phi$ is then equivalent over ordered Boolean structures to the $\FOCrats$ formula $(\phi)^\dagger>0$.
\end{proof}

\begin{lemma}
\label{lem:FOCrats_to_FOC_formulas}
Every \FOCrats formula is equivalent over Boolean structures to an \FOC formula.
\end{lemma}
\begin{proof}
    We prove by simultaneous induction on formulas and terms that, over Boolean structures,
    \begin{enumerate}
        \item For every \FOCrats formula $\phi$ there is an equivalent \FOC formula $(\phi)^\dagger$;
        \item For every \FOCrats term $\theta$ there is a triple $(\theta)^\dagger = (f(\theta),t(\theta),r(\theta))$ where 
        $f(\theta)$ is an \FOC formula, $t(\theta)$ is an \FOC term and $r(\theta)$ is a non-negative rational so that: 
        $$\sem{\theta} = (-1)^{1+\sem{f(\theta)}} \cdot \sem{t(\theta)}  \cdot r(\theta)$$
    \end{enumerate}
    We use $\mathds{1}_\phi$ to abbreviate $\#().\phi$, $\phi \leftrightarrow \psi$ to abbreviate $(\phi \wedge \psi) \vee (\neg \phi \wedge \neg \psi)$ and truncated subtraction $\theta -_t \theta'$ to abbreviate $\#(y < \theta).(\theta'\leq y)$.
    \begin{itemize}
        \item For $q \in \rats$, if $q \geq 0$ then $(q)^\dagger = (\top, 1, q)$, else $(q)^\dagger = (\bot, 1, -q)$
        \item $(\mathds{1}_{x=x'})^\dagger = (\top, \mathds{1}_{x=x'}, 1)$;
        \item $(R(x_1,\dots,x_k))^\dagger = (\top, \mathds{1}_{R(x_1,\dots,x_k)}, 1)$;
        \item $(\cdot)^\dagger$ commutes with negation and conjunction;
        \item $(\theta_1 \cdot \theta_2)^\dagger = (f(\theta_1)\leftrightarrow f(\theta_2), t(\theta_1)\cdot t(\theta_2), r(\theta_1)\cdot r(\theta_2))$.

    \end{itemize}
    The remaining cases are addition, max, iterated sum and inequality. Given $\FOCrats$ terms $\theta_1,\theta_2$, let $f(\theta_1)=\phi_1$, $f(\theta_2)=\phi_2$ and let $p_1,p_2,q_1,q_2 \in \mathbb{N}$ so that $r(\theta_1)=\frac{p_1}{{q_1}}$, $r(\theta_2)=\frac{p_2}{{q_2}}$. Let $p'_1 = p_1\cdot q_2$, and $p'_2=p_2\cdot q_1$. 
    Then:
    \begin{align*}
        \sem{\theta_1} &= (-1)^{1+\sem{\phi_1}} \cdot \sem{t(\theta_1)} \cdot \frac{p'_1}{q_1\cdot q_2}\\
        \sem{\theta_2} &= (-1)^{1+\sem{\phi_2}} \cdot \sem{t(\theta_2)} \cdot \frac{p'_2}{q_1\cdot q_2}
    \end{align*}
\begin{itemize}
    \item 
    If $\theta$ is $\theta_1+\theta_2$, then:
    \begin{align*}
        f(\theta) &= \Big(\phi_1 \wedge \Big((p'_1 \cdot t(\theta_1)) \geq (p'_2 \cdot t(\theta_2))\Big)\Big) \lor \Big(\phi_2 \wedge \Big((p'_2 \cdot t(\theta_2)) \geq (p'_1 \cdot t(\theta_1))\Big)\Big)\\
     t(\theta) &= \mathds{1}_{\phi_1 \leftrightarrow \phi_2} \cdot \big(p'_1 \cdot t(\theta_1) + p'_2 \cdot t(\theta_2)\big)\\
&\quad + \mathds{1}_{\neg(\phi_1 \leftrightarrow \phi_2)} \cdot \Big(\big(p'_1 \cdot t(\theta_1) -_t p'_2 \cdot t(\theta_2)\big) + \big(p'_2 \cdot t(\theta_2) -_t p'_1 \cdot t(\theta_1)\big)\Big)\\
        r(\theta) &= \frac{1}{q_1\cdot q_2}
    \end{align*}
\item
    If $\theta$ is $\max(\theta_1,\theta_2)$, then:
    \begin{align*}
        f(\theta) &= \phi_1 \lor \phi_2\\
        t(\theta) &= \Big(\mathds{1}_{(\phi_1 \wedge \neg \phi_2)} \cdot p'_1 \cdot t(\theta_1) \Big) + \Big(\mathds{1}_{(\phi_2 \wedge \neg \phi_1)} \cdot p'_2 \cdot t(\theta_2)\Big)\\
        & \quad + \Big(\mathds{1}_{\phi_1 \wedge \phi_2} \cdot \Big(\mathds{1}_{p'_1 \cdot t(\theta_1) \geq p'_2 \cdot t(\theta_2)} \cdot p'_1 \cdot t(\theta_1) + \mathds{1}_{p'_1 \cdot t(\theta_1) < p'_2 \cdot t(\theta_2)} \cdot p'_2 \cdot t(\theta_2)\Big)\Big)\\
        & \quad+ \Big(\mathds{1}_{\neg \phi_1 \wedge \neg \phi_2} \cdot \Big(\mathds{1}_{p'_1 \cdot t(\theta_1) \geq p'_2 \cdot t(\theta_2)} \cdot p'_2 \cdot t(\theta_2) + \mathds{1}_{p'_1 \cdot t(\theta_1) < p'_2 \cdot t(\theta_2)} \cdot p'_1 \cdot t(\theta_1)\Big)\Big)\\
        r(\theta) &= \frac{1}{q_1\cdot q_2}        
    \end{align*}
\item
    If $\theta$ is $\sum(x_1,\dots,x_n).\theta_1$, then:
    \begin{align*}
         f(\theta) &= \Big(\#(x_1,\dots,x_n,y< t(\theta_1)).\phi_1\Big)\geq \Big(\#(x_1,\dots,x_n,y< t(\theta_1)).(\neg \phi_1)\Big)\\
         t(\theta) &= \bigg(\Big(\#(x_1,\dots,x_n,y< t(\theta_1)).\phi_1\Big) -_t \Big(\#(x_1,\dots,x_n,y< t(\theta_1)).(\neg \phi_1)\Big)\bigg)\\
         &\quad + \bigg(\Big(\#(x_1,\dots,x_n,y< t(\theta_1)).(\neg \phi_1)\Big) -_t 
         \Big(\#(x_1,\dots,x_n,y< t(\theta_1)).\phi_1\Big)\bigg)\\
         r(\theta) &= r(\theta_1)
     \end{align*}
     \item
    Finally, if $\phi$ is $\theta_1 > \theta_2$, then:
    \begin{align*}
        (\phi)^\dagger =& \Big(\phi_1 \wedge \phi_2 \wedge p'_1\cdot t(\theta_1)>p'_2\cdot t(\theta_2)\Big) \lor \Big(\neg\phi_1 \wedge \neg\phi_2 \wedge p'_1 \cdot t(\theta_1)<p'_2\cdot t(\theta_2)\Big)\\
        &\, \lor
        \Big(\phi_1 \wedge \neg \phi_2 \wedge (p'_1\cdot t(\theta_1)>0 \lor p'_2\cdot t(\theta_2)>0)\Big) 
    \end{align*}
    \end{itemize}
\end{proof}

\thmFOCratsisFOCunweightedordered*
\begin{proof}
    The equivalence between (2) and (3) is well known \cite{barrington1990uniformity}. The equivalence between (1) and (2) follows from Corollary~\ref{cor:FOC_to_FOCrats_formulas} and Lemma~\ref{lem:FOCrats_to_FOC_formulas}.
\end{proof}

\thmFOCratsContainedInTCO*
\begin{proof}
    Given $m \geq 1$, represent a 
    rational $\frac{p}{q} \in \rats$ with $p \in \mathbb{Z}, q \in \mathbb{N}$ as a binary string of length $2m$ by concatenating the first $m$ bits of the binary representations of $p$ and of $q$. We assume that 
    $\frac{p}{q}$ is reduced. For an ordered rational weighted structure $\calA$ and $\tup a \in A^k$ for some $k \geq 0$, let $s(\calA, \tup a) \in \{0,1\}^*$ be as follows. Start with $1^{|A|}01^m0$ where $2m$ is the size of the largest exact bit representation of a rational weight in $\calA$. Then the weights of $\calA$ follow as bitstrings of size $2m$ in order of the signature and the domain $A$. The final $k\cdot \log(|A|)$ bits encode $\tup a$, yielding a string of length:
    \begin{align*}
        ||s(\calA, \tup a)|| = O(|A| + \sum_{R \in \sigma} \sum_{\tup a \in A^{\arity(R)}} 2m)
    \end{align*}
    The inclusion follows from known results on arithmetic in uniform \TCO. The representation of addition, multiplication and iterated addition over integers in non-uniform \TCO goes back to \cite{chandra1984constant}. It is not hard to see that the arguments can be adapted to uniform \TCO (see \cite{barrington1990uniformity}). Iterated addition over rationals requires iterated multiplication over integers, since one multiplies the denominators of summed rationals. It was shown in \cite{hesse2002uniform} that this is also available in uniform \TCO.

    For strictness, we give a separating example over structures with integer weights. Let $a \in \calQ(\calA)$ if and only if $\sem{R}^\calA(a)$ is an even integer.
    To compute this in \TCO, let $n=||s(\calA,a)||$ for rational weighted $\calA$ and $a \in A$. The circuit $C_n$ selects the weight $\frac{p_i}{q_i}$, where $i$ is the position of $a$ in the order, represented by the last $\log(|A|)$ bits. Then verify that $q_i$ is $1$, and the last bit of $p_i$ is $0$. Now consider structures $\calA_r$ consisting of a single element $a$ and a weight $r$ for $R(a)$. Any \FOCrats formula that expresses $\calQ$ is equivalent to a boolean combination over finitely many inequality checks $\theta_i > \theta_j$. Substitute each such check by the equivalent $\theta_i - \theta_j >0$. Since the truth value alternates infinitely for $r \in \{1, 2, 3, \dots\}$ at least one of these inequality checks must alternate infinitely many times. However, each $\theta_i-\theta_j$ defines a univariate piecewise polynomial with finitely many pieces in $r$, since $r$ is the only weight in $\calA_r$ and $\theta_i$ is constructed from $r$ and constants with addition, multiplication an max. Since each contributing polynomial has a number of roots bounded by its degree, $\theta_i-\theta_j$ cannot alternate between $>0$ and $\leq 0$ infinitely many times. Hence $\calQ$ is not defined in \FOCrats.
\end{proof}

\end{document}